\pdfoutput=1

	\documentclass[a4paper,final,11pt,american,DIV=12,headinclude,twoside,BCOR=1.2cm]{scrartcl}
	\def\version{arxiv}

\usepackage{ifthen}
\newcommand\ifarxiv[2][]{\ifthenelse{\equal{\version}{arxiv}}{#2}{#1}}

\usepackage{etex}
\usepackage[utf8]{inputenc}
\usepackage[T1]{fontenc}

\usepackage{babel}
\input{ushyphex.tex} %

\ifarxiv{%
	\usepackage{lmodern} 
	\usepackage{natbib}
	\setlength\parindent{1.5em}
	\usepackage{colonequals}
	\usepackage[headinclude,headsepline]{scrpage2}
	\pagestyle{scrheadings}
	\clearscrheadfoot
	\automark[section]{section}
	\ohead{\pagemark}
	\ihead{\headmark}
	\addtokomafont{caption}{\sffamily\small}
	\addtokomafont{captionlabel}{\sffamily\textbf}
	\setcapmargin{2em}
}

\usepackage{tocbibind} %

\usepackage{microtype}

\usepackage{xspace}
\usepackage{enumitem}
\usepackage{array}
\usepackage{booktabs}
\usepackage{dcolumn}
\usepackage{multicol}
\usepackage{multirow}
\usepackage[referable]{threeparttablex}
\usepackage{colortbl}
\usepackage[TABBOTCAP]{subfigure}
\usepackage{rotating}

\usepackage{wref}

\usepackage{xfrac}
\usepackage{amsmath}
\usepackage{amssymb}
\usepackage{dsfont}
\usepackage{mathtools}

\usepackage{tikz}
\usepackage{ifthen}
\usepackage{placeins}
\usepackage{needspace}

\ifarxiv{%
	\newcommand*\llbracket{\mathopen{[\mkern-2.5mu[}}
	\newcommand*\rrbracket{\mathclose{]\mkern-2.5mu]}}
	\usepackage[hyphens]{url}
}
\usepackage{bm}
\let\sim\thicksim

\usetikzlibrary{positioning,decorations.pathreplacing,external,calc,spy,plotmarks}
\usepackage{pgfplots}
\ifarxiv{%
	\pgfplotsset{compat=1.5}
}

\pgfplotscreateplotcyclelist{validationplots}{%
very thick,black!50\\%
semithick,black,dash pattern=on 5pt off 1.8pt\\%
very thin,densely dotted,black,mark=*,every mark/.append style={solid,thin,fill=black!50}\\%
very thin,densely dotted,black,mark=asterisk,every mark/.append style={solid,thin}\\%
very thin,densely dotted,black,mark=square,every mark/.append style={solid,thin,fill=white}\\%
very thin,densely dotted,mark=triangle*,black,every mark/.append style={solid,thin,fill=white}\\%
}
\pgfplotscreateplotcyclelist{validationplotsCM}{%
very thick,black!50\\%
semithick,black,dash pattern=on 5pt off 1.8pt\\%
very thick,black,dash pattern=on 1.5pt off 1pt\\%
very thin,densely dotted,black,mark=*,every mark/.append style={solid,thin,fill=black!50}\\%
very thin,densely dotted,black,mark=asterisk,every mark/.append style={solid,thin}\\%
very thin,densely dotted,black,mark=square,every mark/.append style={solid,thin,fill=white}\\%
very thin,densely dotted,mark=triangle*,black,every mark/.append style={solid,thin,fill=white}\\%
very thin,densely dotted,mark size=2pt,mark=otimes,black,every mark/.append style={solid,very thin}\\%
}

\pgfdeclarelayer{foreground}
\pgfdeclarelayer{background}
\pgfsetlayers{background,main,foreground}

\usepackage{clrscode3e}

\def\EndFor{\End\li\kw{end for} }
\def\EndIf{\End\li\kw{end if} }
\def\EndWhile{\End\li\kw{end while} }

\usepackage{float}
\floatstyle{ruled}
\newfloat{algorithm}{tbp}{loa}
\floatname{algorithm}{Algorithm}

\usepackage[
			final,
			unicode=true, 
            bookmarks=true,
            bookmarksnumbered=true,
            bookmarksopen=true,
            breaklinks=true,
            pdfborder={0 0 0}, 
            colorlinks=false   
]{hyperref}

\hypersetup{
	pdftitle={Analysis of Pivot Sampling in Dual Pivot Quicksort - A Holistic Analysis of Yaroslavskiy's Partitioning Scheme}, 
	pdfauthor={Markus E. Nebel, Sebastian Wild and Conrado Martínez}, 
	pdfsubject={},
	pdfkeywords={Quicksort, dual-pivot, Yaroslavskiy's partitioning method, %
		median of three, average-case analysis, I/O operations, external-memory model}%
}

\usepackage[amsmath,hyperref,thmmarks]{ntheorem}

\theorembodyfont{\slshape}
\theoremseparator{:}
\makeatletter
\newtheoremstyle{proofstyle}%
  {\item[\theorem@headerfont\hskip\labelsep ##1\theorem@separator]}%
  {\item[\theorem@headerfont\hskip\labelsep ##1 of ##3\theorem@separator]}
\makeatother

\newtheorem{theorem}{Theorem}[section]

\theoremstyle{plain}
\newtheorem{proposition}[theorem]{Proposition}
\newtheorem{lemma}[theorem]{Lemma}

\theoremstyle{plain}
\theorembodyfont{\upshape}

\theoremsymbol{\raisebox{-.25ex}{$\Box$}}
\qedsymbol{\raisebox{-.25ex}{$\Box$}}

\theoremstyle{proofstyle}
\newtheorem{proof}{Proof}

\vfuzz \hfuzz	

\allowdisplaybreaks[3]

\newcommand\weakemph[1]{\textsl{#1}}
 
\newdimen\makeboxdimen
\newcommand\makeboxlike[3][l]{%
\setbox0=\hbox{#2}%
\global\makeboxdimen=\wd0%
\setbox1=\hbox{\makebox[\makeboxdimen][#1]{%
\makebox[0pt][#1]{#3}%
}}%
\ht1=\ht0%
\dp1=\dp0%
\box1%
}

\newcommand\plaincenter[1]{%
	\mbox{}\hfill#1\hfill\mbox{}%
}

\newcounter{inlineenum}

\newcommand*\ie{\mbox{{i.\hspace{.2ex}e.}}}
\newcommand*\eg{\mbox{{e.\hspace{.2ex}g.}}}

\newcommand\E{\mathop{\mbox{$\mathbb{E}$}}\nolimits}
\newcommand\given{\mathbin{\mid}}

\newcommand\R{\mathds R}
\newcommand\N{\mathds N}

\newcommand\Oh{O}
\def\.{\mskip1mu}

\newcommand{\Prob}{\ensuremath{\mathbb{P}}}

\ifarxiv{%
	\newcommand\ui[2]{#1^{\smash{(}#2\smash{)}}}
}
\newcommand{\vect}[1]{\boldsymbol{\mathbf{#1}}}

\newcommand\eqdist{	
	\mathchoice{
		\mathrel{\overset{\raisebox{0ex}{$\scriptstyle \cal D$}}=}%
	}{
		\mathrel{\like{=}{%
			\overset{\raisebox{-1ex}{$\scriptscriptstyle \cal D$}}=%
		}}%
	}{
		\mathrel{\overset{\cal D}=}%
	}{
		\mathrel{\overset{\cal D}=}%
	}%
}

\newcommand\ppe{\phantom{=}}

\newcommand\like[3][c]{\makeboxlike[#1]{\ensuremath{#2}}{\ensuremath{#3}}}

\newcommand\uniform{\mathcal U}
\newcommand\bernoulli{\mathrm B}
\newcommand\hypergeometric{\mathrm{HypG}}
\newcommand\multinomial{\mathrm{Mult}}
\newcommand\binomial{\mathrm{Bin}}
\newcommand\dirichlet{\mathrm{Dir}}

\newcommand\BetaFun{\mathrm B}

\newcommand\harm[1]{\ensuremath{H_{#1}}}

\ifarxiv{%
	\newcommand\ce{\colonequals}
}

\newcommand\rel[1]{\mathrel{\:{#1}\:}}
\newcommand\wrel[1]{\mathrel{\;{#1}\;}}
\newcommand\wwrel[1]{\mathrel{\;\;{#1}\;\;}}
\newcommand\bin[1]{\mathbin{\:{#1}\:}}
\newcommand\wbin[1]{\mathbin{\;{#1}\;}}

\newcommand{\eqwithref}[2][c]{%
	\relwithref[#1]{#2}{=}%
}
\newcommand{\relwithref}[3][c]{%
	\mathrel{\underset{\mathclap{\makebox[\widthof{$=$}][#1]{\scriptsize\wref{#2}}}}{#3}}%
}

\newcommand\toll[2][]{%
	\ensuremath{%
	\ifthenelse{\equal{#1}{}}{%
		T_{\!#2}%
	}{%
		T_{\!#2}({#1})%
	}}%
}

\newcommand\istoll[2][]{%
	\ensuremath{%
	\ifthenelse{\equal{#1}{}}{%
		\iscost_{\!#2}%
	}{%
		\iscost_{\!#2}({#1})%
	}}%
}

\newcommand\insertsortcost{W}
\newcommand\iscost{\insertsortcost}

\newcommand\bytecodes{\mathit{BC}}
\newcommand\scans{\mathit{SE}}

\newcommand\values[1]{#1}

\newcommand\positionsets[1]{\mathcal{#1}}

\newcommand\numberat[2]{\values{#1}\mbox{\emph{@}}\.\positionsets{#2}}
\newcommand\satK{\numberat sK}
\newcommand\latK{\numberat lK}
\newcommand\satG{\numberat sG}

\newcommand\indicatorsymbol{\mathds 1}
\newcommand\indicator[1]{\indicatorsymbol_{\{#1\}}}

\newcommand\discreteEntropy[1][]{%
	\ensuremath{\mathchoice{
		{\mathcal{H}} \ifthenelse{\equal{#1}{}}{ }{ \mkern-1mu(#1) }
	}{
		{\mathcal{H}} \ifthenelse{\equal{#1}{}}{ }{ \mkern-1mu(#1) }
	}{
		{\mathcal{H}} \ifthenelse{\equal{#1}{}}{ }{ (#1) }
	}{
		{\mathcal{H}} \ifthenelse{\equal{#1}{}}{ }{ (#1) }
	}}\xspace%
}
\newcommand\contentropy[1][]{%
	\ensuremath{\mathchoice{
		{\mathcal{H}^*} \ifthenelse{\equal{#1}{}}{ }{ \mkern-1mu(#1) }
	}{
		{\mathcal{H}^*} \ifthenelse{\equal{#1}{}}{ }{ \mkern-1mu(#1) }
	}{
		{\mathcal{H}}^* \ifthenelse{\equal{#1}{}}{ }{ (#1) }
	}{
		{\mathcal{H}}^* \ifthenelse{\equal{#1}{}}{ }{ (#1) }
}}\xspace%
}

\newcommand\arrayA{%
	\ensuremath{\mathchoice{
		\smash{\raisebox{-.2pt}{\scalebox{1.25}[1.18]{$\mathtt{A}$}}}%
	}{
		\smash{\raisebox{-.2pt}{\scalebox{1.25}[1.18]{$\mathtt{A}$}}}%
	}{
		\smash{\raisebox{-.2pt}{\scalebox{1.25}[1.18]{$\scriptstyle\mathtt{A}$}}}%
	}{
		\smash{\raisebox{-.2pt}{\scalebox{1.25}[1.18]{$\scriptscriptstyle\mathtt{A}$}}}%
	}}\xspace%
}

\usepackage{manfnt}

\newcommand*{\YQS}{\mathrm{YQS}}
\newcommand*{\CQS}{\mathrm{CQS}}

\newcommand*\generalYarostM{%
	\ensuremath{\sss{\YQS}{\vect t}{\isthreshold}}\xspace%
}
\newcommand*\isthreshold{\ensuremath{\mathnormal{w}}\xspace}

\newcommand\sss[3]{\like[l]{#1}{#1_{#2}}^{#3}}
\newcommand\sssh[3]{\like[l]{#1}{#1_{#2}^{}}^{#3}}

\newcommand*\pprime{\prime\mkern-1mu\prime}

\let\oldparagraph\paragraph
\makeatletter%
\renewcommand\paragraph{%
    \@ifstar{\myparagraphStar}{\myparagraphNoStar}%
}
\makeatother%
\newcommand\myparagraphStar[1]{%
	\oldparagraph*{#1.}%
}
\newcommand\myparagraphNoStar[2][]{%
	\ifthenelse{\equal{#1}{}}{%
		\oldparagraph[#2]{#2.}%
	}{%
		\oldparagraph[#1]{#2.}%
	}%
}

\colorlet{symmetriccolor}{black!10}

\ifarxiv{%
	\bibliographystyle{abbrev-plainnat}
}

\title{%
	Analysis of Pivot Sampling in Dual-Pivot Quicksort\strut%
	\thanks{%
		This work has been partially supported by funds from the Spanish Ministry for 
		Economy and Competitiveness (MINECO) and the European Union (FEDER funds) 
		under grant COMMAS (ref. \texttt{TIN2013-46181-C2-1-R}).\protect\\
		A preliminary version of this article was presented at 
		AofA 2014~\citep{NebelWild2014}.%
	}
}
\subtitle{A Holistic Analysis of Yaroslavskiy's Partitioning Scheme}

\ifarxiv{%
	\author{Sebastian Wild${}^\dagger$
		\and Markus E. Nebel\thanks{%
			Computer Science Department,
			University of Kaiserslautern,
			Germany\protect\\
			\texttt{\{wild,nebel\}\,@cs.uni-kl.de}
		}
		\and Conrado Martínez\thanks{%
			Department of Computer Science,
			Univ.\ Polit\`ecnica de Catalunya,
			\texttt{conrado@cs.upc.edu}
		}
	}
}

\begin{document}  
\maketitle

\begin{abstract}
	\noindent 
	The new dual-pivot Quicksort by Vladimir Yaroslavskiy\,---\,used in Oracle's
	Java runtime library since version 7\,---\,features intriguing asymmetries. 
	They make a basic variant of this algorithm use
	less comparisons than classic single-pivot Quicksort.
	In this paper, we extend the analysis to the case where the two pivots are
	chosen as fixed order statistics of a random sample. 
	Surprisingly, dual-pivot Quicksort then needs \emph{more} comparisons
	than a corresponding version of classic Quicksort, 
	so it is clear that counting comparisons is not sufficient to explain the running
	time advantages observed for Yaroslavskiy's algorithm in practice.
	Consequently, we take a more holistic approach 
	and give also the precise leading term of the average number of swaps, the 
	number of executed Java Bytecode instructions and 
	the number of scanned elements, a new simple cost measure that approximates I/O costs
	in the memory hierarchy.
	We determine optimal order statistics for each of the cost measures.
	It turns out that the asymmetries in
	Yaroslavskiy's algorithm render pivots with a systematic skew more
	efficient than the symmetric choice.
	Moreover, we finally have a convincing explanation for the success of
	Yaroslavskiy's algorithm in practice:
	Compared with corresponding versions of classic single-pivot Quicksort,
	dual-pivot Quicksort needs significantly less I/Os, both with and without
	pivot sampling.
	\ifarxiv{%
	}
\end{abstract}

\section{Introduction}

Quicksort is one of the most efficient comparison-based sorting algorithms and
is thus widely used in practice, for example in the sort implementations
of the C++ standard library and Oracle's Java runtime library.
Almost all practical implementations are based on the highly tuned version of
\citet{Bentley1993}, often equipped with the strategy of
\citet{Musser1997} to avoid quadratic worst-case behavior.
The Java runtime environment was no exception to this\,---\,up to version~6.
With version~7 released in~2009, Oracle broke with this tradition and
replaced its tried and tested implementation by a dual-pivot Quicksort
with a new partitioning method proposed by Vladimir Yaroslavskiy.

The decision was based on extensive running time experiments that clearly
favored the new algorithm.
This was particularly remarkable as earlier analyzed dual-pivot
variants had not shown any potential for performance gains over classic
single-pivot Quicksort \citep{Sedgewick1975,hennequin1991analyse}.
However, we could show for pivots from fixed array positions (\ie\ no
sampling) that Yaroslavskiy's asymmetric partitioning method
beats classic Quicksort in the comparison model:
asymptotically $1.9\,n\ln n$ vs.\ $2\,n\ln n$ comparisons on
average \citep{Wild2012}.
	It is an interesting question how far one
	can get by exploiting asymmetries in this way.
	For dual-pivot Quicksort with an \emph{arbitrary} partitioning method,
	\citet{Aumuller2013icalp} establish a lower bound of asymptotically
	$1.8 \, n\ln n$ comparisons
	and they also propose a partitioning method that attains this bound
	by dynamically switching the order of comparisons depending on the
	current subproblem.

The savings in comparisons are opposed by a large increase in the number of swaps,
so the competition between classic Quicksort and Yaroslavskiy's Quicksort remained open.
To settle it, we compared Java implementations of both
variants and found that Yaroslavskiy's method executes \emph{more} Java
Bytecode instructions on average \citep{Wild2013Quicksort}.
A possible explanation why it still shows better running times was recently
given by \citet{Kushagra2014}:
Yaroslavskiy's algorithm in total needs fewer passes over
the array than classic Quicksort, and is thus more efficient in the
\textsl{external-memory model}.
(We rederive and extend their results in this article.)

Our analyses cited above ignore a very effective strategy in Quicksort:
for decades, practical implementations choose their pivots as \emph{median of a
random sample} of the input to be more efficient (both in terms of average
performance and in making worst cases less likely).
Oracle's Java~7 implementation also employs this optimization:
it chooses its two pivots as the \emph{tertiles of five} sample elements.
This equidistant choice is a plausible generalization, since
selecting the median as pivot is known to be optimal for classic Quicksort
\citep{Sedgewick1975,Martinez2001}.

However, the classic partitioning methods treat elements smaller
and larger than the pivot in symmetric ways\,---\,unlike Yaroslavskiy's
partitioning algorithm:
depending on how elements relate to the two pivots, one of
\emph{five} different execution paths is taken in the partitioning loop, and
these can have highly different costs!
How often each of these five paths is taken depends on the \emph{ranks} of the
two pivots, which we can push in a certain direction by selecting \emph{skewed} order
statistics of a sample instead of the tertiles.
The partitioning costs alone are then minimized if the cheapest execution path
is taken all the time.
This however leads to very unbalanced distributions of sizes
for the recursive calls, such that a \emph{trade-off} between partitioning costs
and balance of subproblem sizes has to be found.

We have demonstrated experimentally that there is potential to tune dual-pivot
Quicksort using skewed pivots \citep{Wild2013Alenex}, but only considered a
small part of the parameter space.
It will be the purpose of this paper to identify the optimal way to
sample pivots by means of a precise analysis of the resulting overall costs,
and to validate (and extend) the empirical findings that way.

There are scenarios where, even for the symmetric, classic Quicksort, a
skewed pivot can yield benefits over median of~$k$ \citep{Martinez2001,kaligosi2006branch}.
An important difference to Yaroslavskiy's algorithm is, however, that the
situation remains symmetric:
a relative pivot rank $\alpha < \frac12$ has the same effect as
one with rank~$1-\alpha$.

Furthermore, it turns out that dual-pivot Quicksort needs more comparisons 
than classic Quicksort, if both choose their pivots from a sample 
(of the same size), but the running time advantages of Yaroslavskiy's algorithm
remain, so key comparisons do not dominate running time in practice.
As a consequence, we consider other cost measures like the number of executed Bytecode
instructions and I/O operations.

\subsection{Cost Measures for Sorting}
\label{sec:cost-measures}

As outlined above, we started our attempt to explain the success of Yaroslavskiy's
algorithm by counting comparisons and swaps, 
as it is classically done for the evaluation of
sorting strategies.
Since the results were not conclusive, 
we switched to primitive instructions and determined the expected number of 
\textsl{Java Bytecodes} as well as the number of operations executed by Knuth's 
\textsl{MMIX} computer (see \citep{Wild2012thesis}), 
comparing the different Quicksort variants on this basis.
To our surprise, Yaroslavskiy's algorithm is not superior in terms of primitive instructions, 
either.

At this point we were convinced that features of modern computers like memory hierarchies
and/or pipelined execution must be responsible for the speedup empirically observed for 
the new dual-pivot Quicksort.
The memory access pattern of partitioning in Quicksort is essentially like for a sequential scan,
only that several scans with separate index variables are interleaved:
two indices that alternatingly run towards each other in classic Quicksort, 
the three indices $k$, $g$ and $\ell$ in Yaroslavskiy's Quicksort 
(see \wref{sec:yaroslavskiys-partitioning-method}) or
even four indices in the three-pivot Quicksort of \citet{Kushagra2014}.
We claim that a good cost measure is
the \textsl{total distance covered by all scanning indices}, which 
we call the number of \emph{``scanned elements''}
(where the number of visited elements is used as the unit of ``distance'').

As we will show, this cost measure is rather easy to analyze,
but it might seem artificial at first sight.
It is however closely related to the number of cache misses in practice
(see \wref{sec:validation-scans-vs-cache-misses}) and the
number of I/O operations in the \textsl{external-memory model:}
For large inputs in external memory, 
one has to assume that each block of elements of the input array is
responsible for one I/O when it is accessed for the first time in a
partitioning run.
No spatial locality between accesses through different
scanning indices can be assumed, so
memory accesses of one index will not save (many) I/Os for another index.
Finally, accesses from different partitioning runs lack temporal locality,
so (most) elements accessed in previous partitioning runs will have been 
removed from internal memory before recursively sorting subarrays.
Therefore, the number of I/Os is very close to the number of scanned elements,
when the blocks contain just single array elements.
	This is in fact not far from reality for the caches close to the CPU:
	the L1 and L2 caches in the AMD Opteron architecture, for example, use 
	block sizes of 64 bytes, which on a 64-bit computer means that only 8 
	array entries fit in one block
	\citep{Hennessy2006}.

The external-memory model is an idealized view itself.
Actual hardware has a hierarchy of caches with different characteristics, 
and for caches near the CPU, only very simple addressing and replacement strategies 
yield acceptable access delays.
From that perspective, we now have three layers of abstraction:
Scanned elements are an approximation of I/O operations of the external-memory model 
(for scanning-based algorithms like Quicksort), which in turn
are an approximation of memory hierarchy delays like cache misses.

The theoretical cost measure ``scanned elements'' has been used
implicitly in earlier analyses of the caching behavior of Quicksort
and other scanning-based algorithms like, \eg, Mergesort 
\citep{LaMarca1999,Kushagra2014},
even though it has (to our knowledge) never been made explicit;
it was merely used as an intermediate step of the analysis.
In particular, 
\citeauthor{Kushagra2014} essentially compute the number of scanned
elements for different Quicksort variants for the case of random pivots (\ie, no sampling),
and find that Yaroslavskiy's algorithm outperforms classic Quicksort in this cost measure.

Besides the memory hierarchy, the effects of pipelined execution might 
be an explanation for the speedup observed for the new algorithm.
However, the numbers of branch misses (a.\,k.\,a.\, pipeline stalls) incurred 
by classic Quicksort and Yaroslavskiy's Quicksort 
do not differ significantly 
under simple branch predictions schemes \citep{MartinezNebelWild2015},
so pipelining is not a convincing explanation.

\medskip\noindent
The rest of this article is organized as follows:
After listing some general notation,
\wref{sec:generalized-yaroslavskiy-quicksort} introduces the subject of study:
Yaroslavskiy's algorithm.
\wref{sec:results} collects the main analytical results of this paper, the proof of
which is given in
\wref[Sections]{sec:distributional-analysis} and~\ref{sec:average-case-analysis}.
Mathematical arguments in the main text are kept concise, 
but the interested reader is provided with details in the appendices.
In \wref{sec:validation}, we compare the analytical result with experimental data
for practical input sizes.
The algorithmic consequences of our analysis are discussed in
\wref{sec:discussion} in detail.
\wref{sec:conclusion} concludes the paper.

\section{Notation and Preliminaries}
\label{sec:notation}

We write vectors in bold font, for example
$\vect t=(t_1,t_2,t_3)$.
For concise notation, we use expressions like $\vect t + 1$ to mean
\emph{element-wise} application, \ie, $\vect t + 1 = (t_1+1,t_2+1,t_3+1)$.
By $\dirichlet(\vect\alpha)$, we denote a random variable with \textsl{Dirichlet
distribution} and shape parameter
$\vect\alpha = (\alpha_1,\ldots,\alpha_d) \in \R_{>0}^d$.
Likewise for parameters $n\in\N$ and $\vect p =
(p_1,\ldots,p_d) \in [0,1]^d$ with $p_1+\cdots+p_d=1$, we write
$\multinomial(n,\vect p)$ for a random variable with \textsl{multinomial
distribution} with $n$ trials.
$\hypergeometric(k,r,n)$ is a random variable with
\textsl{hypergeometric distribution}, \ie, the number of red balls when
drawing $k$ times without replacement from an urn of $n\in\N$ balls,
$r$ of which are red, (where $k,r\in\{1,\ldots,n\}$).
Finally, $\uniform(a,b)$ is a random variable uniformly distributed in the
interval $(a,b)$, and $\bernoulli(p)$ is a Bernoulli variable with probability
$p$ to be $1$.
We use ``$\eqdist$'' to denote equality in distribution.

As usual for the average case analysis of sorting algorithms, we assume the
\textsl{random permutation model}, \ie, all elements are
different and every ordering of them is equally likely.
The input is given as an array $\arrayA$ of length $n$ and we denote the initial
entries of $\arrayA$ by $U_1,\ldots,U_n$.
We further assume that $U_1,\ldots,U_n$ are
i.\,i.\,d.\ uniformly $\uniform(0,1)$ distributed;
as their ordering forms a random permutation \citep{mahmoud2000sorting}, this
assumption is without loss of generality.
Some further notation specific to our analysis is introduced below; for
reference, we summarize all notations used in this paper in \wref{app:notations}.

\section{Generalized Yaroslavskiy Quicksort}
\label{sec:generalized-yaroslavskiy-quicksort}

In this section, we review Yaroslavskiy's partitioning method and combine it
with the pivot sampling optimization to obtain what we call the
\textsl{Generalized Yaroslavskiy Quicksort} algorithm.
We give a full-detail implementation of the algorithm,
because \emph{preservation of randomness} is somewhat tricky to
achieve in presence of pivot sampling, but vital for precise analysis.
The code we give here can be fully analyzed, 
but is admittedly not suitable for productive use;
it should rather be considered as a mathematical \emph{model} for practical 
implementations, which often do \emph{not} preserve randomness 
(see, \eg, the discussion of Java 7's implementation below).

\subsection{Generalized Pivot Sampling}
\label{sec:general-pivot-sampling}

Our pivot selection process is declaratively specified as
follows, where $\vect t = (t_1,t_2,t_3) \in \N^3$ is a fixed parameter:
Choose a random sample $\vect V = (V_1,\ldots,V_k)$ of size 
$k = k(\vect t) \ce t_1+t_2+t_3+2$ from the elements
and denote by
$(V_{(1)},\ldots,V_{(k)})$ the \emph{sorted} sample,
\ie,
$
	V_{(1)} \le V_{(2)} \le \cdots \le V_{(k)}
$.
	(In case of equal elements any possible ordering will do;
	in this paper, we assume distinct elements.)
Then choose the two pivots $P \ce V_{(t_1+1)}$ and $Q\ce V_{(t_1+t_2+2)}$ such
that they divide the sorted sample into three regions of respective sizes $t_1$,
$t_2$ and~$t_3$:
\begin{equation*}
	\underbrace{V_{(1)} \ldots V_{(t_{1})}}
			_{t_{1}\,\mathrm{elements}}
	\wrel\le
	\underbrace{V_{(t_{1}+1)}} _ {=P}
	\wrel\le
	\underbrace{V_{(t_{1}+2)} \ldots V_{(t_{1}+t_{2}+1)}}
			_{t_{2}\,\mathrm{elements}}
	\wrel\le
	\underbrace{V_{(t_{1}+t_{2}+2)}} _ {=Q}
	\wrel\le
	\underbrace{V_{(t_{1}+t_{2}+3)} \ldots V_{(k)}}
			_{t_{3}\,\mathrm{elements}}
	.
\end{equation*}
The parameter choice $\vect t = (0,0,0)$ corresponds to the case 
without sampling.
Note that by definition, $P$ is the small(er) pivot and $Q$ is the
large(r) one.
We refer to the $k-2$ elements of the sample that are not chosen as
pivots as \emph{``sampled-out''};
$P$~and $Q$ are the chosen \emph{pivots}. 
All other elements\,---\,those which have not been part of the sample\,---\,are
referred to as \emph{ordinary} elements. 

We assume that the sample size $k$ does not depend on the size $n$ of the 
current (sub)problem for several reasons:
First of all, 
such strategies are not very practical because they complicate code.
Furthermore, if the sample size grows recognizably with $n$,
they need a sorting method for the samples that is
efficient also when samples get large.
If, on the other hand, $k$ grows very slowly with $n$,
the sample is essentially constant for practical input sizes.

Analytically, any growing sample size $k=k(n)=\omega(1)$ immediately provides 
asymptotically \emph{precise} order statistics (\textsl{law of large numbers}) 
and thus allows an optimal choice of the pivots.
As a consequence, the leading term of costs is the \emph{same} for all such 
sample sizes and only the linear term of costs is affected 
(as long as $k = \Oh(n^{1-\epsilon})$), see \citet{Martinez2001}.
This would make it impossible to distinguish pivot selection strategies
by looking at leading-term asymptotics.

	Note that with $k=\Oh(1)$, we hide the cost of selecting 
	order statistics in the second order term, so
	our leading-term asymptotics ignores the costs of sorting the sample 
	in the end.
	However, it is a fixed constant whose contribution we can still roughly estimate
	(as validated in \wref{sec:validation}).
	Also, we retain the possibility of letting $k\to\infty$ analytically 
	(see \wref{sec:continuous-ranks}).

\subsection{Yaroslavskiy's Dual-Pivot Partitioning Method}
\label{sec:yaroslavskiys-partitioning-method}

\begin{algorithm}
	\vspace{-1ex}
	\def\pind{\id{i_p}}
	\def\qind{\id{i_q}}
	\begin{codebox}
		\Procname{$\proc{PartitionYaroslavskiy}\,(\arrayA,\id{left},\id{right},P,Q)$}
		\zi \Comment Assumes $\id{left} \le \id{right}$. 
		\zi \Comment Rearranges \arrayA s.\,t.\ with return value $(\pind,\qind)$
				holds \smash{ $\begin{cases}
	 				\arrayA[j] < P,			& \text{for }\id{left} \le j \le \pind; \\
	 				P \le \arrayA[j] \le Q,	& \text{for }\like[l]{\id{left}}{\pind} < j < \qind; \\
	 				\arrayA[j] \ge Q,		& \text{for }\like[l]{\id{left}}{\qind} \le j \le \id{right} .
				\end{cases}$}
				\rule[-2.5ex]{0pt}{1ex}
		\li $\ell\gets \id{left}$; 
		 	\quad $g\gets \id{right}$; 
		 	\quad $k\gets \ell$ \label{lin:yaroslavskiy-init-l-g-k} 
		\li	\While $k\le g$  \label{lin:yarosavskiy-outer-loop-branch}
		\li	\Do
				\If $\arrayA[k] < P$ \label{lin:yaroslavskiy-comp-1}
		\li		\Then
					Swap $\arrayA[k]$ and $\arrayA[\ell]$ \label{lin:yaroslavskiy-swap-1}
		\li			$\ell\gets \ell+1$ \label{lin:yaroslavskiy-l++-1}
		\li		\Else 
		\li			\If $\arrayA[k] \ge Q$ \label{lin:yaroslavskiy-comp-2}
		\li			\Then
						\While $\arrayA[g] > Q$ and $k<g$ \label{lin:yaroslavskiy-comp-3} 
		\li				\Do 
							$g\gets g-1$ 
						\EndWhile
		\li				\If $\arrayA[g] \ge P$ \label{lin:yaroslavskiy-comp-4}
		\li				\Then
							Swap $\arrayA[k]$ and $\arrayA[g]$ \label{lin:yaroslavskiy-swap-2}
		\li				\Else
		\li					Swap $\arrayA[k]$ and $\arrayA[g]$ \label{lin:yaroslavskiy-swap-3a}
		\li					Swap $\arrayA[k]$ and $\arrayA[\ell]$ \label{lin:yaroslavskiy-swap-3b}
		\li					$\ell\gets \ell+1$ \label{lin:yaroslavskiy-l++-2}
						\EndIf
		\li				$g\gets g-1$ \label{lin:yaroslavskiy-g--}
					\EndIf
				\EndIf
		\li		$k\gets k+1$ \label{lin:yaroslavskiy-k++}
			\EndWhile \label{lin:yaroslavskiy-end-while}
		\li	\Return $(\ell-1, g+1)$ \label{lin:yaroslavskiy-return}
		\zi
	\end{codebox}
	\vspace{-4ex}
	\caption{\strut%
		Yaroslavskiy's dual-pivot partitioning algorithm.
	}
\label{alg:partition}
\end{algorithm}

\noindent
Yaroslavskiy's partitioning method is given in \wref{alg:partition}.
In bird's-eye view, it consists of two
indices, $k$ and $g$, that start at the left resp.\ right end of \arrayA and
scan the array until they meet. 
Elements left of $k$ are smaller or equal than $Q$, elements right of $g$ are
larger.
Additionally, a third index $\ell$ lags behind $k$ and separates elements
smaller than $P$ from those between both
pivots.
Graphically speaking, this invariant of the algorithm is given in \wref{sfig:invariant}.

\begin{figure}
	\plaincenter{\subfigure[%
		Invariant of \protect\wref{alg:partition} during partitioning.%
		\label{sfig:invariant}
	]{%
		\begin{tikzpicture}[
			yscale=0.5, xscale=0.6,
			baseline=(ref.south),
			every node/.style={font={}},
			semithick,
		]	

		\fill[black!2] (-1.5,0) rectangle ++(16,1) ;
		\draw[dotted,black!50] 
				(-1.5,0) -- (0,0) 
				(-1.5,1) -- (0,1)
				(13,0) -- (14.5,0)
				(13,1) -- (14.5,1)
		;
		\draw[fill=black!5] (0,0) rectangle ++(13,1) ;
		
		\node at (0.3,-0.4) {\small$\id{left}$} ;
		\node at (12.7,-0.4) {\small$\id{right}$} ;
		
		\node at (1.5,0.5) {$< P$};
		\draw (3,1) -- ++ (0,-1);
		\node at (3.3,-0.4) {$\ell$};
		
		\node at (12,0.5) {$\ge Q$};
		\draw (11,1) -- ++ (0,-1);
		\node at (10.7,-0.4) {\strut$g$};
		
		\node at (5,0.5) {$P\le \circ\le Q$};
		\draw (7,1) -- ++(0,-1);
		\node at (7.3,-0.4) {$k$};

		\node[below] at (10.7,-0.6) {$\leftarrow$};
		\node[below] at (3.3,-0.6) {$\rightarrow$};
		\node[below] at (7.3,-0.6) {$\rightarrow$};

		\node[inner sep=0pt] (ref) at (9,0.5) {?};
		\end{tikzpicture}%
	}}\\[2ex]
	\plaincenter{\subfigure[%
		State after the partitioning loop has been left.%
		\label{sfig:state-after-partitioning}
	]{%
		\begin{tikzpicture}[
			yscale=0.5, xscale=0.6,
			baseline=(ref.south),
			every node/.style={font={}},
			semithick,
		]	

		\fill[black!2] (-1.5,0) rectangle ++(16,1) ;
		\draw[dotted,black!50] 
				(-1.5,0) -- (0,0) 
				(-1.5,1) -- (0,1)
				(13,0) -- (14.5,0)
				(13,1) -- (14.5,1)
		;
		\draw[fill=black!5] (0,0) rectangle ++(13,1) ;
		
		\node at (0.3,-0.4) {\small$\id{left}$} ;
		\node at (12.7,-0.4) {\small$\id{right}$} ;
		
		\node at (2,0.5) {$< P$};
		\draw (4,1) -- ++ (0,-1);
		\node at (4.3,-0.4) {$\ell$};
		
		\node at (13-1.75,0.5) {$\ge Q$};
		\draw (9.5,1) -- ++ (0,-1);
		\node at (9.2,-0.4) {\strut$g$};
		\node at (9.8,-0.4) {$k$};
		
		\node at (4+.5*5.5,0.5) {$P\le \circ\le Q$};
		
		\begin{scope}[-,very thick,draw=black!50,decoration=brace]
			\draw[decorate] (4,-1.8) -- 
				node[below=.5ex] {$\positionsets{L}$} ++(-4,0) ;
			\draw[decorate] (13,-1.8) -- 
				node[below=.5ex] {$\positionsets{G}$} ++(-13+9.5,0) ;
			\draw[decorate,decoration={brace,aspect=.4}] (9.6,-1.2) -- 
				node[pos=.4,below=.5ex] {$\positionsets{K}$} ++(-9.6,0) ;
		\end{scope}
		\end{tikzpicture}%
	}}%
	\caption{%
		The state of the array \protect\arrayA during and after partitioning.
		Note that the last values attained by $k$, $g$ and $\ell$ are \emph{not}
		used to access the array \protect\arrayA, so the positions of the indices after
		partitioning are by definition not contained in the corresponding position sets.
	}
	\label{fig:invariant}
\end{figure}
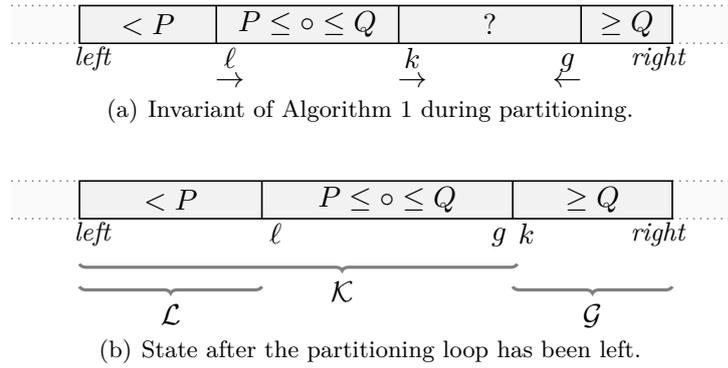

When partitioning is finished, $k$ and $g$ have met and thus $\ell$ and $g$
divide the array into three ranges; 
precisely speaking, in \wref{lin:yaroslavskiy-return} of \wref{alg:partition}
the array has the shape shown in \wref{sfig:state-after-partitioning}.

We write~$\positionsets{K}$, $\positionsets{G}$ and~$\positionsets{L}$ 
for the sets of all indices
that $k$, $g$ resp.~$\ell$ attain in the course of the partitioning 
process\,---\,more precisely: $\positionsets{K}$ is the set of all values
attained by variable $k$, for which we access the array via $\arrayA[k]$;
similarly for $\positionsets{G}$ and $\positionsets{L}$.
(We need a precise definition for the analysis later.%
\footnote{%
	Note that the meaning of $\positionsets{L}$ is 
	different in our previous work \citep{Wild2013Quicksort}:
	therein $\positionsets{L}$ includes the last value 
	index variable~$\ell$ attains which is never used to access the array.
	The authors consider the new definition clearer and
	therefore decided to change it.%
})
As the indices move sequentially these sets are in fact (integer) intervals,
as indicated in \wref{sfig:state-after-partitioning}.

Moreover, we call an element \emph{small}, \emph{medium}, or \emph{large} if
it is smaller than~$P$, between $P$ and $Q$, or larger than~$Q$, respectively.
The following properties of the algorithm are needed for the analysis, 
(see \citet{Wild2012,Wild2013Quicksort} for details):
\begin{enumerate}[label=(Y\arabic*),leftmargin=3em,itemsep=1ex]
	\item \label{prop:in-K-first-with-p} 
		Elements $U_i$ with $i \in \positionsets{K}$ are first compared with $P$ (\wref{lin:yaroslavskiy-comp-1}).
		Only if $U_i$ is not small, it is also compared to $Q$ (\wref{lin:yaroslavskiy-comp-2}).
	\item \label{prop:in-G-first-with-q} 
		Elements $U_i$ with $i \in \positionsets{G}$ are first compared with~$Q$ (\wref{lin:yaroslavskiy-comp-3}).
		If they are not large, they are also compared to~$P$ (\wref{lin:yaroslavskiy-comp-4}).
	\item \label{prop:small-eventually-swappel-left} 
		Every small element $U_i<P$ eventually causes one swap to put it behind $\ell$ 
		(at \wref{lin:yaroslavskiy-swap-1} if $i\in\positionsets{K}$ resp.\ 
		 at \wref{lin:yaroslavskiy-swap-3b} if $i\in\positionsets{G}$).
	\item \label{prop:latK-smatG-swapped-in-pairs}
		The large elements located in $\positionsets{K}$ and the non-large
		elements in $\positionsets{G}$ are always swapped in pairs
		(\wref{lin:yaroslavskiy-swap-2} resp.\ \wref{lin:yaroslavskiy-swap-3a}).
\end{enumerate}
For the number of comparisons we will (among other quantities) need to
count the large elements $U_i > Q$ with $i \in \positionsets{K}$, cf.\ \ref{prop:in-K-first-with-p}.
We abbreviate their number by ``$\latK$''.
Similarly, $\satK$ and $\satG$ denote the number of small elements in $k$'s
resp.\ $g$'s range.

\subsection{Implementing Generalized Pivot Sampling}
\label{sec:generalized-pivot-sampling-implementation}

While extensive literature on the analysis of (single-pivot) Quicksort with
pivot sampling is available, most works do not specify the pivot
selection process in detail.
	(Noteworthy exceptions are \citeauthor{Sedgewick1977}'s seminal works
	which give detailed code for the median-of-three strategy
	\citep{Sedgewick1975,Sedgewick1978} and \citeauthor{Bentley1993}'s
	influential paper on engineering a practical sorting
	method~\citep{Bentley1993}.)%
The usual justification is that, in any case, we only draw pivots a
\emph{linear} number of times and from a constant-size sample. 
So the costs of pivot selection are negligible for the
leading-term asymptotic, 
and hence also the precise way of how selection is done is not important.

There is one caveat in the argumentation: 
Analyses of Quicksort usually rely on setting up a recurrence equation of
expected costs that is then solved (precisely or asymptotically).
This in turn requires the algorithm to \emph{preserve} the distribution of
input permutations for the subproblems subjected to recursive
calls\,---\,otherwise the recurrence does not hold.
Most partitioning algorithms, including the one of Yaroslavskiy, have the
desirable property to preserve randomness \citep{Wild2012}; but this is not
sufficient! 
We also have to make sure that the main procedure of Quicksort does not 
alter the distribution of inputs for recursive calls;
in connection with elaborate pivot sampling algorithms, this is harder to
achieve than it might seem at first sight.

For these reasons, the authors felt the urge to include a minute discussion
of how to implement the generalized pivot sampling scheme of
\wref{sec:general-pivot-sampling} in such a way that the recurrence equation
remains precise.
We have to address the following questions:

\oldparagraph{Which elements to choose for the sample?}

In theory, a \emph{random} sample produces the most reliable results and
also protects against worst case inputs.
The use of a random pivot for classic Quicksort has been considered right from
its invention \citep{Hoare1961} and is suggested as a general strategy to deal
with biased data \citep{Sedgewick1978}.

However, all programming libraries known to the authors actually avoid the
additional effort of drawing random samples. 
They use a set of deterministically selected positions of the array, instead;
chosen to give reasonable results for common special cases like almost sorted arrays.
For example, the positions used in Oracle's Java~7 implementation are depicted
in \wref{fig:sample-choice-jre7}.

For our analysis, the input consists of i.\,i.\,d.\ random variables, so
\emph{all} subsets (of a certain size) have the same distribution.
We might hence select the positions of sample elements such that they
are convenient for our (analysis) purposes.
For reasons elaborated in \wref{sec:randomness-preservation} below, we have to 
\emph{exclude} sampled-out elements from partitioning to keep analysis feasible,
and therefore, our implementation uses the $t_1+t_2+1$ leftmost
and the $t_3+1$ rightmost elements of the array as sample, as
illustrated in \wref{fig:sample-choice-generalized-yaroslavskiy}.
Then, partitioning can simply be restricted to the range between the two parts
of the sample, namely positions $t_1+t_2+2$ through $n-t_3-1$
(cf.\ \wref{lin:generalized-partition-call} of \wref{alg:generalized-yaroslavskiy}).

\begin{figure}
	\plaincenter{%
	\begin{tikzpicture}
	[scale=0.5,baseline=2,every node/.style={font={\footnotesize}}]
		\draw (0,0) rectangle (3,1) ; 
		\draw[decoration=brace,decorate] 
		  		(0,1.25) -- node[above] {\scriptsize$\frac3{14}n$} (3,1.25) ; 
		\foreach \x in {0,3,6,9} {
		  \draw[fill=black!10,very thick] (3+\x,0) rectangle (4+\x,1) ; 
		  \draw (4+\x,0) rectangle (6+\x,1) ;
		  \draw[decoration=brace,decorate] 
		  		(4+\x,1.25) -- node[above]{\scriptsize$\frac17n$} (6+\x,1.25) ; 
		}
		\draw[fill=black!10,very thick] (15,0) rectangle (16,1) ;
	    \draw (16,0) rectangle (19,1) ;
	    \draw[decoration=brace,decorate] 
		  		(16,1.25) -- node[above]{\scriptsize$\frac3{14}n$} (19,1.25) ;
	    \node at (6.5,-0.5) {$P$};
	    \node at (12.5,-0.5) {$Q$};
		\foreach \x in {1,...,5} {
			\node at (\x*3+0.5,0.5) {$V_{\x}$} ;
		}
	\end{tikzpicture}
	}
	\caption{%
		The five sample elements in Oracle's Java~7 implementation of
		Yaroslavskiy's dual-pivot Quicksort are chosen such that
		their distances are approximately as given above. 
	}
	\label{fig:sample-choice-jre7}
\end{figure}
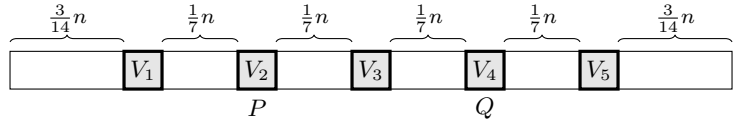

\begin{figure}[tbhp]
	\plaincenter{%
	\begin{tikzpicture}[scale=0.5,every node/.style={font=\footnotesize}]
		\begin{scope}[very thick,fill=black!10]
			\filldraw (0,0) rectangle (6,1);
			\filldraw (15,0) rectangle (20,1); 
		\end{scope}
		\begin{scope}[thick]
			\draw (3,0) -- (3,1);
			\draw (4,0) -- (4,1);
			\draw (16,0) -- (16,1); 
		\end{scope}
	
		\draw[thin] (0,0) grid (20,1) ;
		\foreach \x in {1,...,20}  \node at (\x-0.5,1.5) {\scriptsize\x} ;
	
		\begin{scope}[thick,decoration=brace, yshift=-5pt]
			\draw[decorate] (3,0) -- node[below=.5ex]{$t_1$} (0,0) ;
			\draw[decorate] (6,0) -- node[below=.5ex]{$t_2$} (4,0) ;
			\draw[decorate] (20,0) -- node[below=.5ex]{$t_3$} (16,0) ;
		\end{scope}
		\node at (3.5,-0.5) {$P$};
		\node at (15.5,-0.5) {$Q$};
		
		\foreach \x in {1,...,6} {
			\node at (\x-0.5,0.5) {$V_{\x}$} ; 
		}
		\foreach \x in {7,...,11} {
			\node at (16-7+\x-0.5,0.5) {$V_{\x}$} ; 
		}
		
	\end{tikzpicture}
	}
	\caption{%
		Location of the sample in our implementation of \generalYarostM
		with $\vect t = (3,2,4)$.
		Only the non-shaded region $\protect\arrayA[7..15]$ is subject to partitioning. 
	}
	\label{fig:sample-choice-generalized-yaroslavskiy}
\end{figure}
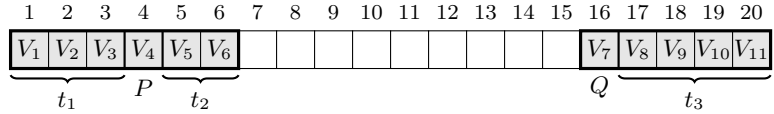

\oldparagraph{How do we select the desired order statistics from the sample?}
Finding a given order statistic of a list of elements is known as the
\weakemph{selection problem} and can be solved by specialized algorithms like
Quickselect.
Even though these selection algorithms are superior by far on large lists,
selecting pivots from a reasonably small sample is most efficiently done by
fully sorting the whole sample with an elementary sorting method.
Once the sample has been sorted, we find the pivots in $\arrayA[t_1+1]$ and
$\arrayA[n-t_3]$, respectively.

We will use an Insertionsort variant for sorting samples.
Note that the implementation has to “jump” across the gap
between the left part and the right part of the sample.
\wpref{alg:samplesort-left} and its symmetric cousin \wref{alg:samplesort-right}
do that by internally ignoring the gap in index variables and then correct for that
whenever the array is actually accessed.

\begin{figure}
	\plaincenter{%
	\begin{tikzpicture}[scale=0.5,every node/.style={font=\footnotesize}]
		\begin{scope}[very thick,fill=black!10]
			\filldraw (0,0) rectangle (6,1);
			\filldraw (15,0) rectangle (20,1); 
		\end{scope}
		\begin{scope}[thick]
			\draw (3,0) -- (3,1);
			\draw (4,0) -- (4,1);
			\draw (16,0) -- (16,1); 
		\end{scope}
	
		\draw[thin] (0,0) grid (20,1) ;

		\begin{scope}[thick,decoration=brace, yshift=5pt]
			\draw[decorate] (0,1) -- node[above=.5ex]{$t_1$} (3,1) ;
			\draw[decorate] (4,1) -- node[above=.5ex]{$t_2$} (6,1) ;
			\draw[decorate] (16,1) -- node[above=.5ex]{$t_3$} (20,1) ;
		\end{scope}
		\node at (3.5,0.5) {$P$};
		\node at (15.5,0.5) {$Q$};
		
		\foreach \x in {1,...,3,7,8} 			\node at (\x-0.5,0.5) {$s$} ; 
		\foreach \x in {5,6, 9,10,11} 			\node at (\x-0.5,0.5) {$m$} ;
		\foreach \x in {12,...,15,17,18,19,20}	\node at (\x-0.5,0.5) {$l$} ;
		
		\draw[thick] ( 8,-0.1) -- ( 8,1.1) ;
		\draw[thick] (11,-0.1) -- (11,1.1) ;
		
		\begin{scope}[yshift=-4cm]
			\begin{scope}[fill=black!10]
				\fill (0,0) rectangle (3,1);
				\fill (5,0) rectangle (8,1);
				\fill (11,0) rectangle ++(1,1);
				\fill (16,0) rectangle (20,1); 
			\end{scope}
			\begin{scope}[thick]
				\draw (5,0) rectangle ++(1,1);
				\draw (11,0) rectangle ++(1,1);
			\end{scope}
		
			\draw[thin] (0,0) grid (20,1) ;

			\node at (6-0.5,0.5) {$P$};
			\node at (12-0.5,0.5) {$Q$};
			
			\foreach \x in {1,...,5} 	\node at (\x-0.5,0.5) {$s$} ; 
			\foreach \x in {7,...,11}	\node at (\x-0.5,0.5) {$m$} ;
			\foreach \x in {13,...,20}	\node at (\x-0.5,0.5) {$l$} ;
			
			\begin{scope}[thick,decoration=brace, yshift=-10pt]
				\draw[decorate] (5,0) -- 
					node[below=2pt]{\scriptsize left recursive call} (0,0);
				\draw[decorate] (11,0) -- 
					node[below=2pt]{\scriptsize middle recursive call} (6,0); 
				\draw[decorate] (20,0) -- 
					node[below=2pt]{\scriptsize right recursive call} (12,0);
			\end{scope}
		\end{scope}
		\def\y{-3cm}
		\begin{pgfonlayer}{background}
			\begin{scope}[black!5]
	 			\fill (3,0) -- (6,0) -- (8,\y) -- (5,\y) -- cycle ;
	 			\fill (15,0) -- (16,0) -- (12,\y) -- (11,\y) -- cycle ;
			\end{scope}
			\begin{scope}[black!50]
				\draw[densely dashed] (3,0) -- (5,\y)  (6,0) -- (8,\y) ;
				\draw[densely dotted] (6,0) -- (3,\y)  (8,0) -- (5,\y) ;
				\draw[densely dotted] (11,0) -- (15,\y)  (12,0) -- (16,\y) ;
				\draw[densely dashed] (15,0) -- (11,\y)  (16,0) -- (12,\y) ;
			\end{scope}
		\end{pgfonlayer}
	\end{tikzpicture}%
	}
	\caption{%
		\textbf{First row:} State of the array just after partitioning the ordinary
		elements (after \wref{lin:generalized-partition-call} of
		\wref{alg:generalized-yaroslavskiy}).
		The letters indicate whether the element at this location is smaller ($s$),
		between ($m$) or larger ($l$) than the two pivots $P$ and $Q$.
		Sample elements are {\setlength\fboxsep{2pt}\colorbox{black!10}{shaded}}.
		\protect\newline
		\textbf{Second row:} State of the array after pivots and sample parts have
		been moved to their partition (after \wref{lin:generalized-swap-2}).
		The “rubber bands” indicate moved regions of the array.
	}
	\label{fig:swapping-of-sampled-out}
\end{figure}
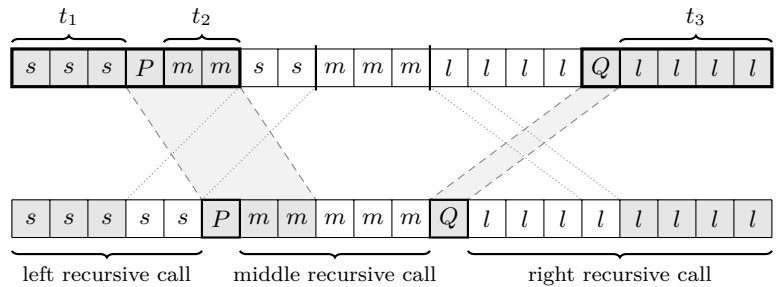

\oldparagraph{How do we deal with sampled-out elements?}
As discussed in \wref{sec:randomness-preservation}, we exclude sampled-out
elements from the partitioning range.
After partitioning, we thus have to move the $t_2$ sampled-out
elements, which actually belong between the pivots, to the middle
partition. 
Moreover, the pivots themselves have to be swapped in place. 
This process is illustrated in \wref{fig:swapping-of-sampled-out} and
spelled out in
lines~\ref*{lin:generalized-swap-t2-loop}\,--\,\ref*{lin:generalized-swap-2} of
\wref{alg:generalized-yaroslavskiy}.
Note that the order of swaps has been chosen carefully to correctly deal with
cases where the regions to be exchanged overlap.

\subsection{Randomness Preservation}
\label{sec:randomness-preservation}

For analysis, it is vital to preserve the input distribution for
recursive calls, as this allows us to set up a recurrence equation for
costs.
While Yaroslavskiy's method (as given in \wref{alg:partition})
preserves randomness inside partitions, pivot sampling requires special care.
For efficiently selecting the pivots, we \emph{sort} the entire sample, so the
sampled-out elements are far from randomly ordered; including
them in partitioning would not produce randomly ordered subarrays!
But there is also no need to include them in partitioning, as
we already have the sample divided into the three groups of 
$t_1$ small, $t_2$ medium and $t_3$ large elements.
All ordinary elements are still in random order and Yaroslavskiy's
partitioning divides them into three randomly ordered subarrays.

What remains problematic is the order of elements for recursive calls.
The second row in \wref{fig:swapping-of-sampled-out} shows the situation after
all sample elements (shaded gray) have been put into the correct subarray.
As the sample was sorted, the left and middle subarrays have sorted prefixes of
length $t_1$ resp.\ $t_2$ followed by a random permutation of the remaining
elements. Similarly, the right subarray has a sorted suffix of $t_3$ elements.
So the subarrays are \emph{not} randomly ordered, (except for the trivial case
$\vect t = 0$)! 
How shall we deal with this non-randomness?

The maybe surprising answer is that we can indeed \emph{exploit} this
non-randomness; not only in terms of a precise
analysis, but also for efficiency:
the sorted part \emph{always} lies completely inside the \emph{sample range} for
the next partitioning phase.
So our specific kind of non-randomness only affects sorting the sample (in
subsequent recursive calls), but it never affects the partitioning process
itself!

It seems natural that sorting should somehow be able to profit from partially
sorted input, and in fact, many sorting methods are known to be
\textsl{adaptive} to existing order \citep{EstivillCastro1992}.
For our special case of a fully sorted prefix or suffix of length $s\ge1$ and a
fully random rest, we can simply use Insertionsort where the first $s$
iterations of the outer loop are skipped.
Our Insertionsort implementations accept $s$ as an additional
parameter.

For Insertionsort, we can also precisely \emph{quantify} the savings resulting from
skipping the first $s$ iterations:
Apart from per-call overhead, we save exactly what it would have costed
to sort a random permutation of the length of this prefix/suffix
with Insertionsort. 
As all prefixes/suffixes have constant lengths (independent of the length
of the current subarray), precise analysis remains feasible, 
see \wref{sec:recurrence-quicksort}.

\begin{algorithm}
	\vspace{-1ex}
	\def\pl{\id{partLeft}}
	\def\pr{\id{partRight}}
	\def\pind{\id{i_p}}
	\def\qind{\id{i_q}}
	\def\case#1{$\kw{in case }\like[l]{\texttt{middle}\,}{\texttt{#1}}\;\kw{do}\;$}
	\begin{codebox}
		\Procname{$\proc{GeneralizedYaroslavskiy}\,(\arrayA,\id{left},\id{right},\id{type})$}
		\zi \Comment Assumes $\id{left} \le \id{right}$, $\isthreshold\ge k-1$
		\zi \Comment Sorts $A[\id{left},\ldots,\id{right}]$.
				\rule[-2ex]{0pt}{1ex}
		\li	\If $\id{right} - \id{left} < \isthreshold$
		\li	\Then
				\kw{case distinction} on $\id{type}$
		\li		\Do
		 			\case{root}		
		 			$\like[l]{\proc{InsertionSortRight}}{\proc{InsertionSortLeft}}
		 				\,(\arrayA,\id{left},\id{right},1)$
		\li			\case{left}	
					$\like[l]{\proc{InsertionSortRight}}{\proc{InsertionSortLeft}}
						\,(\arrayA,\id{left},\id{right},\max\{t_1,1\})$
		\li			\case{middle}	
					$\like[l]{\proc{InsertionSortRight}}{\proc{InsertionSortLeft}}
						\,(\arrayA,\id{left},\id{right},\max\{t_2,1\})$
					\li			\case{right}	
					$\proc{InsertionSortRight}\,(\arrayA,\id{left},\id{right},\max\{t_3,1\})$
				\End
		\li		\kw{end cases}
		\li	\Else
		\li			\kw{case distinction} on $\id{type}$ 
					\quad \Comment Sort sample
		\li			\Do
						\case{root}
						$\like[l]{\proc{SampleSortRight}}{\proc{SampleSortLeft}}
		 					\,(\arrayA,\id{left},\id{right},1)$
		\li				\case{left}
						$\like[l]{\proc{SampleSortRight}}{\proc{SampleSortLeft}}
		 					\,(\arrayA,\id{left},\id{right},\max\{t_1,1\})$
		\li				\case{middle}
						$\like[l]{\proc{SampleSortRight}}{\proc{SampleSortLeft}}
		 					\,(\arrayA,\id{left},\id{right},\max\{t_2,1\})$
		\li				\case{right}
						$\proc{SampleSortRight}
		 					\,(\arrayA,\id{left},\id{right},\max\{t_3,1\})$
					\End
		\li			\kw{end cases}
					\label{lin:generalized-sort-sample}
		\li			$p \gets \arrayA[\id{left} + t_1]$;
					\quad $q \gets \arrayA[\id{right}-t_3]$
		\li			$\pl \gets \id{left} + t_1 + t_2 + 1$;
					\quad $\pr \gets \id{right} - t_3 - 1$
		\li			$(\pind,\qind) \gets
						\proc{PartitionYaroslavskiy}\,(\arrayA,\pl,\pr,p,q)$
					\label{lin:generalized-partition-call}
					\rule[-2ex]{0pt}{1ex}
		\zi			\Comment Swap middle part of sample and $p$ to final place 
						(cf.\ \wref{fig:swapping-of-sampled-out})
		\li			\For $j \gets t_2 ,\ldots, 0$ 
					\quad \Comment iterate downwards
					\label{lin:generalized-swap-t2-loop}
		\li			\Do
						Swap $\arrayA[\id{left} + t_1 + j]$ and $\arrayA[\pind - t_2 + j]$
					\EndFor
		\zi			\Comment Swap $q$ to final place.
		\li			Swap $\arrayA[\qind]$ and $\arrayA[\pr+1]$
					\label{lin:generalized-swap-2}
					\rule[-2ex]{0pt}{1ex}
		\li			$\proc{GeneralizedYaroslavskiy}\,
						(\arrayA, 
						\like[l]{\pind-t_2+1,\;}{\id{left},}
						\like[l]{\pind-t_2-1,\;}{\pind-t_2-1,} 
						\like[l]{\texttt{middle}}{\texttt{left}})$
		\li			$\proc{GeneralizedYaroslavskiy}\,
						(\arrayA, 
						\like[l]{\pind-t_2+1,\;}{\pind-t_2+1,}
						\like[l]{\pind-t_2-1,\;}{\qind-1,} 
						\texttt{middle})$
		\li			$\proc{GeneralizedYaroslavskiy}\,
						(\arrayA,
						\like[l]{\pind-t_2+1,\;}{\qind+1,}
						\like[l]{\pind-t_2-1,\;}{\id{right},}
						\like[l]{\texttt{middle}}{\texttt{right}})$
			\EndIf
	\end{codebox}
	\vspace{-1ex}
	\caption{%
		\strut Yaroslavskiy's Dual-Pivot Quicksort with Generalized Pivot Sampling
	}
	\label{alg:generalized-yaroslavskiy}
\end{algorithm}

\begin{algorithm}
	\vspace{-1ex}
	\begin{codebox}
		\Procname{$\proc{InsertionSortLeft}(\arrayA,\id{left},\id{right},s)$}
		\zi \Comment Assumes $\id{left} \le \id{right}$ and 
				$s \le \id{right}-\id{left} -1$.
		\zi \Comment Sorts $\arrayA[\id{left},\ldots,\id{right}]$, assuming that the
				$s$ \textit{leftmost} elements are already sorted.
				\rule[-2ex]{0pt}{1ex}
		\li	\For $i=\id{left} + s \,,\dots,\, \id{right}$
		\li	\Do
				$j\gets i-1$; \quad
				$v\gets \arrayA[i]$
		\li		\While $j \ge \id{left} \wbin\wedge v < \arrayA[j]$ 
					\label{lin:insertionsort-left-comp-1}
		\li		\Do
					$\arrayA[j+1] \gets \arrayA[j]$; \quad	\label{lin:insertionsort-left-write-1}
					$j\gets j-1$
				\EndWhile
		\li		$\arrayA[j+1] \gets v$ \label{lin:insertionsort-left-write-2}
			\EndFor
	\end{codebox}
	\vspace{-1ex}
	\caption{%
		\strut Insertionsort “from the left”, exploits sorted prefixes. 
	}
	\label{alg:insertionsort-left}
\end{algorithm}

\begin{algorithm}
	\vspace{-1ex}
	\begin{codebox}
		\Procname{$\proc{InsertionSortRight}(\arrayA,\id{left},\id{right},s)$}
		\zi \Comment Assumes $\id{left} \le \id{right}$ and 
				$s \le \id{right}-\id{left} -1$.
		\zi \Comment Sorts $\arrayA[\id{left},\ldots,\id{right}]$, assuming that the
				$s$ \textit{rightmost} elements are already sorted.
				\rule[-2ex]{0pt}{1ex}
		\li	\For $i=\id{right} - s \,,\dots,\, \id{left}$ 
			\quad\Comment iterate downwards 
		\li	\Do
				$j\gets i+1$; \quad
				$v\gets \arrayA[i]$
		\li		\While $j \le \id{right} \wbin\wedge v > \arrayA[j]$
						\label{lin:insertionsort-right-comp}
		\li		\Do
					$\arrayA[j-1] \gets \arrayA[j]$; \quad	\label{lin:insertionsort-right-write-1}
					$j\gets j+1$
				\EndWhile
		\li		$\arrayA[j-1] \gets v$ \label{lin:insertionsort-right-write-2}
			\EndFor
	\end{codebox}
	\vspace{-1ex}
	\caption{%
		\strut Insertionsort “from the right”, exploits sorted suffixes. 
	}
	\label{alg:insertionsort-right}
\end{algorithm}

\begin{algorithm}
	\vspace{-1ex}
	\begin{codebox}
		\Procname{$\proc{SampleSortLeft}(\arrayA,\id{left},\id{right},s)$}
		\zi \Comment Assumes $\id{right} - \id{left} + 1 \ge k$ and 
				$s \le t_1+t_2+1$.
		\zi \Comment Sorts the $k$ elements
		$\arrayA[\id{left}],\ldots,\arrayA[\id{left}+t_1+t_2],
		 \arrayA[\id{right}-t_3],\ldots,\arrayA[\id{right}]$, 
		\zi \Comment assuming that the $s$ leftmost 
				elements are already sorted.
				\rule[-1.75ex]{0pt}{1ex}
		\zi \Comment $\arrayA\llbracket i \rrbracket$ is used as abbreviation 
			for $\arrayA[i+\id{offset}]$, where $\id{offset}$ 
			has to be computed as follows:
		\zi\Comment \kw{if} $i > \id{left} + t_1+t_2$ \kw{then} $\id{offset} \gets n-k$ 
			\kw{else} $\id{offset}\gets0$ \kw{end if},
		\zi\Comment where $n=\id{right}-\id{left}+1$.
				\rule[-2ex]{0pt}{1ex}
		\li $\proc{InsertionSortLeft}(\arrayA,\id{left},
					\id{left}+t_1+t_2,s)$
		\li	\For $i=\id{left} + t_1+t_2+1 \,,\dots,\, \id{left} + k - 1$
		\li	\Do
				$j\gets i-1$; \quad
				$v\gets \arrayA \llbracket i \rrbracket$
		\li		\While $j \ge \id{left} \wbin\wedge v < \arrayA\llbracket j\rrbracket$
		\li		\Do
					$\arrayA\llbracket j+1\rrbracket \gets \arrayA\llbracket j\rrbracket$;
					\quad \label{samplesort-left-write-1} $j\gets j-1$
				\EndWhile
		\li		$\arrayA\llbracket j+1\rrbracket \gets v$
				\label{lin:samplesort-left-write-2}
			\EndFor
	\end{codebox}
	\vspace{-1ex}
	
	\caption{%
		\strut Sorts the sample with Insertionsort “from the left” 
	}
	\label{alg:samplesort-left}
\end{algorithm}

\begin{algorithm}
	\vspace{-1ex}
	\begin{codebox}
		\Procname{$\proc{SampleSortRight}(\arrayA,\id{left},\id{right},s)$}
		\zi \Comment Assumes $\id{right} - \id{left} + 1 \ge k$ and 
				$s \le t_3+1$.
		\zi \Comment Sorts the $k$ elements
		$\arrayA[\id{left}],\ldots,\arrayA[\id{left}+t_1+t_2],
		 \arrayA[\id{right}-t_3],\ldots,\arrayA[\id{right}]$, 
		\zi \Comment assuming that the $s$ rightmost elements are
			already sorted.
				\rule[-1.75ex]{0pt}{1ex}
		\zi \Comment $\arrayA \llbracket i \rrbracket$ is used as abbreviation 
			for $\arrayA[i+\id{offset}]$, where $\id{offset}$ 
			has to be computed as follows:
		\zi\Comment \kw{if} $i > \id{left} + t_1+t_2$ \kw{then} $\id{offset} \gets n-k$ 
			\kw{else} $\id{offset}\gets0$ \kw{end if},
		\zi\Comment 
			where $n=\id{right}-\id{left}+1$.
				\rule[-2ex]{0pt}{1ex}
		\li $\proc{InsertionSortRight}(\arrayA,
					\id{right}-t_3,\id{right},s)$
		\li	\For $i=\id{left} + k - t_3 - 2 \,,\dots,\, \id{left}$ 
			\quad\Comment iterate downwards 
		\li	\Do
				$j\gets i+1$; \quad
				$v\gets \arrayA \llbracket i\rrbracket$
		\li		\While $j \le \id{left} + k - 1 \wbin\wedge 
					v > \arrayA \llbracket j \rrbracket$ 
		\li		\Do
					$\arrayA \llbracket j-1 \rrbracket \gets 
					\arrayA \llbracket j \rrbracket$; \quad	
					\label{lin:samplesort-right-write-1}
					$j\gets j+1$
				\EndWhile
		\li		$\arrayA \llbracket j-1 \rrbracket \gets v$
		\label{lin:samplesort-right-write-2}
			\EndFor
	\end{codebox}
	\vspace{-1ex}
	\caption{%
		\strut Sorts the sample with Insertionsort “from the right”
	}
	\label{alg:samplesort-right}
\end{algorithm}

\subsection{Generalized Yaroslavskiy Quicksort}

Combining the implementation of generalized pivot sampling\,---\,paying
attention to the subtleties discussed in the previous sections\,---\,with
Yaroslavskiy's partitioning method, 
we finally obtain \wref{alg:generalized-yaroslavskiy}.
We refer to this sorting method as \textsl{Generalized Yaroslavskiy Quicksort}
with pivot sampling parameter $\vect t = (t_1,t_2,t_3)$ and Insertionsort
threshold \isthreshold, shortly written as $\generalYarostM$.
We assume that $\isthreshold \ge k-1 = t_1+t_2+t_3+1$ to make sure
that every partitioning step has enough elements for pivot sampling.

The last parameter of \wref{alg:generalized-yaroslavskiy} tells the current call
whether it is a topmost call (\texttt{root}) or a recursive call on a left,
middle or right subarray of some earlier invocation. 
By that, we know which part of the array is already sorted:
for \texttt{root} calls, we cannot rely on anything being sorted,
in \texttt{left} and \texttt{middle} calls, we have a sorted prefix of length
$t_1$ resp.\ $t_2$, and for a \texttt{right} call, the $t_3$ rightmost
elements are known to be in order.
The initial call then takes the form
$\proc{GeneralizedYaroslavskiy}\,(\arrayA,1,n,\texttt{root})$.
\medskip

\FloatBarrier

\section{Results}
\label{sec:results}

For $\vect t \in \N^3$ and $\harm{n} = \sum_{i=1}^n \!\frac1i$ the $n$th harmonic number, we define the
\textsl{discrete entropy} $\discreteEntropy = \discreteEntropy[\vect t]$ of~$\vect t$ as
\begin{align}
\label{eq:discrete-entropy}
		\discreteEntropy[\vect t]
	&\wwrel=
		\sum_{r=1}^3 \frac{t_r+1}{k+1} 
				(\harm{k+1} - \harm{t_r+1})
	\;. 
\end{align} 
The name is justified by the following connection between \discreteEntropy and
the \weakemph{entropy function} $\contentropy[]$ of information theory: 
for the sake of analysis, let $k\to\infty$, such that  
ratios ${t_r}/k$ converge to constants~$\tau_r$. Then
\begin{align}
\label{eq:limit-g-entropy}
		\discreteEntropy
	&\wwrel\sim
		- \sum_{r=1}^3 \tau_r \bigl( \ln(t_r+1) - \ln(k+1) \bigr)
	\wwrel\sim	  - \sum_{r=1}^3 \tau_r \ln(\tau_r)
	\wwrel{\equalscolon} \contentropy[\vect\tau]
	\;.
\end{align}
The first step follows from the asymptotic equivalence
$\harm{n} \sim \ln(n)$ as $n\to\infty$. 
\wref[Equation]{eq:limit-g-entropy} shows that for large~$\vect t$, 
the maximum of $\discreteEntropy$ is attained for 
$\tau_1 = \tau_2 = \tau_3 = \frac13$.
Now we state our main result.

\begin{theorem}[Main theorem]
\label{thm:expected-costs}
	Generalized Yaroslavskiy Quicksort with pivot sampling
	parameter $\vect t = (t_1,t_2,t_3)$ performs on average
	$ C_n \sim \frac{a_C}{\discreteEntropy} \, n\ln n$
	comparisons,
	$ S_n \sim \frac{a_S}{\discreteEntropy} \, n\ln n$
	swaps and
	$ \scans_n \sim \frac{a_{\scans}}{\discreteEntropy} \, n\ln n$
	element scans to sort a random permutation of $n$ elements, where 
	\begin{align*}
			a_C
		&\wwrel=
				  1 + \frac{t_2+1}{k+1}
				    + \frac{(2t_1 + t_2 + 3)(t_3+1)}{(k+1)(k+2)},
	\\
			a_S
		&\wwrel=
				  \frac{t_1+1}{k+1}
				+ \frac{(t_1 + t_2 + 2)(t_3 + 1)}{(k+1)(k+2)} \qquad\text{and}
	\\
			a_{\scans}
		&\wwrel=
				1 + \frac{t_1+1}{k+1}\;.
	\end{align*}
	Moreover, if the partitioning loop is implemented as in Appendix~C 
	of~\citep{Wild2013Quicksort}, it executes on average 
	$\bytecodes_{\!n} \sim \frac{a_{\bytecodes}}{\discreteEntropy} \, n\ln n$
	Java Bytecode instructions to sort a random permutation of size $n$ with
	\begin{align*}
			a_\bytecodes
		&\wwrel=
				  		10
				\bin+ 	13 \frac{t_1+1}{k+1}
				\bin+  	 5 \frac{t_2+1}{k+1}
				\bin+ 	11 \frac{(t_1+t_2 + 2)(t_3+1)}{(k+1)(k+2)}
	\nonumber\\*	&\wwrel\ppe \quad {}
				\bin+ 	   \frac{(t_1+1)(t_1 + t_2 + 3)}{(k+1)(k+2)}
			\;.
	\end{align*}
\end{theorem}

\noindent
The following sections are devoted to the proof of \wref{thm:expected-costs}.
\wref{sec:distributional-analysis} sets up a recurrence of costs
and characterizes the distribution of costs of one partitioning step.
The expected values of the latter are computed in \wref{sec:expectations}.
Finally, \wref{sec:solution-recurrence} provides a generic solution to the
recurrence of the expected costs; in combination with the expected partitioning
costs, this concludes our proof.

\section{Distributional Analysis}
\label{sec:distributional-analysis}

\subsection{Recurrence Equations of Costs}
\label{sec:recurrence-quicksort}

\newcommand*\typetype{\smash{\mathtt{type}}}
\newcommand*\typeroot{\smash{\mathtt{root}}}
\newcommand*\typeleft{\smash{\mathtt{left}}}
\newcommand*\typeright{\smash{\mathtt{right}}}
\newcommand*\typemiddle{\smash{\mathtt{middle}}}

Let us denote by $C_n^{\typeroot}$ the \emph{costs} of \generalYarostM on a random
permutation of size~$n$\,---\,where the different \emph{cost measures} 
introduced in \wref{sec:cost-measures} will take the place of $C_n^{\typeroot}$ later. 
$C_n^{\typeroot}$~is a non-negative \emph{random} variable whose distribution
depends on~$n$.
The total costs decompose into those for the first partitioning step plus the
costs for recursively solving subproblems.

Due to our implementation of the pivot sampling method (see~\wref{sec:generalized-pivot-sampling-implementation}), 
the costs for a recursive call do not only depend on the size of the
subarray, but also on the \emph{type} of the call, \ie, whether it
is a left, middle or right subproblem or the topmost call:
Depending on the type, a part of the array will already be in order, 
which we exploit either in sorting the sample (if $n>\isthreshold$)
or in sorting the whole subarray by Insertionsort (if $n\le\isthreshold$).
We thus write $C_n^{\typetype}$ for the (random) cost of a call to 
$\proc{GeneralizedYaroslavskiy}(\arrayA,i,j,\mathtt{type})$
with $j-i-1 = n$ (\ie, $\arrayA[i..j]$ contains $n$ elements)
where $\typetype$ can either be $\typeroot$ (for the initial topmost call)
or one of $\typeleft$, $\typemiddle$ and $\typeright$.

As Yaroslavskiy's partitioning method applied to a random permutation
always generates subproblems with the same distribution (see
\wref{sec:randomness-preservation}), we can express the total costs
recursively in terms of the same cost functions with smaller arguments:
for sizes $J_1$, $J_2$ and $J_3$ of the three subproblems, the costs of
corresponding recursive calls are distributed like $C_{J_1}^{\typeleft}$, 
$C_{J_2}^{\typemiddle}$
and~$C_{J_3}^{\typeright}$, and conditioned on $\vect J = (J_1,J_2,J_3)$, these
random variables are independent.
Note, however, that the subproblem sizes are
themselves random and not independent of each other 
(they have to sum to $n-2$).
Denoting by $T_n^{\typetype}$ the (random) cost contribution of the first 
partitioning round to $C_n^{\typetype}$, 
we obtain the following \emph{distributional recurrence} for the four families
$(C_n^{\typetype})_{n\in\N}$ of random variables 
with $\typetype\in\{\typeroot,\typeleft,\typemiddle,\typeright\}$:
\begin{align}
\label{eq:distributional-recurrence-with-types}
	C_n^{\typetype} &\wwrel\eqdist \begin{cases}
			T_n^{\typetype} 
			\wbin+ C_{J_1}^{\typeleft} 
				+ C_{J_2}^{\typemiddle} 
				+ C_{J_3}^{\typeright} ,
			& \text{for } n > \isthreshold ; \\
		\iscost_n^{\typetype} ,
			& \text{for } n \le \isthreshold .
	\end{cases}
\end{align}
Here $\iscost_n^{\typetype}$ denotes the (random) cost of sorting a subarray of 
size $n\le\isthreshold$ using Insertionsort from a (recursive) call of type $\typetype$.
We call $T_n^{\typetype}$ the \emph{toll functions} of the recurrence, as they quantify
the ``toll'' we have to pay for unfolding the recurrence once.
Our cost measures only differ in the toll functions, such that we can
treat them all in a uniform fashion by studying \wref[Equation]{eq:distributional-recurrence-with-types}.

Dealing with the mutually recursive quantities 
of~\wref[Equation]{eq:distributional-recurrence-with-types} is rather inconvenient, 
but we can luckily avoid it for our purposes.
$T_n^{\typeroot}$, $T_n^{\typeleft}$, $T_n^{\typemiddle}$ and $T_n^{\typeright}$
(potentially) differ in the cost of selecting pivots from the sample,
but they do \emph{not} differ in the cost caused by the partitioning procedure itself:
in all four cases, we invoke \proc{Partition} on a subarray containing $n-k$ elements
that are in random order and the (random) pivot values $P$ and $Q$ always 
have the same distribution.
As we assume that the sample size~$k$ is a constant independent of $n$, 
the toll functions differ by a constant at most; 
in fact for all $\mathtt{type}$s, we have
$T_n^{\typetype} \eqdist T_n + \Oh(1)$ where $T_n$ denotes the cost caused
by \proc{Partition} alone.
Since the total costs are a linear function of the toll costs, we can separately deal
with the two summands.
The contribution of the $\Oh(1)$ toll to the overall costs is then trivially 
bounded by $\Oh(n)$, as two (new) elements are chosen as pivots in each 
partitioning step, so we can have at most $n/2$ pivot sampling rounds in total.

Similarly, $\iscost_n^{\typetype} \eqdist \iscost_n + \Oh(1)$,
where $\iscost_n$ denotes the (random) costs of sorting a random permutation
of size $n$ with Insertionsort (without skipping the first few iterations).
The contribution of Insertionsort to the total costs
are in $\Oh(n)$ as the Insertionsort threshold \isthreshold is constant
and we can only have a linear number of calls to Insertionsort.
So for the leading term, the precise form of $\iscost_n$ is immaterial.
In summary, we have shown that 
$C_n^{\typetype} \eqdist C_n + \Oh(n)$, and in particular
$C_n^{\typeroot} \eqdist C_n + \Oh(n)$,
where the distribution of $C_n$ is defined by the following distributional recurrence:
\begin{align}
\label{eq:distributional-recurrence}
	C_n &\wwrel\eqdist \begin{cases}
			T_n 
			\wbin+ C^{\vphantom{\prime}}_{\smash{J_1}} 
			+ C'_{\smash{J_2}} 
			+ C^{\pprime}_{\smash{J_3}} ,
			& \text{for } n > \isthreshold ; \\
		\iscost_n ,
			& \text{for } n \le \isthreshold ,
	\end{cases}
\end{align}
with $(C'_{\smash{j}})_{j\in\N}$ and $(C^{\pprime}_{\smash{j}})_{j\in\N}$ 
independent copies of
\smash{$(C^{\vphantom{\prime}}_{\smash{j}})_{j\in\N}$}, 
\ie, for all $j$, the variables 
$C^{\vphantom{\prime}}_{\smash{j}}$, 
$C'_{\smash{j}}$ and
$C^{\pprime}_{\smash{j}}$ 
are identically distributed and for all \smash{$\vect j \in \N^3$}, 
$C^{\vphantom{prime}}_{\smash{j_1}}$, 
$C'_{\smash{j_2}}$ and 
$C^{\pprime}_{\smash{j_3}}$ 
are (totally) independent%
\footnote{%
	Total independence means that the joint probability function 
	of all random variables factorizes into the product of the 
	individual probability functions \citep[p.\,53]{Chung2001},
	and does so not only pairwise.
}, and they are also independent of $T_n$. 

To obtain an expression for $\Prob(\vect J = \vect j)$, 
we note that there are \smash{$\binom nk$} ways to choose
$k$ out of $n$ given elements in total. 
If there shall be exactly 
$j_1$ small, $j_2$ medium and $j_3$ large elements,
we have to choose $t_1$ of the $j_1$ small elements for the sample,
plus $t_2$ of the $j_2$ medium and $t_3$ of the $j_3$ large elements.
Combining all possibly ways to do so gives the number of samples
that are consistent with subproblem sizes $\vect j = (j_1,j_2,j_3)$;
we thus have
\begin{align}
		\Prob(\vect J = \vect j)
	&\wwrel=
		\binom{j_1}{t_1} \binom{j_2}{t_2} \binom{j_3}{t_3}
		\bigg/
		\binom nk
		\;.
\label{eq:prob-for-J-equals-j}
\end{align}

\needspace{5\baselineskip}
\subsection{Distribution of Partitioning Costs}

Recall that we only have to partition the \emph{ordinary} elements, \ie, 
the elements that have \emph{not} been part of the sample 
(cf.\ \wref{lin:generalized-partition-call} of \wref{alg:generalized-yaroslavskiy}).
Let us denote by $I_1$, $I_2$ and $I_3$ the number of small, medium and large
elements among these elements, \ie, $I_1+I_2+I_3 = n-k$.
Stated differently, $\vect I = (I_1,I_2,I_3)$ is the vector
of sizes of the three partitions (excluding sampled-out elements).
There is a close relation between the vectors of \emph{partition} sizes $\vect I$ 
and \emph{subproblem} sizes $\vect J$;
we only have to add the sampled-out elements again before the recursive calls:
$\vect J = \vect I + \vect t$ (see \wref{fig:swapping-of-sampled-out}).

Moreover, we define the indicator variable 
$\delta = \indicator{U_\chi \rel> Q}$ where $\chi$ 
is the array position on which indices $k$ and $g$ first meet.
$\delta$ is needed to account for an idiosyncrasy of 
Yaroslavskiy's algorithm:
depending on the element $U_\chi$ that is initially located at the position where 
$k$ and $g$ first meet,
$k$ overshoots $g$ at the end by either $2$\,---\,namely if $U_\chi > Q$\,---\,or 
by $1$, otherwise
\citep[``Crossing-Point Lemma'']{Wild2013Quicksort}.

As we will see, we can precisely characterize the distribution of partitioning costs
\emph{conditional} on $\vect I$, \ie, when considering $\vect I$ \emph{fixed}.
Therefore, we give the conditional distributions of all quantities relevant
for the analysis in \wref{tab:quantities-distribution-for-fixed-I}.
They essentially follow directly from the discussion in our previous work 
\citep{Wild2013Quicksort}, 
but for convenience, we give the main arguments again
in this paper.

\begin{table}
	\newcommand*\nextcol{\strut\\[1ex]\strut}
	\plaincenter{%
	\begin{tabular}{>{\(}c<{\)} >{\(=\)}c<{}  >{\(}c<{\)}>{\(\eqdist\)}c<{}  >{\(}c<{\)}}
	\toprule
		\multicolumn3{c}{Quantity} &\multicolumn1{c}{}& \multicolumn1{c}{Distribution given $\vect I$} \\
	\midrule
		\delta				&& \indicator{U_\chi > Q}	&& \bernoulli\bigl(\tfrac{I_3}{n-k}\bigr) \nextcol
		|\positionsets{K}|	&& I_1 + I_2 + \delta && I_1 + I_2 + \bernoulli\bigl(\tfrac{I_3}{n-k}\bigr) \nextcol
		|\positionsets{G}|	&& I_3 && I_3 \nextcol
		|\positionsets{L}|	&& I_1 && I_1 \nextcol
		\latK				&& (\numberat l{K'}) + \delta && \hypergeometric(I_1+I_2, I_3, n-k) + \bernoulli\bigl(\tfrac{I_3}{n-k}\bigr) \nextcol
		\satG 				&\multicolumn2{c}{} && \hypergeometric(I_3, I_1, n-k) \\
	\bottomrule
	\end{tabular}%
	}

	\caption{%
		Quantities that arise in the analysis of \proc{Partition} (\wref{alg:partition})
		and their distribution conditional on $\vect I$.
		A detailed discussion of these quantities and their 
		distributions is given in \citep{Wild2013Quicksort}.
		\protect\\
		Note that $|\positionsets{K}|$ depends on $\delta$, which is inconvenient for further analysis,
		so we work with $\positionsets{K'}$, defined as the first $I_1+I_2$ elements of 
		$\positionsets{K}$. 
		When $\delta=0$ we have $\positionsets{K'} = \positionsets{K}$,
		see \citep{Wild2013Quicksort} for details.
	}
	\label{tab:quantities-distribution-for-fixed-I}
\end{table}

Recall that $I_1$, $I_2$ and $I_3$ are the
number of small, medium and large elements, respectively. 
Since the elements right of $g$ after partitioning are exactly all large
elements (see also \wref{sfig:state-after-partitioning}), $g$ scans $I_3$ elements.
Note that the last value that variable $g$ attains is not part of $\positionsets{G}$, 
since it is never used to access the array.

All small and medium elements are for sure left of $k$ after partitioning. 
But $k$ might also run over the first large element, if $k$ and $g$ meet on a large element.
Therefore,
$
		|\positionsets{K}|
	=
		I_1+I_2+\delta
$
(see also the ``Crossing-Point Lemma'' of \citet{Wild2013Quicksort}).

The distribution of $\satG$, conditional on~$\vect I$, is given by the
following urn model:
We put all $n-k$ ordinary elements in an urn and draw their
positions in \arrayA.
$I_1$ of the elements are colored red (namely the small ones), the rest is black
(non-small).
Now we draw the $|\positionsets{G}| = I_3$ elements in $g$'s range from
the urn without replacement. 
Then $\satG$ is exactly the number of red (small)
elements drawn and thus 
$\satG \rel\eqdist \hypergeometric(I_3, I_1, n-k)$.

The arguments for $\latK$ are similar, however the additional $\delta$
in $|\positionsets{K}|$ needs special care.
As shown in the proof of Lemma~3.7 of \citet{Wild2013Quicksort}, 
the additional element in $k$'s range for the case $\delta=1$ is
$U_\chi$, which then is large by definition of~$\delta$.
It thus simply contributes as additional summand:
$\latK \rel\eqdist \hypergeometric(I_1+I_2, I_3, n-k) + \delta$.
Finally, the distribution of $\delta$ is Bernoulli
$\bernoulli\bigl(\tfrac{I_3}{n-k}\bigr)$, since conditional on $\vect I$, the
probability of an ordinary element to be large is $I_3 / (n-k)$.

\subsubsection{Comparisons}
Recall that we consider for $C_n$ only the comparisons from the \proc{Partition}
procedure;
as the sample size and the Insertionsort threshold are both constant, 
the number of other comparisons is bounded by $\Oh(n)$ and can thus 
be ignored for the leading term of costs.
It remains to count the comparisons during the first partitioning step,
which we will denote by $\toll C = \toll[n]C$ instead of the generic toll $T_n$.
Similarly, we will write $\toll S$, $\toll{\bytecodes}$ and $\toll{\scans}$
for the number of swaps, executed Bytecode instructions and scanned elements
incurred in the first call to \proc{Partition}.

One can approximate $\toll[n]C$ on an abstract and intuitive level as follows:
We need one comparison per ordinary element for sure, but
some elements require a a second one to classify them as 
small, medium or large.
Which elements are expensive and which are cheap (w.\,r.\,t.\ comparisons)
depends on the index\,---\,either $k$ or $g$\,---\,by which an element is reached: 
$k$ first compares with $P$, so small elements are classified with only
one comparison.
Elements scanned by $g$ are first compared with $Q$, so here the large
ones are beneficial.
Note that medium elements always need both comparisons.
Using the notation introduced in \wref{sec:yaroslavskiys-partitioning-method},
this gives a total of
$(n-k) + I_2 + (\latK) + (\satG)$ 
comparisons in the first partitioning step.

Some details of the partitioning algorithm are, however, easily
overlooked at this abstract level of reasoning: 
a summand $+2\delta$ is missing in the above result.
Essentially, the reason is that how much $k$ overshoots $g$ at the end of
partitioning depends on the class of the element $U_\chi$ on which they meet.
For the precise analysis, we therefore keep the argumentation closer to the actual 
algorithm at hand:
for each location in the code where a key comparison is done,
determine how often it is reached, then sum over all locations.
The result is given in the following lemma.

\begin{lemma}
\label{lem:distribution-partitioning-comparisons}
	Conditional on the partition sizes $\vect I$, the number of comparisons
	$\toll C = \toll [n] C$ in the first partitioning step of \generalYarostM on a
	random permutation of size $n>\isthreshold$ fulfills
	\begin{align*}
			\toll [n] C
		&\wwrel=
			|\positionsets{K}| + |\positionsets{G}| 
			\bin+ I_2 
			\bin+ (\latK) 
			\bin+ (\satG)
			\bin+ \delta
	\\[1ex]	&\wwrel{\like{=}{\displaystyle\eqdist}}
	\nonumber
			(n-k) 
			\bin+ I_2 
			\bin+ \hypergeometric(I_1+I_2,I_3,n-k)
	\\*	&\wwrel\ppe \hphantom{(n-k) \bin+ I_2}{}
			\bin+ \hypergeometric(I_3,I_1,n-k)
			\bin+ 3\bernoulli\bigl(\tfrac{I_3}{n-k}\bigr)
		\;.
	\end{align*} 
\end{lemma}

\begin{proof}
Each element that is accessed as $\arrayA[k]$ or $\arrayA[g]$ is directly compared
(\wref[lines]{lin:yaroslavskiy-comp-1} and~\ref*{lin:yaroslavskiy-comp-3} of 
\wref{alg:partition}),
so we get $|\positionsets{K}| + |\positionsets{G}|$ ``first'' comparisons.
The remaining contributions come from \wref[lines]{lin:yaroslavskiy-comp-2} and
\ref*{lin:yaroslavskiy-comp-4}.

\wref[Line]{lin:yaroslavskiy-comp-2} is reached for every
\emph{non-small} element in $k$'s range, giving a contribution of 
$(\numberat mK) + (\latK)$, where $\numberat mK$ denotes the number of medium elements
in $k$'s range.
Likewise, \wref{lin:yaroslavskiy-comp-4} is executed for every non-large element in $g$'s
range, giving $(\satG)+(\numberat mG)$ additional comparisons\,---\,but 
\wref{lin:yaroslavskiy-comp-4} is also reached when the inner loop is left 
because of the second part of the loop condition, \ie,
when the current element $\arrayA[g]$ is large, but $k\ge g$.
This can happen at most once since $k$ and $g$ have met then.
It turns out that we get an additional execution of \wref{lin:yaroslavskiy-comp-4}
\emph{if and only if} the element $U\chi$ where $k$ and $g$ meet is large; 
this amounts to $\delta$ additional comparisons.

We never reach a medium element by both $k$ and $g$
because the only element that is potentially accessed through both indices
is $U_\chi$ and it is only accessed via $k$ in case $U_\chi > Q$,
\ie, when it is \emph{not} medium.
Therefore, $(\numberat mK) + (\numberat mG) = I_2$,
which proves the first equation.
\citet{Wild2013Quicksort} give a more detailed explanation of the above arguments.
The equality in distribution directly follows from \wref{tab:quantities-distribution-for-fixed-I}.
\end{proof}

\subsubsection{Swaps}
As for comparisons, we only count the swaps in the partitioning step.

\begin{lemma}
\label{lem:distribution-partitioning-swaps}
	Conditional on the partition sizes $\vect I$, the number of swaps
	$\toll S = \toll [n] S$ in the first partitioning step of \generalYarostM on a
	random permutation of size $n>\isthreshold$ fulfills
	\begin{align*}
			\toll [n] S
		&\wwrel=
			 I_1 
			\bin+ (\latK) 
		\wwrel{\rel{\eqdist}}
			I_1 
			\bin+ \hypergeometric(I_1+I_2, I_3, n-k) 
			\bin+ \bernoulli\bigl(\tfrac{I_3}{n-k}\bigr)
		\;.
	\end{align*} 
\end{lemma}

\begin{proof}
No matter where a small element is located initially, it will eventually incur
one swap that puts it at its final place (for this partitioning step) to the
left of $\ell$, see \ref{prop:small-eventually-swappel-left};
this gives a contribution of $I_1$ swaps.
The remaining swaps come from the ``crossing pointer'' scheme, where $k$ stops
on every large element on its way and $g$ stops on all non-large elements. 
Whenever both $k$ and $g$ have stopped, the two out-of-order elements are
exchanged in one swap \ref{prop:latK-smatG-swapped-in-pairs}.
The number of such pairs is $\latK$, which proves the first equation.
The second equality follows from \wref{tab:quantities-distribution-for-fixed-I}.
\end{proof}

\subsubsection{Bytecode Instructions}

A closer investigation of the partitioning method reveals the number
of executions for every single Bytecode instruction in the algorithm.
Details are omitted here; the analysis is very similar to the case without
pivot sampling that is presented in detail in \citep{Wild2013Quicksort}.

\begin{lemma}
\label{lem:distribution-partitioning-bytecodes}
	\strut Conditional on the partition sizes $\vect I$, the number of executed Java
	Bytecode instructions $\toll \bytecodes = \toll [n] \bytecodes$ of the first
	partitioning step of \generalYarostM{}\,---\,implemented as in Appendix~C of
	\citep{Wild2013Quicksort}\,---\,fulfills
	on a random permutation of size	$n>\isthreshold$
	\vspace{-1ex}
	\begin{multline*}
			\toll[n]\bytecodes
		\wrel{\rel{\eqdist}}		
					10 n + 13 I_1 + 5 I_2
			+ 		11\, \hypergeometric(I_1+I_2,I_3,n-k)
 	\\*	%
			+		   \hypergeometric(I_1,I_1+I_2,n-k)
			\wbin+	\Oh(1) \;.
	\end{multline*}
\qed\end{lemma}

\subsubsection{Scanned Elements}

\begin{lemma}
\label{lem:distribution-partitioning-scans}
	\strut
	Conditional on the partition sizes $\vect I$, the number of scanned elements
	$\toll \scans = \toll [n] \scans$ in the first partitioning step of 
	\generalYarostM on a random permutation of size $n>\isthreshold$ fulfills
	\begin{align*}
			\toll[n]\scans
		&\wwrel= 
			|\positionsets{K}| + |\positionsets{G}| + |\positionsets{L}|
		\wwrel{\eqdist}
			(n-k) 
			\bin+ I_1 
			\bin+ \bernoulli\bigl(\tfrac{I_3}{n-k}\bigr)
		\;.
	\label{eq:distribution-partitioning-scans-dist}
	\end{align*}
\end{lemma}

\begin{proof}
The first equality follows directly from the definitions: 
Our position sets include exactly the indices of array accesses.
The equation in distribution is found using \wref{tab:quantities-distribution-for-fixed-I}.
\end{proof}

\subsubsection{Distribution of Partition Sizes}

By \wref{eq:prob-for-J-equals-j} and the relation $\vect J = \vect I + \vect t$ 
between $\vect I$, the number of small, medium and
large ordinary elements, and $\vect J$, the size of subproblems, we have
\smash{$
		\Prob(\vect I = \vect i)  
	=
		\binom{i_1+t_1}{t_1} \binom{i_2+t_2}{t_2} \binom{i_3+t_3}{t_3}
		\Big/
		\binom nk
$}.
Albeit valid, this form results in nasty sums with three binomials when we try
to compute expectations involving $\vect I$.

An alternative characterization of the distribution of $\vect I$ that is better
suited for our needs exploits that we have i.\,i.\,d.\ $\uniform(0,1)$
variables.
If we condition on the pivot \emph{values}, \ie, consider $P$ and $Q$
fixed, an ordinary element $U$ is small, if $U \in (0,P)$, 
medium if $U \in (P,Q)$ and 
large if $U \in (Q,1)$.
The lengths $\vect D = (D_1,D_2,D_3)$ of these three intervals 
(see \wref{fig:relations-DPQ}), 
thus are the \emph{probabilities} for an element to be small, medium or large,
respectively.
Note that this holds \emph{independently} of all other ordinary elements!
The partition sizes $\vect I$ are then obtained as the collective outcome of
$n-k$ independent drawings from this distribution, so 
conditional on $\vect D$, $\vect I$ is multinomially $\multinomial(n-k,\vect D)$
distributed.

With this alternative characterization, we have \emph{decoupled} the 
pivot \emph{ranks} (determined by $\vect I$) from the pivot
\emph{values}, which allows for a more elegant computation of expected values 
(see \wref{app:proof-of-lem-expectations}). 
This decoupling trick has (implicitly) been applied to the analysis of classic
Quicksort earlier, \eg, by \citet{Neininger2001multivariate}.

\begin{figure}
	\def\r{1.5pt}
	\plaincenter{%
		\begin{tikzpicture}[
			every node/.style={font=\footnotesize},
		]

			\draw[|-|] (0,0) node[below=.5ex] {$0$}  -- (5,0) node[below=.5ex] {$1$};
			\filldraw (1,0) circle (\r) node[below=.5ex] {$P$};
			\filldraw (3.5,0) circle (\r) node[below=.5ex] {$Q$};
			\begin{scope}[
					yshift=1.5ex,<->,
					shorten >=.3pt,shorten <=.3pt,
					fill=white,inner sep=1pt,
			]	
				\draw (0,0)   -- node[fill] {$D_1$} (1,0) ;
				\draw (1,0)   -- node[fill] {$D_2$} (3.5,0) ;
				\draw (3.5,0) -- node[fill] {$D_3$} (5,0) ;
			\end{scope}
		\end{tikzpicture}%
	}
	\caption{%
		Graphical representation of the relation between $\vect D$ and the pivot
		values $P$ and $Q$ on the unit interval.
	}
	\label{fig:relations-DPQ}
\end{figure}
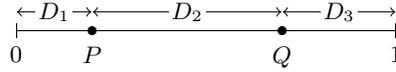

\subsubsection{Distribution of Pivot Values}

The input array is initially filled with $n$ i.\,i.\,d.\ $\uniform(0,1)$
random variables %
from which we choose a sample $\{V_1,\ldots,V_k\} \subset \{U_1,\ldots,U_n\}$ 
of size $k$. 
The pivot values are then selected as order statistics of the sample:
$P \ce V_{(t_1+1)}$ and $Q \ce V_{(t_1+t_2+2)}$ 
(cf.\ \wref{sec:general-pivot-sampling}). 
In other words, $\vect D$ is the vector of \textsl{spacings}
induced by the order statistics $V_{(t_1+1)}$ and $V_{(t_1+t_2+2)}$ of $k$
i.\,i.\,d.\ $\uniform(0,1)$ variables $V_1,\ldots,V_k$, 
which is known to have a \weakemph{Dirichlet}
$\dirichlet(\vect t + 1)$ distribution
(\wref{pro:spacings-dirichlet-general-dimension}).

\section{Average-Case Analysis}
\label{sec:average-case-analysis}

\subsection{Expected Partitioning Costs}
\label{sec:expectations}

In \wref{sec:distributional-analysis}, we characterized the full distribution of
the costs of the first partitioning step. 
However, since those distributions are \emph{conditional} on other
random variables, we have to apply the \textsl{law of total expectation}.
By linearity of the expectation, it suffices to consider the
summands given in the following lemma.

\begin{lemma}
\label{lem:expectations}
	For pivot sampling parameter 
	$\vect t \in \N^3$ and
	partition sizes $\vect I \eqdist \multinomial(n-k,\vect D)$,
	based on random spacings $\vect D \eqdist \dirichlet(\vect t + 1)$, 
	the following (unconditional) expectations hold:
	\begin{align*}
			\E[I_j]
		&\wwrel=
			\frac{t_j+1}{k+1} (n-k) \,,
			\qquad\qquad 
			(j=1,2,3),
	\\
			\E\bigl[\bernoulli\bigl(\tfrac{I_3}{n-k}\bigr)\bigr]
		&\wwrel=
			\frac{t_3+1}{k+1}
			\wwrel{\wrel=} \Theta(1) \,,
			\qquad\quad
			(n\to\infty),
	\\
			\E\bigl[ \hypergeometric(I_3,I_1,n-k) \bigr]
		&\wwrel=
			\frac{(t_1+1)(t_3+1)}{(k+1)(k+2)} (n-k-1) \,,
	\\
			\E\bigl[ \hypergeometric(I_1+I_2, I_3, n-k) \bigr]
		&\wwrel=
			\frac{(t_1+t_2+2)(t_3+1)}{(k+1)(k+2)} (n-k-1) \;.
	\end{align*}
\end{lemma}

\noindent
Using known properties of the involved distributions, the proof is
an elementary computation. 
It is given in detail in \wref{app:proof-of-lem-expectations} 
for interested readers.

The direct consequence of \wref{lem:expectations} is that for all our cost measures,
we have expected partitioning costs of the form $\E[T_n] = a\.n + b$ with constants
$a$ and $b$.

\subsection{Solution of the Recurrence}
\label{sec:solution-recurrence}

By taking expectations on both sides of the 
distributional recurrence (\wildtpageref[Equation]{eq:distributional-recurrence}{\eqref}),
we obtain an ordinary recurrence for the sequence of expected costs
$\bigl(\E[C_n]\bigr)_{n\in\N}$.
We solve this recurrence using \citeauthor{Roura2001}'s
\textsl{Continuous Master Theorem (CMT)} \citep{Roura2001},
but first give an informal derivation of the solution to
convey the main intuition behind the CMT.
Precise formal arguments are then given in \wref{app:CMT-solution}.

\subsubsection{Rewriting the Recurrence}

To solve the recurrence, it is convenient first to rewrite
\wref[Equation]{eq:distributional-recurrence} a little.
We start by \emph{conditioning} on $\vect J$. 
For $n > \isthreshold$, this gives
\begin{align*}
		C_n 
	&\wwrel\eqdist
		T_n \wbin+ 
		\sum_{j=0}^{n-2} \Bigl( 
			\indicator{J_1=j} C_j
			\bin+\indicator{J_2=j} C'_j
			\bin+\indicator{J_3=j} C''_j
		\Bigr)
		\;.
\end{align*}
Taking expectations on both sides and exploiting independence yields
\begin{align}
		\E[C_n] 
	&\wwrel= 
	\begin{cases}
		\displaystyle\vphantom{\bigg)}
			\E[T_n] \wbin+ 
			\smash{
				\sum_{j=0}^{n-2} \E[C_j] 
				\sum_{r=1}^3 \Prob\bigl( J_r = j \bigr)
			}
		& \text{for } n > \isthreshold  ; \\[2ex]
			\displaystyle\vphantom{\bigg)}
			\E[\iscost_n] ,
		& \text{for } n \le \isthreshold .
	\end{cases}
\label{eq:ECn-recurrence}
\end{align}
By definition, $J_r = I_r + t_r$ and, conditional on $\vect D$, 
$I_r$ is $\binomial(n-k,D_r)$ distributed for $r=1,2,3$.
(The marginal distribution of a multinomial vector is the binomial distribution.)
We thus have conditional on $\vect D$ that 
$$
		\Prob\bigl( I_r = i \bigr)
	\wwrel=
		\binom{n-k}{i} D_r^i (1-D_r)^{n-k-i}
$$
and upon unconditioning
$$
		\Prob\bigl( J_r = j \bigr)
	\wwrel=
		\binom{n-k}{j-t_r}
		\E_{\vect D}\Bigl[ 
			D_r^{j-t_r} (1-D_r)^{n-k-j+t_r}
		\Bigr]
		\;.
$$

There are three cases to distinguish depending on the toll function,
which are well-known from the classical \textsl{master theorem} 
for divide-and-conquer recurrences:
\begin{enumerate}
\item
	If the toll function grows very fast with $n$, the first recursive call will 
	dominate overall costs, as the toll costs of subproblems are small in 
	relation to the first step.
\item
	On the other hand, if the toll function grows very slow with $n$, 
	the topmost calls will be so cheap in relation that the 
	\emph{number} of base case calls on constant size subproblems 
	will dictate overall costs.
\item
	Finally, for toll functions of just the right rate of growth, the recursive calls
	on each level of the recursion tree sum up to (roughly) the same cost and
	the overall solution is given by this sum of costs times the recursion depth.
\end{enumerate}
Binary search and Mergesort are prime examples of the third case, 
in the analysis of Karatsuba's integer multiplication or Strassen's matrix multiplication, 
we end up in the second case and 
in the Median-of-Medians selection algorithm the initial call is asymptotically dominating 
and we get the first case (see, \eg, \citet{Cormen2009}).

Our \wref[Equation]{eq:ECn-recurrence} shows essentially the same three 
cases depending on the asymptotic growth of $\E[T_n]$.
The classical master theorem distinguishes the cases by comparing, for large $n$, 
the toll of the topmost call with the total tolls of all its immediate child 
recursive calls.
If there is an (asymptotic) imbalance to the one or the other side,
this imbalance will eventually dominate for large $n$.
The same reasoning applies to our recurrence, only that
computations become a little trickier since 
the subproblem sizes are not fixed a priori.

Let us first symbolically substitute $z\.n$ for $j$ in \wref{eq:ECn-recurrence},
so that $z\in[0,1]$ becomes the \emph{relative subproblem size}:
\begin{align*}
		\E[C_n] 
	&\wwrel= 
		\E[T_n] \wbin+ \sum_{z\.n=0}^{n-2} \E[C_{zn}] 
			\sum_{r=1}^3 \Prob\biggl( \frac{J_r}n=z \biggr) 
			\;.
\end{align*}
In the sum over $z\.n$, $n$ of course remains unchanged, and
$z$ moves $0$ towards $1$.
When $n$ gets larger and larger, $z$ ``scans'' the unit interval
more and more densely, so that it is plausible to approximate 
the sum by an integral:
%
\begin{align*}
		\sum_{z\.n=0}^{n-2} \E[C_{zn}] 
			\sum_{r=1}^3 \Prob\biggl( \frac{J_r}n=z \biggr) 
	&\wwrel\approx
		\int_{z=0}^1\E[C_{zn}]
			\sum_{r=1}^3 \Prob\biggl( \frac{J_r}n=z\pm\frac1{2n} \biggr) 
		\:dz \;.
\end{align*}
This idea has already been used by \citet{VanEmden1970VanFun} to
compute the number of comparisons for classic Quicksort with median-of-three\,---\,in fact
he was the first to derive that number analytically.
However, some continuity assumptions are silently made in this step and a 
rigorous derivation has to work out the error terms that we make by this approximation.
We defer a formal treatment of these issues to \wref{app:CMT-solution}.

Finally, $J_r=I_r+t_r$ has the expectation
$\E[J_l\given \vect D] = D_l\.n + t_r$ conditional on $\vect D$ and so for large $n$
\begin{align*}
		\sum_{z\.n=0}^{n-2} \E[C_{zn}] 
			\sum_{r=1}^3 \Prob\biggl( \frac{J_r}n=z \biggr) 
	&\wwrel\approx
		\int_{z=0}^1\E[C_{zn}]
			\sum_{r=1}^3 \Prob\Bigl( D_r=z\pm\tfrac1{2n} \Bigr) 
		\:dz \,.
\end{align*}
%
Intuitively, this means that the \emph{relative subproblem sizes} 
in dual-pivot Quicksort with pivot sampling parameter $\vect t$
have a Dirichlet distribution with parameters
$\dirichlet(t_1+1,k-t_1)$, $\dirichlet(t_2+1,k-t_2)$ and 
$\dirichlet(t_3+1,k-t_3)$, respectively.
The main advantage of this last form is that the integral does not depend on $n$
anymore and we obtain the following \emph{continuous} recurrence for~$\E[C_n]$:
\begin{align}
\label{eq:continuous-recurrence}
		\E[C_n]
	&\wwrel\approx
		\E[T_n] \bin+ \int_{0}^1 \!\! w(z) \E[C_{zn}] \:dz \,,
\end{align}
for a ``shape function'' $w(z) \ce \sum_{r=1}^3 f_{D_r} (z) $ where
$f_{D_r}$ is the density function of the $\dirichlet(t_r+1,k-t_r)$ distribution.

\subsubsection{Which Case of the Master Theorem?}

We are now in the position to compare the toll of the first call $\E[T_n]$
to the total tolls of its child recursive calls, \ie, how
\begin{align}
\label{eq:CMT-integral-symbolic-T}
		\int_{0}^1\E[T_{zn}]\, w(z)
		\:dz
\end{align}
relates to $\E[T_n]$.
We assume $\E[T_n] = a\.n + \Oh(n^{1-\epsilon})$ for $\epsilon>0$,
which for our cost measures is fulfilled with $\epsilon=1$.
As $\E[C_n]$ is linear in $\E[T_n]$, we can solve the recurrence for the 
leading term $an$ and the error term $\Oh(n^{1-\epsilon})$ separately.
When working out the integrals, it turns out that
\begin{align}
\label{eq:master-theorem-case1-integral}
		\int_{0}^1 a\.zn \, w(z)
		\:dz
	&\wwrel= an \,,
\end{align}
so the last case from above applies: 
The total cost of the child subproblems is (asymptotically) the same as the cost of
the initial call.
In analogy with the classical master theorem,
the overall costs $\E[C_n]$ are thus the toll cost of the initial call times the
number of levels in the recursion tree.

\subsubsection{Solve by Ansatz}
\def\ansatzconstant{\eta}
Guessing that the number of recursion levels will be logarithmic as in the
case of the classical master theorem, 
we make the \textsl{ansatz} $\E[C_n] = \frac{a}{\ansatzconstant} n\ln n$ 
with an unknown constant $\ansatzconstant$.
Inserting into the \wref[continuous recurrence]{eq:continuous-recurrence} yields
$$
		\frac{a}{\ansatzconstant}n\ln n
	\wwrel{=}
		a n + \int_{0}^{1}\!\! w(z) \, \frac{a}{\ansatzconstant}zn\ln(zn)\: dz
		\;.
$$
Multiplying by $\frac{\ansatzconstant}{a n}$ and rearranging, we find
$$
		\ansatzconstant
	\wwrel=
		\ln n \cdot\Bigl({\textstyle 1-\int_{0}^{1}z w(z)\, dz}\Bigr)
		\wbin -\int_{0}^{1}\!\! z\ln(z)w(z)\, dz \,,
$$
where the first integral is $1$ (see \wref{eq:master-theorem-case1-integral}),
which is good since otherwise the ``constant'' $\ansatzconstant$ would 
involve $\ln n$.
The second integral turns out to be precisely $-\discreteEntropy$, 
for $\discreteEntropy=\discreteEntropy[\vect t]$ the discrete
entropy of $\vect t$ defined in \wref[Equation]{eq:discrete-entropy}
and so
$$
		\E[C_{n}]
	\wwrel =
		\frac{a}{\discreteEntropy} n \ln n
$$
fulfills the \wref[continuous recurrence]{eq:continuous-recurrence} exactly.

Working out the error terms that we get by approximating the sum
of the original recurrence by an integral and
by approximating the weights in the discrete recurrence 
by the shape function $w(z)$, we obtain the following theorem.

\begin{theorem}
\label{thm:leading-term-expectation}
	Let $\E[C_n]$ be a sequence of numbers satisfying
	\wildtpageref[Equation]{eq:ECn-recurrence}{\eqref} for 
	$\vect t \in \N^3$ and a constant $\isthreshold \ge k=t_1+t_2+t_3+2$ and 
	let the toll function $\E[T_n]$ be of the form 
	$\E[T_n] = an+\Oh(n^{1-\epsilon})$ for constants $a$ and $\epsilon>0$.
	Then we have
	\smash{$
			\E[C_n] 
		\sim 
			\frac{a}{\discreteEntropy} \, n \ln n
	$},
	where $\discreteEntropy$ is given by
	\wildtpageref[Equation]{eq:discrete-entropy}{\eqref}.
\end{theorem}

A slightly weaker form of \wref{thm:leading-term-expectation} has first 
been proven by \citet[Proposition~III.9]{hennequin1991analyse} using direct 
arguments on the Cauchy-Euler differential equations that the recurrence
implies for the generating function of $\E[C_n]$. 
Building on the toolbox of handy and ready-to-apply theorems developed by
the analysis-of-algorithms community,
we can give a rather concise and elementary proof making our informal derivation
from above precise: 
\wref{app:CMT-solution} gives the detailed argument for solving the recurrence
using the \textsl{Continuous Master Theorem} by \citet{Roura2001}.
An alternative tool that remains closer to
\citeauthor{hennequin1991analyse}'s original arguments is offered by
\citet{Chern2002}.

\medskip\noindent
\wref{thm:expected-costs} now directly follows by using
\wref{lem:expectations} on the partitioning costs from
\wref{lem:distribution-partitioning-comparisons},
\ref{lem:distribution-partitioning-swaps}
and~\ref{lem:distribution-partitioning-bytecodes} and plugging the result into
\wref{thm:leading-term-expectation}.

\section{Validation}
\label{sec:validation}

The purpose of this paper is to approach an explanation for 
the efficiency of Yaroslavskiy's Quicksort in practice
using the methods of the mathematical analysis of algorithms,
which means that we define a \emph{model} of the actual program 
(given by our \wref{alg:generalized-yaroslavskiy}) and its \emph{costs}.
For the latter, different cost measures have proven valuable for 
different purposes, so we consider several of them.
As in the natural sciences, our model typically loses 
some details of the ``real world'', 
which means that we make a \textsl{modeling error}.
For example, counting scanned elements comes close to, but is not the
same as counting actual cache misses, see \wref{sec:validation-scans-vs-cache-misses}.

On top of that, the precise analysis of the model of an algorithm 
can still be infeasible or at least overly complicated.
For example in our recurrence \wref{eq:ECn-recurrence}, rather elementary means 
sufficed to determine the leading term of an asymptotic expansion of the solution;
obtaining more terms of the expansion is much harder, though.
Luckily, one can often resort to such asymptotic approximations for $n\to\infty$ 
without losing too much accuracy for practical input sizes;
yet we do make an \textsl{analysis error} whenever we use asymptotics,
see \wref{sec:validation-asymptotics}.

To assess the predictive quality of our analysis, 
we compare our results to some practical values.
Wherever possible, we try to separate modeling errors from analysis errors
to indicate whether further effort should be put in a more detailed 
analysis of the present model or in a refined model.

As discussed in \wref{sec:generalized-pivot-sampling-implementation}, 
\wref{alg:generalized-yaroslavskiy} should be considered an ``academic'' program, 
which is tailor-made for analysis, not for productive use and therefore,
we do not report running times. 
Other works contain actual running times of (more) realistic implementations:
\citet{Wild2012thesis} investigates the basic variants without pivot sampling.
\citet{Wild2013Alenex} compare different choices for the pivots from a sample of size $k=5$.
\citet{Aumuller2013icalp} compare several variants with and without pivot sampling 
and also other dual-pivot partitioning methods.
Moreover, \citet{Kushagra2014} include a three-pivot Quicksort and report 
measured cache misses as well (see also \wref{sec:validation-scans-vs-cache-misses}).

\subsection{Quality of Asymptotic Approximations}
\label{sec:validation-asymptotics}

In this section, we focus on the analysis error.
To obtain values to compare the asymptotic approximations with, 
we implemented \generalYarostM
(as given in \wref{alg:generalized-yaroslavskiy}) and augmented the code to
count key comparisons, swaps and scanned elements.
For counting the number of executed Java Bytecode instructions, we
used our tool \textsl{MaLiJAn}, which can automatically
generate code to count the number of Bytecodes \citep{Wild2013Alenex}.

All reported counts are averages of runs on $1\.000$ random 
permutations of the same size.
We use powers of 2 as input sizes and the plots show $n$
on a logarithmic $x$-axis. 
The $y$-axis is normalized by dividing by $n \ln n $.

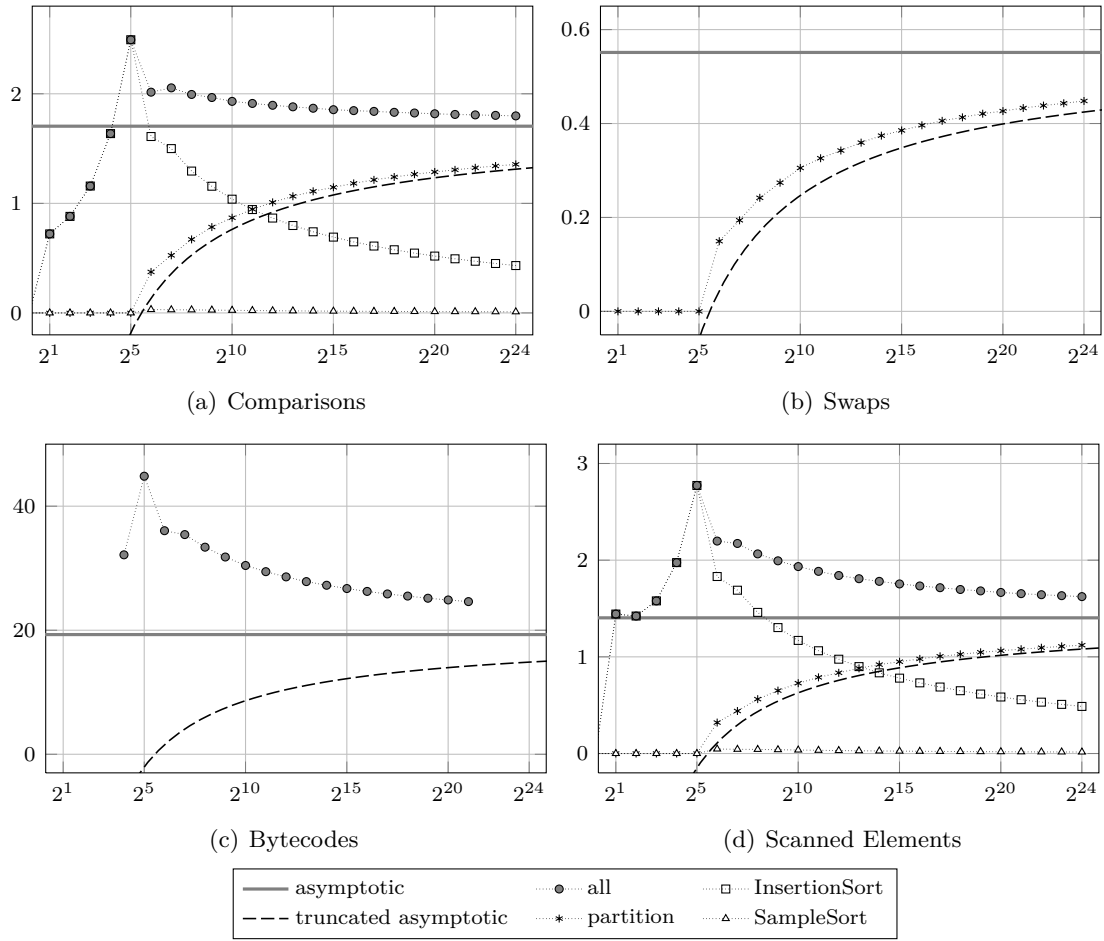
\begin{figure}
	\tikzset{every node/.style={font=\scriptsize},mark size=1.5pt}
	\pgfplotsset{every semilogx axis/.append style={
		width=.55\textwidth,
		height=.4\textwidth,
		xmin=1.1,xmax={3e7},
		xtickten={1,5,10,15,20,24},
		grid=both,
		cycle list name=validationplots,
		log basis x=2,
	}}
	\plaincenter{%
		\subfigure[Comparisons]{%
			\def\ltf{680/399}%
			\def\isco{46}%
			\def\costm{cmps}%
			\def\plotdatafile{pslt-pics/validation-asymptotic-\costm-t1-1-1-w\isco.tab}%
			\tikzset{external/export=true}\tikzsetnextfilename{validation-asymptotic-\costm-t1-1-1-w\isco}%
			\begin{tikzpicture}[]
			\begin{semilogxaxis}[
				ymin=-0.2,ymax=2.8,
				legend to name=validation-asymptotic-legend,
				every axis legend/.append style={
					legend columns=2,
					transpose legend=true,
					legend cell align=left,
				},
			]
				\addplot+[domain=0.1:1e8] expression {\ltf} ;
				\label{plot:validation-asymptotic-leading-term}
				\addlegendentry{asymptotic}
				\addplot+[domain=2:1e8,samples=50] expression {\ltf*( (ln(x) - ln(\isco)) / ln(x) )} ;
				\label{plot:validation-asymptotic-leading-term-cutoff}
				\addlegendentry{truncated asymptotic\quad\mbox{}}
				\addplot table[x=n,y=all-\costm-norm] {\plotdatafile} ;
				\label{plot:validation-asymptotic-all}
				\addlegendentry{all}
				\addplot table[x=n,y=partition-\costm-norm] {\plotdatafile} ;
				\label{plot:validation-asymptotic-partition}
				\addlegendentry{partition\quad\mbox{}}
				\addplot table[x=n,y=insertionsort-\costm-norm] {\plotdatafile} ;
				\label{plot:validation-asymptotic-insertionsort}
				\addlegendentry{InsertionSort}
				\addplot table[x=n,y=samplesort-\costm-norm] {\plotdatafile} ;
				\label{plot:validation-asymptotic-samplesort}
				\addlegendentry{SampleSort}
			\end{semilogxaxis}
			\end{tikzpicture}
		}%
		\hfill%
		\subfigure[Swaps]{%
			\def\ltf{220/399}%
			\def\isco{46}%
			\def\costm{swaps}%
			\def\plotdatafile{pslt-pics/validation-asymptotic-\costm-t1-1-1-w\isco.tab}%
			\tikzset{external/export=true}\tikzsetnextfilename{validation-asymptotic-\costm-t1-1-1-w\isco}%
			\begin{tikzpicture}[]
			\begin{semilogxaxis}[
				ymin=-0.05,ymax=0.65,
			]
				\addplot+[domain=0.1:1e8] expression {\ltf} ;
				\addplot+[domain=2:1e8,samples=50] expression {\ltf*( (ln(x) - ln(\isco)) / ln(x) )} ;
				\addplot coordinates {} ;
				\addplot table[x=n,y=partition-\costm-norm] {\plotdatafile} ;
				\addplot coordinates {} ;
				\addplot coordinates {} ;
			\end{semilogxaxis}
			\end{tikzpicture}
		}
	}\\
	\plaincenter{%
		\subfigure[Bytecodes]{%
			\def\ltf{1100/57}%
			\def\isco{46}%
			\def\costm{bytecodes}%
			\def\plotdatafile{pslt-pics/validation-asymptotic-\costm-t1-1-1-w\isco.tab}%
			\tikzset{external/export=true}\tikzsetnextfilename{validation-asymptotic-\costm-t1-1-1-w\isco}%
			\begin{tikzpicture}[]
			\begin{semilogxaxis}[
				ymin=-3,ymax=50,
			]
				\addplot+[domain=0.1:1e8] expression {\ltf} ;
				\addplot+[domain=2:1e8,samples=50] expression {\ltf*( (ln(x) - ln(\isco)) / ln(x) )} ;
				\addplot table[x=n,y=all-\costm-norm] {\plotdatafile} ;
			\end{semilogxaxis}
			\end{tikzpicture}
		}%
		\hfill%
		\subfigure[Scanned Elements]{%
			\def\ltf{80/57}%
			\def\isco{46}%
			\def\costm{scans}%
			\def\plotdatafile{pslt-pics/validation-asymptotic-\costm-t1-1-1-w\isco.tab}%
			\tikzset{external/export=true}\tikzsetnextfilename{validation-asymptotic-\costm-t1-1-1-w\isco}%
			\begin{tikzpicture}[]
			\begin{semilogxaxis}[
				ymin=-0.2,ymax=3.2,
			]
				\addplot+[domain=0.1:1e8] expression {\ltf} ;
				\addplot+[domain=2:1e8,samples=50] expression {\ltf*( (ln(x) - ln(\isco)) / ln(x) )} ;
				\addplot table[x=n,y=all-\costm-norm] {\plotdatafile} ;
				\addplot table[x=n,y=partition-\costm-norm] {\plotdatafile} ;
				\addplot table[x=n,y=insertionsort-\costm-norm] {\plotdatafile} ;
				\addplot table[x=n,y=samplesort-\costm-norm] {\plotdatafile} ;
			\end{semilogxaxis}
			\end{tikzpicture}
		}
	}\\
	\plaincenter{\makebox[\linewidth][c]{\ref{validation-asymptotic-legend}}}
	\caption{%
		Comparison, swap, Bytecode and scanned element counts 
		(\ref{plot:validation-asymptotic-all})
		normalized by $n\ln n$, for \generalYarostM with
		$\vect t=(1,1,1)$ and $\isthreshold=46$
		against the leading-term asymptotic $\frac a{\discreteEntropy} n \ln(n)$
		(\ref{plot:validation-asymptotic-leading-term}) from \wref{thm:expected-costs}
		and its truncated version $\frac a{\discreteEntropy} n \ln(\frac n\isthreshold)$
		(\ref{plot:validation-asymptotic-leading-term-cutoff}).
		For comparisons and scanned elements, the contributions from
		\proc{Partition} (\ref{plot:validation-asymptotic-partition}), 
		\proc{InsertionSort} (\ref{plot:validation-asymptotic-insertionsort}) and 
		\proc{SampleSort} (\ref{plot:validation-asymptotic-samplesort}) 
		are also given separately.
		Note that swaps only occur during partitioning 
		(Insertionsort uses single write accesses).
		For reasonably large $n$, the main contribution indeed
		comes from \proc{Partition}, however, \proc{InsertionSort} on short subarrays
		also contributes significantly.
		This is probably true for all cost measures, 
		even though not shown here in detail.
	}
	\label{fig:analysis-error-t111-w46}
\end{figure}

\begin{figure}
	\tikzset{every node/.style={font=\scriptsize},mark size=1.5pt}
	\pgfplotsset{every semilogx axis/.append style={
		width=.55\textwidth,
		height=.4\textwidth,
		xmin=8,xmax={7.5e6},
		xtickten={5,10,15,20},
		grid=both,
		cycle list name=validationplots,
		log basis x=2,
	}}
	\plaincenter{%
		\subfigure[Comparisons]{%
			\def\ltf{680/399}%
			\def\isco{7}%
			\def\costm{cmps}
			\def\plotdatafile{pslt-pics/validation-asymptotic-\costm-t1-1-1-w\isco.tab}%
			\tikzset{external/export=true}\tikzsetnextfilename{validation-asymptotic-\costm-t1-1-1-w\isco}%
			\begin{tikzpicture}[]
			\begin{semilogxaxis}[
				ymin=-0.17,ymax=2.2,
				legend to name=validation-asymptotic-legend-w7,
				every axis legend/.append style={
					legend columns=2,
					transpose legend=true,
					legend cell align=left,
				},
			]
				\addplot+[domain=0.1:1e8] expression {\ltf} ;
				\addlegendentry{asymptotic}
				\addplot+[domain=2:1e8,samples=50] expression {\ltf*( (ln(x) - ln(\isco)) / ln(x) )} ;
				\addlegendentry{asymptotic with cutoff\quad\mbox{}}
				\addplot table[x=n,y=all-cmps-norm] {\plotdatafile} ;
				\addlegendentry{all}
				\addplot table[x=n,y=partition-cmps-norm] {\plotdatafile} ;
				\addlegendentry{partition\quad\mbox{}}
				\addplot table[x=n,y=insertionsort-cmps-norm] {\plotdatafile} ;
				\addlegendentry{InsertionSort}
				\addplot table[x=n,y=samplesort-cmps-norm] {\plotdatafile} ;
				\addlegendentry{SampleSort}
			\end{semilogxaxis}
			\end{tikzpicture}
		}%
		\hfill%
		\subfigure[Swaps]{%
			\def\ltf{220/399}%
			\def\isco{7}%
			\def\costm{swaps}%
			\def\plotdatafile{pslt-pics/validation-asymptotic-\costm-t1-1-1-w\isco.tab}%
			\tikzset{external/export=true}\tikzsetnextfilename{validation-asymptotic-\costm-t1-1-1-w\isco}%
			\begin{tikzpicture}[]
			\begin{semilogxaxis}[
				ymin=-0.05,ymax=0.65,
			]
				\addplot+[domain=0.1:1e8] expression {\ltf} ;
				\addplot+[domain=2:1e8,samples=50] expression {\ltf*( (ln(x) - ln(\isco)) / ln(x) )} ;
				\addplot coordinates {} ;
				\addplot table[x=n,y=partition-\costm-norm] {\plotdatafile} ;
				\addplot coordinates {} ;
				\addplot coordinates {} ;
			\end{semilogxaxis}
			\end{tikzpicture}
		}
	}\\
	\plaincenter{%
		\subfigure[Bytecodes]{%
			\def\ltf{1100/57}%
			\def\isco{7}%
			\def\costm{bytecodes}%
			\def\plotdatafile{pslt-pics/validation-asymptotic-\costm-t1-1-1-w\isco.tab}%
			\tikzset{external/export=true}\tikzsetnextfilename{validation-asymptotic-\costm-t1-1-1-w\isco}%
			\begin{tikzpicture}[]
			\begin{semilogxaxis}[
				ymin=-2,ymax=40,
			]
				\addplot+[domain=0.1:1e8] expression {\ltf} ;
				\addplot+[domain=2:1e8,samples=50] expression {\ltf*( (ln(x) - ln(\isco)) / ln(x) )} ;
				\addplot table[x=n,y=all-\costm-norm] {\plotdatafile} ;
			\end{semilogxaxis}
			\end{tikzpicture}
		}%
		\hfill%
		\subfigure[Scanned Elements]{%
			\def\ltf{80/57}%
			\def\isco{7}%
			\def\costm{scans}%
			\def\plotdatafile{pslt-pics/validation-asymptotic-\costm-t1-1-1-w\isco.tab}%
			\tikzset{external/export=true}\tikzsetnextfilename{validation-asymptotic-\costm-t1-1-1-w\isco}%
			\begin{tikzpicture}[]
			\begin{semilogxaxis}[
				ymin=-0.15,ymax=2.15,
			]
				\addplot+[domain=0.1:1e8] expression {\ltf} ;
				\addplot+[domain=2:1e8,samples=50] expression {\ltf*( (ln(x) - ln(\isco)) / ln(x) )} ;
				\addplot table[x=n,y=all-\costm-norm] {\plotdatafile} ;
				\addplot table[x=n,y=partition-\costm-norm] {\plotdatafile} ;
				\addplot table[x=n,y=insertionsort-\costm-norm] {\plotdatafile} ;
				\addplot table[x=n,y=samplesort-\costm-norm] {\plotdatafile} ;
			\end{semilogxaxis}
			\end{tikzpicture}
		}
	}
	\caption{%
		Same as \wref{fig:analysis-error-t111-w46}, but with smaller 
		Insertionsort threshold $\isthreshold = 7$.%
	}
	\label{fig:analysis-error-t111-w7}
\end{figure}

For an actual execution, one has to fix the parameters $\vect t$ and 
\isthreshold.
We experimented with several choices, but found the quality of the 
asymptotic expansions to be very stable w.\,r.\,t.\ moderate values of $\vect t$,
\ie, for sample sizes up to $k=11$.
Unless otherwise stated, all plots below show the \textsl{tertiles-of-five} choice
$\vect t = (1,1,1)$.
For the Insertionsort threshold \isthreshold, values used in practice
($\isthreshold = 46$ for Oracle's Java 7 library) 
yield a significant influence on overall costs for moderate $n$,
see \wref{fig:analysis-error-t111-w46}.
This contribution is completely ignored in the leading term,
and thus the predictive quality of the asymptotic is limited for
large values of \isthreshold.
For $\isthreshold = 7$, the analysis error is much smaller, 
but still clearly visible, see \wref{fig:analysis-error-t111-w7}.

In plain numbers, we have with $\isthreshold=46$ and input size 
$n=2^{20}\approx 10^6$ 
around $5\.\%$ error for comparisons, $28\.\%$ error in the number
of swaps, $23\.\%$ for Bytecodes and $16\.\%$ error for scanned elements.
For $\isthreshold=7$, the errors are $9\.\%$, $6\.\%$, $15\.\%$ and $1\.\%$ for
comparisons, swaps, Bytecodes and scanned elements, respectively.

Although a complete derivation of the linear term of costs is out
of the question here, a simple heuristic allows to improve 
the predictive quality of our asymptotic formulas
for the partitioning costs.
The main error that we make is to ignore that \proc{Partition} 
is not called at all for subarrays of size at most \isthreshold.
We can partially correct for that by truncating the recursion tree at
level $\ln(\frac n\isthreshold)$, instead of going down all $\ln(n)$ levels,
\ie, instead of total costs $\frac a{\discreteEntropy}\,n \ln n$, we
use the \textsl{truncated term} 
$\frac a{\discreteEntropy}\,n \ln(\frac n\isthreshold)$.
(This means that the last $\ln(\isthreshold)$ levels of the recursion tree
are subtracted from the leading term.)
The plots in this section always include the pure leading term as a straight black
line and the truncated term as a dashed black line.
It is clearly visible that the truncated term gives a much better approximation
of the costs from \proc{Partition}.

Of course, the above argument is informal reasoning on an oversimplified 
view of the recurrence;
the actual recursion tree does neither have exactly $\ln(n)$ levels,
nor are all levels completely filled.
Therefore, the truncated term does not give the correct linear term for 
partitioning costs, and it completely ignores the costs of sorting the 
short subarrays by Insertionsort.
It is thus to be expected that the truncated term is \emph{smaller} 
than the actual costs, 
whereas the leading term alone often lies above them.

\subsection{Scanned Elements vs.\ Cache Misses}
\label{sec:validation-scans-vs-cache-misses}

This section considers the modeling error present in our cost measures.
Comparisons, swaps and Bytecodes are precise by definition; they stand for
themselves and do not model more intricate practical costs.
(They were initially intended as models for running time, 
but as discussed in the introduction were already shown to fail in explaining 
observed running time differences.)
The number of scanned elements was introduced in this paper as a model for
the number of cache misses in Quicksort, so we ought to investigate 
the difference between the two.

The problem with cache misses is that in practice there are multiple levels of
caches and that cache sizes, block sizes, eviction strategies and associativity all
differ from machine to machine.
Moreover, block borders in a hardware cache are aligned with physical address blocks 
(such that one can use the first few bits as cache block address),
so the precise caching behavior depends on the starting address of the array that
we are sorting;
not to speak of the influence other processes have on the content of the cache \dots

We claim, though, that such details do not have a big impact on the overall number
of cache misses in Quicksort and
focus in this paper on an \emph{idealized cache}, \ie,
a fully associative cache (\ie, no block address aliasing) that uses the
\textsl{least-recently-used (LRU)} eviction strategy.
The cache synchronizes itself with main memory in blocks of $B$ consecutive array elements
and it can hold up to $M$ array entries in total, where $M\ge B$ is a multiple of $B$.
Moreover, we assume that our array always starts at a block boundary, that
its length is a multiple of the block size and that the cache is initially empty.
We then simulated Quicksort on such an idealized cache, precisely counting the number
of incurred cache misses, \ie, of accesses to indices of the array, 
whose block is currently not in the cache.

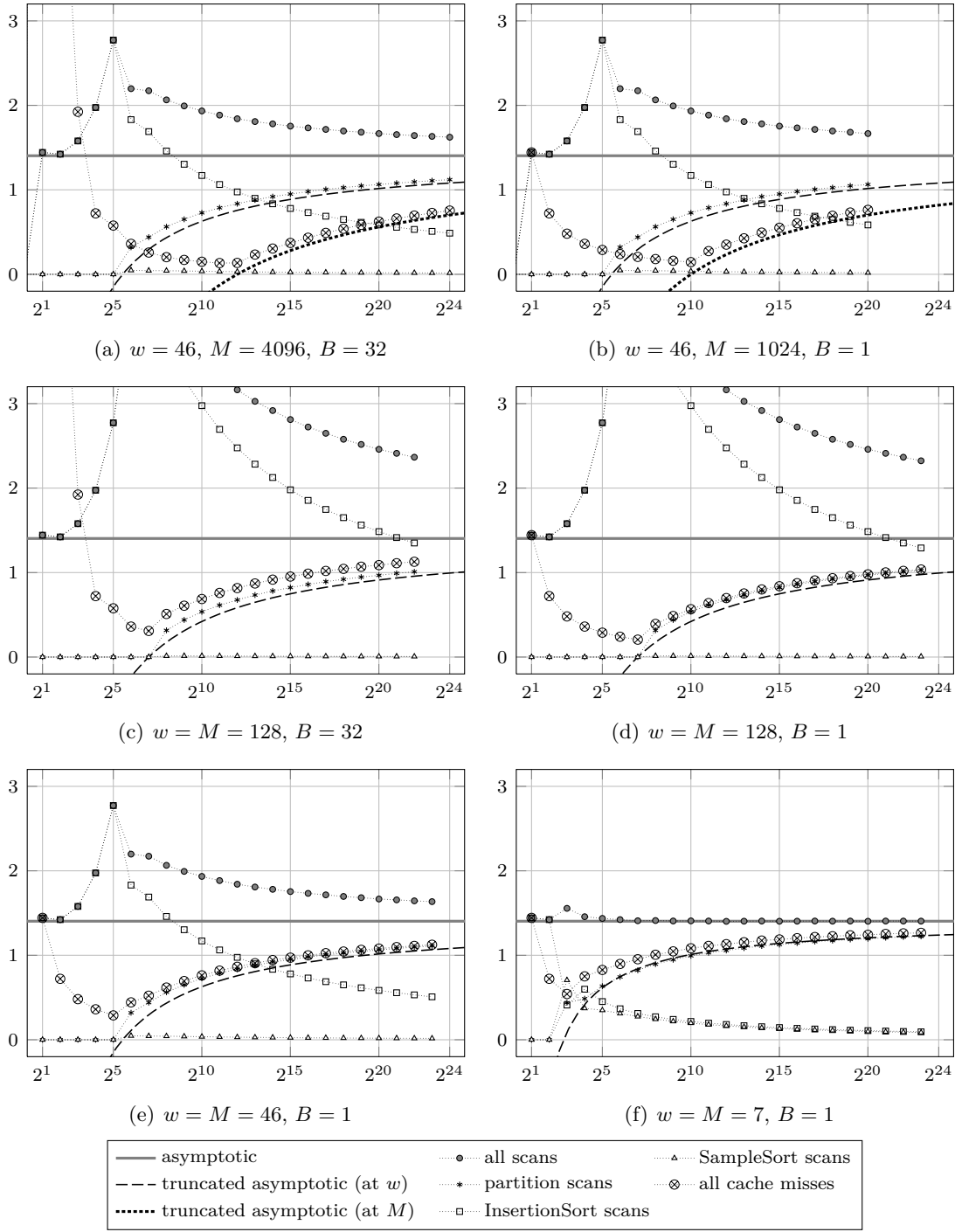
\begin{figure}
	\tikzset{every node/.style={font=\scriptsize},mark size=1.25pt}
	\pgfplotsset{every semilogx axis/.append style={
		width=.55\textwidth,
		height=.4\textwidth,
		xmin=1.1,xmax={3e7},
		ymin=-0.2,ymax=3.2,
		xtickten={1,5,10,15,20,24},
		grid=both,
		cycle list name=validationplotsCM,
		log basis x=2,
	}}
	\def\ltf{80/57}%
	\plaincenter{%
		\subfigure[$\isthreshold=46$, $M=4\.096$, $B=32$]{%
			\def\isco{46}
			\def\csib{128}
			\def\bs{32}
			\def\cs{4098}
			\def\plotdatafile{pslt-pics/validation-scans-vs-cache-misses-t1-1-1-w\isco-M\csib-B\bs.tab}%
			\tikzset{external/export=true}\tikzsetnextfilename{validation-scans-vs-cache-misses-t1-1-1-w\isco-M\csib-B\bs}%
			\begin{tikzpicture}
			\begin{semilogxaxis}[
				legend to name=validation-scans-vs-cache-misses-legend,
				every axis legend/.append style={
					legend columns=3,
					transpose legend=true,
					legend cell align=left,
				},
			]
				\addplot+[domain=0.1:1e8] expression {\ltf} ;
				\label{plot:validation-scans-vs-cache-misses-leading-term}
				\addlegendentry{asymptotic}
				\addplot+[domain=2:1e8,samples=50] expression {\ltf*( (ln(x) - ln(\isco)) / ln(x) )} ;
				\label{plot:validation-scans-vs-cache-misses-leading-term-cutoff-w}
				\addlegendentry{truncated asymptotic (at \isthreshold)\quad\mbox{}}
				\addplot+[domain=2:1e8,samples=50] expression {\ltf*( (ln(x) - ln(\cs)) / ln(x) )} ;
				\label{plot:validation-scans-vs-cache-misses-leading-term-cutoff-M}
				\addlegendentry{truncated asymptotic (at $M$)\quad\mbox{}}
				\addplot table[x=n,y=all-scans-norm] {\plotdatafile} ;
				\label{plot:validation-scans-vs-cache-misses-all-scans}
				\addlegendentry{all scans}
				\addplot table[x=n,y=partition-scans-norm] {\plotdatafile} ;
				\label{plot:validation-scans-vs-cache-misses-partition-scans}
				\addlegendentry{partition scans\quad\mbox{}}
				\addplot table[x=n,y=insertionsort-scans-norm] {\plotdatafile} ;
				\label{plot:validation-scans-vs-cache-misses-insertionsort-scans}
				\addlegendentry{InsertionSort scans}
				\addplot table[x=n,y=samplesort-scans-norm] {\plotdatafile} ;
				\label{plot:validation-scans-vs-cache-misses-samplesort-scans}
				\addlegendentry{SampleSort scans}
				\addplot table[x=n,y=all-cache-misses-norm] {\plotdatafile} ;
				\label{plot:validation-scans-vs-cache-misses-all-cache-misses}
				\addlegendentry{all cache misses}
			\end{semilogxaxis}
			\end{tikzpicture}
			\label{fig:validation-scans-vs-cache-misses-w46-M4096-B32}%
		}%
		\hfill%
		\subfigure[$\isthreshold=46$, $M=1\.024$, $B=1$]{%
			\def\isco{46}
			\def\csib{1024}
			\def\bs{1}
			\def\cs{1024}
			\def\plotdatafile{pslt-pics/validation-scans-vs-cache-misses-t1-1-1-w\isco-M\csib-B\bs.tab}%
			\tikzset{external/export=true}\tikzsetnextfilename{validation-scans-vs-cache-misses-t1-1-1-w\isco-M\csib-B\bs}%
			\begin{tikzpicture}
			\begin{semilogxaxis}[
			]
				\addplot+[domain=0.1:1e8] expression {\ltf} ;
				\addplot+[domain=2:1e8,samples=50] expression {\ltf*( (ln(x) - ln(\isco)) / ln(x) )} ;
				\addplot+[domain=2:1e8,samples=50] expression {\ltf*( (ln(x) - ln(\cs)) / ln(x) )} ;
				\addplot table[x=n,y=all-scans-norm] {\plotdatafile} ;
				\addplot table[x=n,y=partition-scans-norm] {\plotdatafile} ;
				\addplot table[x=n,y=insertionsort-scans-norm] {\plotdatafile} ;
				\addplot table[x=n,y=samplesort-scans-norm] {\plotdatafile} ;
				\addplot table[x=n,y=all-cache-misses-norm] {\plotdatafile} ;
			\end{semilogxaxis}
			\end{tikzpicture}
			\label{fig:validation-scans-vs-cache-misses-w46-M1024-B1}%
		}%
	}\\
	\plaincenter{%
		\subfigure[$\isthreshold=M=128$, $B=32$]{%
			\def\isco{128}
			\def\csib{4}
			\def\bs{32}
			\def\cs{128}
			\def\plotdatafile{pslt-pics/validation-scans-vs-cache-misses-t1-1-1-w\isco-M\csib-B\bs.tab}%
			\tikzset{external/export=true}\tikzsetnextfilename{validation-scans-vs-cache-misses-t1-1-1-w\isco-M\csib-B\bs}%
			\begin{tikzpicture}
			\begin{semilogxaxis}[
			]
				\addplot+[domain=0.1:1e8] expression {\ltf} ;
				\addplot+[domain=2:1e8,samples=50] expression {\ltf*( (ln(x) - ln(\isco)) / ln(x) )} ;
				\addplot+ coordinates {} ;
				\addplot table[x=n,y=all-scans-norm] {\plotdatafile} ;
				\addplot table[x=n,y=partition-scans-norm] {\plotdatafile} ;
				\addplot table[x=n,y=insertionsort-scans-norm] {\plotdatafile} ;
				\addplot table[x=n,y=samplesort-scans-norm] {\plotdatafile} ;
				\addplot table[x=n,y=all-cache-misses-norm] {\plotdatafile} ;
			\end{semilogxaxis}
			\end{tikzpicture}
			\label{fig:validation-scans-vs-cache-misses-w128-M128-B32}%
		}%
		\hfill%
		\subfigure[$\isthreshold=M=128$, $B=1$]{%
			\def\isco{128}
			\def\csib{128}
			\def\bs{1}
			\def\cs{128}
			\def\plotdatafile{pslt-pics/validation-scans-vs-cache-misses-t1-1-1-w\isco-M\csib-B\bs.tab}%
			\tikzset{external/export=true}\tikzsetnextfilename{validation-scans-vs-cache-misses-t1-1-1-w\isco-M\csib-B\bs}%
			\begin{tikzpicture}
			\begin{semilogxaxis}
				\addplot+[domain=0.1:1e8] expression {\ltf} ;
				\addplot+[domain=2:1e8,samples=50] expression {\ltf*( (ln(x) - ln(\isco)) / ln(x) )} ;
				\addplot+ coordinates {} ;
				\addplot+ table[x=n,y=all-scans-norm] {\plotdatafile} ;
				\addplot+ table[x=n,y=partition-scans-norm] {\plotdatafile} ;
				\addplot+ table[x=n,y=insertionsort-scans-norm] {\plotdatafile} ;
				\addplot+ table[x=n,y=samplesort-scans-norm] {\plotdatafile} ;
				\addplot+ table[x=n,y=all-cache-misses-norm] {\plotdatafile} ;
			\end{semilogxaxis}
			\end{tikzpicture}
			\label{fig:validation-scans-vs-cache-misses-w128-M128-B1}%
		}%
	}\\
	\plaincenter{%
		\subfigure[$\isthreshold=M=46$, $B=1$]{%
			\def\isco{46}
			\def\csib{46}
			\def\bs{1}
			\def\cs{46}
			\def\plotdatafile{pslt-pics/validation-scans-vs-cache-misses-t1-1-1-w\isco-M\csib-B\bs.tab}%
			\tikzset{external/export=true}\tikzsetnextfilename{validation-scans-vs-cache-misses-t1-1-1-w\isco-M\csib-B\bs}%
			\begin{tikzpicture}
			\begin{semilogxaxis}[
			]
				\addplot+[domain=0.1:1e8] expression {\ltf} ;
				\addplot+[domain=2:1e8,samples=50] expression {\ltf*( (ln(x) - ln(\isco)) / ln(x) )} ;
				\addplot+ coordinates {} ;
				\addplot table[x=n,y=all-scans-norm] {\plotdatafile} ;
				\addplot table[x=n,y=partition-scans-norm] {\plotdatafile} ;
				\addplot table[x=n,y=insertionsort-scans-norm] {\plotdatafile} ;
				\addplot table[x=n,y=samplesort-scans-norm] {\plotdatafile} ;
				\addplot table[x=n,y=all-cache-misses-norm] {\plotdatafile} ;
			\end{semilogxaxis}
			\end{tikzpicture}
			\label{fig:validation-scans-vs-cache-misses-w46-M46-B1}%
		}%
		\hfill%
		\subfigure[$\isthreshold=M=7$, $B=1$]{%
			\def\isco{7}
			\def\csib{7}
			\def\bs{1}
			\def\cs{7}
			\def\plotdatafile{pslt-pics/validation-scans-vs-cache-misses-t1-1-1-w\isco-M\csib-B\bs.tab}%
			\tikzset{external/export=true}\tikzsetnextfilename{validation-scans-vs-cache-misses-t1-1-1-w\isco-M\csib-B\bs}%
			\begin{tikzpicture}
			\begin{semilogxaxis}[
			]
				\addplot+[domain=0.1:1e8] expression {\ltf} ;
				\addplot+[domain=2:1e8,samples=50] expression {\ltf*( (ln(x) - ln(\isco)) / ln(x) )} ;
				\addplot+ coordinates {} ;
				\addplot table[x=n,y=all-scans-norm] {\plotdatafile} ;
				\addplot table[x=n,y=partition-scans-norm] {\plotdatafile} ;
				\addplot table[x=n,y=insertionsort-scans-norm] {\plotdatafile} ;
				\addplot table[x=n,y=samplesort-scans-norm] {\plotdatafile} ;
				\addplot table[x=n,y=all-cache-misses-norm] {\plotdatafile} ;
			\end{semilogxaxis}
			\end{tikzpicture}
			\label{fig:validation-scans-vs-cache-misses-w7-M7-B1}%
		}%
	}\\
	\plaincenter{\ref{validation-scans-vs-cache-misses-legend}}%
	\caption{%
		Comparison of cache miss counts 
		(\ref{plot:validation-scans-vs-cache-misses-all-cache-misses})
		from our idealized 
		fully-associative LRU cache with different cache and block 
		sizes $M$ resp.\ $B$ with corresponding scanned element counts
		(\ref{plot:validation-scans-vs-cache-misses-all-scans}).
		The latter are also given separately for \proc{Partition} 
		(\ref{plot:validation-scans-vs-cache-misses-partition-scans}),
		\proc{InsertionSort} 
		(\ref{plot:validation-scans-vs-cache-misses-insertionsort-scans}) and 
		\proc{SampleSort} 
		(\ref{plot:validation-scans-vs-cache-misses-samplesort-scans}).
		To make the counts comparable, the number of cache misses has been
		multiplied by $B$.
		All plots are normalized by $n\ln n$ and show results for 
		\generalYarostM with $\vect t=(1,1,1)$ and different Insertionsort
		thresholds \isthreshold.
		The fat line~(\ref{plot:validation-scans-vs-cache-misses-leading-term}) shows 
		the leading-term asymptotic for scanned elements from \wref{thm:expected-costs}, 
		namely \smash{$\frac{80}{57} n \ln n$}.
		The dashed line (\ref{plot:validation-scans-vs-cache-misses-leading-term-cutoff-w}) 
		is the truncated term $\frac{80}{57} n \ln(\frac n\isthreshold)$
		and the dotted line (\ref{plot:validation-scans-vs-cache-misses-leading-term-cutoff-M})
		shows $\frac{80}{57} n \ln(\frac nM)$,
		which is the leading term truncated at subproblems that fit into the cache.
	}
	\label{fig:validation-scans-vs-cache-misses}
\end{figure}

The resulting cache miss counts (averages of $1\.000$ runs) are shown
in \wref{fig:validation-scans-vs-cache-misses} for a variety of parameter choices.
At first sight, the overall picture seem rather disappointing:
the total number of scanned elements and the number of cache misses
do not seem to match particularly well 
(blue and violet dots in \wref{fig:validation-scans-vs-cache-misses}).
The reason is that once the subproblem size is at most $M$,
the whole subarray fits into the cache and at most $M/B$ 
additional cache misses suffice for sorting the whole subarray; 
whereas in terms of scanned elements, the contribution of these subarrays is
at least linearithmic%
\footnote{%
	We use the neologism ``linearithmic'' to say that a function has 
	order of growth $\Theta(n \log n)$.
}
 (for partitioning) or even quadratic (for Insertionsort).

If, however, the cache size $M$ and the Insertionsort threshold \isthreshold
are the same 
(as in \mbox{\wref{fig:validation-scans-vs-cache-misses-w128-M128-B32}{}\,--\,%
\subref{fig:validation-scans-vs-cache-misses-w7-M7-B1}}), 
the number of cache misses
and the number of scanned elements agree very well, 
if we count the latter in procedure \proc{Partition} only.
If we consider the asymptotic for the number of scanned elements, but
truncate the recursion to $\ln(\frac nM)$ levels (cf.\ \wref{sec:validation-asymptotics}), 
we find a very good fit to the number of cache misses
(see dotted lines resp.\ dashed lines in \wref{fig:validation-scans-vs-cache-misses}).
From that we can conclude that 
(a) the main error made in counting scanned elements is to ignore the cutoff at 
$M$ and that
(b) the base cases (subproblems of size at most $M$) have little influence
and can be ignored for performance prediction.
We also note that $\frac{a_{\scans}}{\discreteEntropy} \frac nB \ln (\frac nM)$
is a very good approximation for the overall number of cache misses
for \emph{all} our parameter choices for $M$, $B$ and \isthreshold
(even if the number of blocks $M/B$ that fit in the cache at the same time
is as small as 4, see \wref{fig:validation-scans-vs-cache-misses-w128-M128-B32}).

The most important algorithmic conclusion from these findings is that
\textsl{we can safely use the number of scanned elements to compare different 
Quicksort variants};
the major part of the modeling error, that we make in doing so, 
will cancel out when comparing two algorithms.

\citet{Kushagra2014} immediately report the truncated term
as an asymptotic upper bound for the number of cache misses.
We think that it is worthwhile to have the clean separation 
between the mathematically precise analysis of scanned elements and the
machine-dependent cache misses in practice\,---\,%
we can now compare Quicksort variants in terms of scanned elements 
instead of actual cache misses, 
which is a much more convenient cost measure to deal with.

\section{Discussion}
\label{sec:discussion}

\subsection{Asymmetries Everywhere}
\label{sec:asymmetries-everywhere}

\begin{figure}%
	\ifarxiv{%
		\def\picxscale{1.95}%
		\def\picyscale{.8}%
		\def\bppsep{-2pt}%
		\def\legsep{4pt}%
		\def\legscale{1.3}%
		\def\bpscale{1}%
	}
	\newcommand\relativebarplot[2]{%
		\begin{tikzpicture}[
			xscale=\bpscale
		]
		\ifthenelse{ \lengthtest{#1pt<0pt} }{%
			\colorlet{barcolor}{green!80!black}%
			\def\nodepos{right}%
			\def\format##1##2|{{\bfseries\boldmath\color{green!60!black}$-$##2\%}}%
		}{
			\ifthenelse{ \lengthtest{#1pt>0pt} }{%
				\colorlet{barcolor}{red!80!black}%
				\def\nodepos{left}%
				\def\format##1##2|{\color{barcolor}$+$##1##2\%}%
			}{%
				\colorlet{barcolor}{white}%
				\def\nodepos{right}%
				\def\format##1##2|{}%
			}%
		}
		\useasboundingbox (-1,-.15) rectangle (1,.15) ;
		\ifthenelse{ \lengthtest{#1pt<#2pt} \AND \lengthtest{-#1pt<#2pt} }{
			\draw[thin,fill=barcolor] (0,.06) rectangle ($( #1 / #2 ,-.06)$) ;
			\node[\nodepos] at ($(0,0)$) {\tiny\format#1|}  ;
		}{
			\ifthenelse{ \lengthtest{#1pt>0pt} }{
				\draw[thin,fill=barcolor] 
					(0,-.06) -- ++(.8,0) coordinate (a) -- 
					++(.03,.12) coordinate (b) -- (0,.06) -- 
					++(0,-.12) -- cycle ;
				\draw[thin,fill=barcolor]
					(a) ++(.05,0) coordinate (c) --
					++(.03,.12) coordinate (d) -- (1.2,.06) --
					++(0,-.12) -- (c) -- cycle;
				\draw[thin,shorten >=-1pt,shorten <=-1pt] (a) -- (b) ; 
				\draw[thin,shorten >=-1pt,shorten <=-1pt] (c) -- (d) ;
				\node[\nodepos] at (0,0) {\tiny\format#1|} ;
			}{
				\draw[thin,fill=barcolor] 
					(0,-.06) -- ++(-.8,0) coordinate (a) -- 
					++(.03,.12) coordinate (b) -- (0,.06) -- 
					++(0,-.12) -- cycle ;
				\draw[thin,fill=barcolor]
					(a) ++(-.05,0) coordinate (c) --
					++(.03,.12) coordinate (d) -- (-1.2,.06) --
					++(0,-.12) -- (c) -- cycle;
				\draw[thin,shorten >=-1pt,shorten <=-1pt] (a) -- (b) ; 
				\draw[thin,shorten >=-1pt,shorten <=-1pt] (c) -- (d) ;
				\node[\nodepos] at (0,0) {\tiny\format#1|} ;
			}
		}
		\draw[very thick] (0,.15) -- (0,-.15) ; 
		\end{tikzpicture}%
	}%
	\newcommand\plotst[5]{%
		\node[rotate=0] at (#1,#2) {%
			\scalebox{0.7}{\parbox{2cm}{%
			\centering%
			\relativebarplot{#3}{30}\\[\bppsep]%
			\relativebarplot{#4}{10}\\[\bppsep]%
			\relativebarplot{#5}{10}%
			}}%
		};
	}%
	\plaincenter{%
	\begin{tikzpicture}[
		xscale=\picxscale,
		yscale=\picyscale,
		every node/.style={font=\scriptsize}
	]
	\fill[symmetriccolor] (1.5,1.5) rectangle ++(1,1) ;
	\plotst{0}{0}{68.7}{-17.7}{38.9}
	\plotst{0}{1}{32.5}{-8.84}{20.8}
	\plotst{0}{2}{18.8}{-1.36}{17.2}
	\plotst{0}{3}{15.0}{4.76}{20.4}
	\plotst{0}{4}{18.8}{9.52}{30.1}
	\plotst{0}{5}{32.5}{12.9}{49.7}
	\plotst{0}{6}{68.7}{15.0}{94.0}
	\plotst{1}{0}{32.5}{-11.6}{17.2}
	\plotst{1}{1}{11.4}{-4.76}{6.09}
	\plotst{1}{2}{3.86}{0.680}{4.57}
	\plotst{1}{3}{3.86}{4.76}{8.81}
	\plotst{1}{4}{11.4}{7.48}{19.7}
	\plotst{1}{5}{32.5}{8.84}{44.3}
	\plotst{2}{0}{18.8}{-8.16}{9.08}
	\plotst{2}{1}{3.86}{-3.40}{0.331}
	\plotst{2}{2}{0}{0}{0}
	\plotst{2}{3}{3.86}{2.04}{5.98}
	\plotst{2}{4}{18.8}{2.72}{22.0}
	\plotst{3}{0}{15.0}{-7.48}{6.37}
	\plotst{3}{1}{3.86}{-4.76}{-1.08}
	\plotst{3}{2}{3.86}{-3.40}{0.331}
	\plotst{3}{3}{15.0}{-3.40}{11.1}
	\plotst{4}{0}{18.8}{-9.52}{7.47}
	\plotst{4}{1}{11.4}{-8.84}{1.54}
	\plotst{4}{2}{18.8}{-9.52}{7.47}
	\plotst{5}{0}{32.5}{-14.3}{13.6}
	\plotst{5}{1}{32.5}{-15.6}{11.8}
	\plotst{6}{0}{68.7}{-21.8}{32.0}
	\foreach \t in {0,...,6}{
		\node at (\t,-.9) {$t_1=\t$} ;
		\node at (-.75,\t) {$t_2=\t$} ;
	}
	\begin{scope}[overlay]
		\draw (-1,-0.5) -- ++(7.5,0) ;
		\draw (-1,0.5) -- ++(7.5,0) -- ++(0,-1.7) ;
		\draw (-1,1.5) -- ++(6.5,0) -- ++(0,-2.7) ;
		\draw (-1,2.5) -- ++(5.5,0) -- ++(0,-3.7) ;
		\draw (-1,3.5) -- ++(4.5,0) -- ++(0,-4.7) ;
		\draw (-1,4.5) -- ++(3.5,0) -- ++(0,-5.7) ;
		\draw (-1,5.5) -- ++(2.5,0) -- ++(0,-6.7) ;
		\draw (-1,6.5) -- ++(1.5,0) -- ++(0,-7.7) ;
		\draw            (-0.5,6.5) -- ++(0,-7.7) ;
	\end{scope}
	\draw[ultra thick] (2.5,0.5) rectangle ++(1,1) ;
	\begin{scope}[overlay]
			\node at (5.5,5) {%
			\scalebox{.9}{%
			\parbox{1.3cm}{\raggedleft%
				$1/\discreteEntropy$:\\
				$a_C$:\\
				$a_C / \discreteEntropy$:
			}}\hspace{\legsep}%
			\scalebox{\legscale}{%
			\parbox{2cm}{%
			\centering%
				\relativebarplot{15.0}{30}\\[-4pt]%
				\relativebarplot{-7.48}{10}\\[-4pt]%
				\relativebarplot{6.37}{10}%
			}}%
		};
	\end{scope}
	\end{tikzpicture}%
	}
	\caption{%
		Inverse of discrete entropy (top), number of comparisons per partitioning step
		(middle) and overall comparisons (bottom)
		for all $\vect t$ with $k=8$, relative to the tertiles case $\vect t =
		(2,2,2)$. 
	}
	\label{fig:relative-cmps-8}
\end{figure}
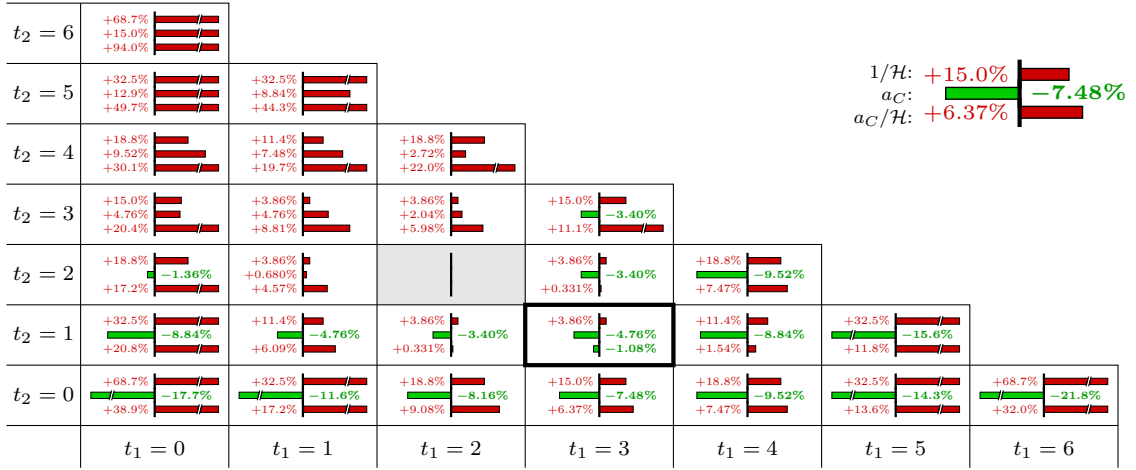

\begin{table}
	\plaincenter{%
	\footnotesize%
	\setlength\tabcolsep{0.25em}%
	\subtable[$a_C / \discreteEntropy$]{%
		\begin{tabular}{c|cccc}
			${}_{t_{1}\!\!\!}\diagdown{}^{\!\! t_{2}}$
				& 0 & 1 & 2 & 3 \\
			\hline 
			0 & 1.9956 & 1.8681 & 2.0055 & 2.4864 \\
			1 & 1.7582 & \textbf{1.7043} \cellcolor{symmetriccolor}
			                                & 1.9231 \\
			2 & 1.7308 & 1.7582 &  \\
			3 & 1.8975 & \\
		\end{tabular}%
	}%
	\hfill%
	\subtable[$a_S / \discreteEntropy$]{%
	\begin{tabular}{c|cccc}
			${}_{t_{1}\!\!\!}\diagdown{}^{\!\! t_{2}}$
				& 0 & 1 & 2 & 3 \\
			\hline 
			0 & 0.4907 & 0.4396 & 0.4121 & \textbf{0.3926} \\
			1 & 0.6319 & 0.5514 \cellcolor{symmetriccolor} & 0.5220 \\
			2 & 0.7967 & 0.7143 \\
			3 & 1.0796 \\
		\end{tabular}%
	}%
	}\\
	\plaincenter{%
	\subtable[$a_{\bytecodes} / \discreteEntropy$]{%
		\begin{tabular}{c|cccc}
			${}_{t_{1}\!\!\!}\diagdown{}^{\!\! t_{2}}$
				& 0 & 1 & 2 & 3 \\
			\hline 
			0 & 20.840 & \textbf{18.791} & 19.478 & 23.293 \\
			1 & 20.440 & 19.298 \cellcolor{symmetriccolor}& 21.264 \\
			2 & 22.830 & 22.967 \\
			3 & 29.378 \\
		\end{tabular}%
	}%
	\hfill%
	\subtable[$a_{\scans} / \discreteEntropy$]{%
		\begin{tabular}{c|cccc}
			${}_{t_{1}\!\!\!}\diagdown{}^{\!\! t_{2}}$
				& 0 & 1 & 2 & 3 \\
			\hline 
			0 & 1.6031 & \textbf{1.3462} & \textbf{1.3462} & 1.6031 \\
			1 & 1.5385 & 1.4035 \cellcolor{symmetriccolor}& 1.5385 \\
			2 & 1.7308 & 1.7308 \\
			3 & 2.2901 \\
		\end{tabular}%
	}%
	}
	\caption{%
		$\frac{a_C}{\discreteEntropy}$, 
		$\frac{a_S}{\discreteEntropy}$,
		$\frac{a_{\bytecodes}}{\discreteEntropy}$ and
		$\frac{a_{\scans}}{\discreteEntropy}$ 
		for all $\vect t$ with $k=5$.
		Rows resp.\ columns give $t_1$ and $t_2$; $t_3$ is then $k-2-t_1-t_2$.
		The symmetric choice $\vect t = (1,1,1)$ is shaded, the minimum is printed in
		bold.%
	}
	\label{tab:results-k5}
\end{table}

With \wref{thm:expected-costs}, we can find the optimal sampling
parameter $\vect t$ for any given sample size~$k$.
As an example, \wref{fig:relative-cmps-8} shows $\discreteEntropy$, $a_C$
and the overall number of comparisons for all possible $\vect t$ with
sample size $k=8$:
The discrete entropy decreases symmetrically as we move away
from the center $\vect t = (2,2,2)$; this corresponds to the effect of less
evenly distributed subproblem sizes. 
The individual partitioning steps, however, are cheap for \emph{small} values of
$t_2$ and optimal in the extreme point $\vect t = (6,0,0)$.
For minimizing the \emph{overall} number of comparisons\,---\,the ratio of
latter\,---\,we have to find a suitable trade-off between the
center and the extreme point $(6,0,0)$; in this case the minimal total number of
comparisons is achieved with $\vect t = (3,1,2)$.

Apart from this trade-off between the evenness of subproblem sizes and the
number of comparisons per partitioning, 
\wref{tab:results-k5} shows that the optimal choices for $\vect t$ w.\,r.\,t.\
comparisons, swaps, Bytecodes and scanned elements heavily differ.
The partitioning costs are, in fact, in \emph{extreme conflict} with each other:
for all $k\ge 2$, the minimal values of $a_C$, $a_S$ and $a_{\bytecodes}$ among
all choices of $\vect t $ for sample size $k$ are attained for
$\vect t = (k-2,0,0)$, 
$\vect t = (0,k-2,0)$,
$\vect t = (0,0,k-2)$ and 
$\vect t = (0,t,k-2-t)$ for $0\le t\le k-2$,
respectively.
Intuitively this is because the strategy minimizing partitioning costs in
isolation executes the cheapest path through the partitioning loop as
often as possible, which naturally leads to extreme choices for $\vect t$.
It then depends on the actual numbers, where the total
costs are minimized.
It is thus not possible to minimize all cost measures at
once, and the rivaling effects described above make it hard to reason
about optimal parameters merely on a qualitative level.

\subsection{Optimal Order Statistics for fixed $k$}

Given any cost measure we can
compute\,---\,although not in closed form\,---\,the optimal sampling parameter 
$\vect{t^\ast}$ for a fixed size of the sample $k=k(\vect{t})$. Here, 
by optimal sampling parameter we mean the parameter 
$\vect{t^\ast}=(t_1^\ast,t_2^\ast,t_3^\ast)$ 
that minimizes the leading term of
the corresponding cost, that is, the choice minimizing 
$q_X \ce a_X/\discreteEntropy$ (where $X$ is $C$, $S$, $\bytecodes$, or $\scans$).
Table~\ref{tab:optimal-finite-k} lists the optimal sampling parameters
of \generalYarostM for several values of $k$ of the form $k=3\lambda+2$ (as well as $k=100$).

\begin{table}
	\def\comps{comparisons}
	\def\swaps{swaps}
	\def\bytec{Bytecodes}
	\def\se{scanned elements}
	\plaincenter{%
	\begin{tabular}{crcD..{6}}
	\toprule
	$k$ & Cost measure & $\vect{t}^*$ 
	    & \multicolumn1c{$q_X=\frac{a_X}{\discreteEntropy}$} \\\midrule
	\multirow4*{no samp\rlap{ling}}   
	    & \comps         & {\color{black!50}(0,0,0)}        & 1.9 \\
	    & \swaps        & {\color{black!50}(0,0,0)}        & 0.6 \\
	    & \bytec           & {\color{black!50}(0,0,0)}        & 21.7 \\
	    & \se           & {\color{black!50}(0,0,0)}        & 1.6 \\\midrule
	\multirow4*{5}   
	    & \comps         & (1,1,1)        & 1.70426\\
	    & \swaps        & (0,3,0)        & 0.392585\\
	    & \bytec           & (0,1,2)        & 18.7912\\
	    & \se           & (0,1,2)        & 1.34615 \\\midrule
	\multirow4*{8}   
	    & \comps         & (3,1,2)        & 1.62274\\
	    & \swaps        & (0,6,0)        & 0.338937\\
	    & \bytec           & (1,2,3)        & 17.8733\\
	    & \se           & (1,2,3)        & 1.27501 \\\midrule
	\multirow4*{11}  
	    & \comps         & (4,2,3)        & 1.58485\\
	    & \swaps        & (0,9,0)        & 0.310338\\
	    & \bytec           & (2,3,4)        & 17.5552\\
	    & \se           & (1,4,4)        & 1.22751\\\midrule
	\multirow4*{17}  
	    & \comps         & (6,4,5)        & 1.55535\\
	    & \swaps        & (0,15,0)       & 0.277809\\
	    & \bytec           & (3,5,7)        & 17.1281\\
	    & \se           & (2,6,7)        & 1.19869\\\midrule
	\multirow4*{32}  
	    & \comps         & (13,8,9)       & 1.52583\\
	    & \swaps        & (0,30,0)       & 0.240074\\
	    & \bytec           & (6,10,14)      & 16.7888\\
	    & \se           & (5,12,13)      & 1.16883 \\\midrule
	\multirow4*{62}  
	    & \comps         & (26,16,18)     & 1.51016\\
	    & \swaps        & (0,60,0)       & 0.209249\\
	    & \bytec           & (12,21,27)     & 16.5914\\
	    & \se           & (10,25,25)     & 1.15207\\\midrule
	\multirow4*{100} 
	    & \comps         & (42,26,30)     & 1.50372\\
	    & \swaps        & (0,98,0)       & 0.19107\\
	    & \bytec           & (20,34,44)     & 16.513\\
	    & \se           & (16,41,41)     & 1.14556\\\bottomrule
	\end{tabular}%
	}
	\caption{%
		Optimal sampling parameter $\vect{t^\ast}$ for the different
		cost measures and several fixed values of the sample size $k$.%
	}
	\label{tab:optimal-finite-k}
\end{table}

In \wref{sec:continuous-ranks} we explore how $\vect{t^\ast}$ evolves as 
$k\to\infty$: for each cost measure there exists an optimal 
parameter $\vect{\tau^\ast}=\lim_{k\to\infty}\vect{t^\ast}/k$. 
For finite $k$ several remarks are in order; 
the most salient 
features of $\vect{t^\ast}$ can be easily spotted
from a short table like \wref{tab:optimal-finite-k}.

First, for swaps the optimal sampling parameter is always 
$\vect{t^*}=(0,k-2,0)$ ($(0,0,k-2)$ is also optimal)
and 
\[
q_S^\ast \wwrel= \frac{2k(k+1)}{(2kH_k-1)(k+2)}  \;.
\]
Indeed, as far as swaps are concerned, pivot $P$ should be as small
as possible while pivot $Q$ is as large as possible, for then 
the expected number of swaps in a single partitioning step 
is $2/(k+2)$. 

For comparisons it is not true that a balanced sampling
parameter $\vect{t}=(\lambda,\lambda,\lambda)$ (when $k=3\lambda+2$)
is the best choice, except for $\lambda=1$. 
For instance, for $k=8$ we have 
$\vect{t^\ast}=(3,1,2)$. The behavior of $\vect{t^\ast}$ as $k$ 
increases is somewhat erratic, although it quickly converges
to $\approx (0.43k,0.27k,0.3k)$ (cf.\ \wref{sec:continuous-ranks}).

For Bytecodes and scanned elements, the optimal sampling parameters
are even more biased. They are not very different from 
each other.

In the case of scanned elements, if $\vect{t}=(t_1,t_2,t_3)$ 
is optimal so is $\vect{t}'=(t_1,t_3,t_2)$ (since 
$\discreteEntropy$ is symmetric in $t_1$, $t_2$ and $t_3$ and 
$a_{\scans}$ is symmetric in $t_2$ and $t_3$). The optimal choice 
for scanned elements seems always to be of the 
form $(t_1,t_2,t_2)$ or 
$(t_1,t_2,t_2+1)$ (or $(t_1,t_2+1,t_2)$). 

Assuming that the optimal parameter is of the form 
$\vect{t^\ast}=(t_1,t_2,t_2)$ with 
$t_2=(k-2-t_1)/2$ we can obtain an approximation for the optimal
$t_1^\ast$ by looking at $q_{\scans}=a_{\scans}/\discreteEntropy$ as
a continuous function of its arguments and substituting $\harm{n}$
by $\ln(n)$: taking derivatives w.\,r.\,t.\ $t_1$, 
and solving $dq_{\scans}/dt_1=0$ gives us 
$t_1^\ast\approx (3-2\sqrt{2})k$. 
Indeed, $\vect{t}=(t_1,t_2,k-2-t_1-t_2)$ 
with 
\begin{align*}
	t_1 \wrel= \bigl\lfloor q^2 (k-2) \bigr\rfloor, 
	\qquad
	t_2 \wrel= \bigl\lceil q(k-2) \bigr\rceil
	\quad\text{ and }\quad
	q \wrel= \sqrt2 -1
\end{align*}
is the optimal sampling parameter for most $k$
(in particular for all values of $k$ in \wref{tab:optimal-finite-k}).

\bigskip\noindent
It is interesting to note in this context that the implementation in
Oracle's Java 7 runtime library\,---\,which uses $\vect t =
(1,1,1)$\,---\,executes asymptotically \emph{more} Bytecodes
and needs more element scans
(on random permutations) than \generalYarostM 
with $\vect t=(0,1,2)$, 
despite using the same sample size $k=5$.
Whether this also results in a performance gain in practice, 
however, depends
on details of the runtime environment \citep{Wild2013Alenex}.
(One should also note that the savings are only $2\.\%$ respectively $4\.\%$.)
Since these two cost measures, Bytecodes and scanned elements, 
are arguably the ones
with highest impact on running time, it is very good 
news from the practitioner's point of view that the optimal choice for 
one of them is also reasonably good for the other; such choice 
should yield a close-to-optimal running time (as far as 
sampling is involved).

\subsection{Continuous ranks}
\label{sec:continuous-ranks}

It is natural to ask for the optimal \emph{relative ranks} of $P$ and $Q$
if we are not constrained by the discrete nature of pivot sampling.
In fact, one might want to choose the sample size depending on those optimal
relative ranks to find a discrete order statistic that falls close to the
continuous optimum.

We can compute the optimal relative ranks by considering the limiting
behavior of $\generalYarostM$ as $k\to\infty$. 
Formally, we consider the following family of algorithms:
let $(\sssh tr{(j)})_{j\in\N}$ for $r=1,2,3$ be three
sequences of non-negative integers and set 
$$\ui kj \ce \sssh t1{(j)} + \sssh t2{(j)} + \sssh t3{(j)} + 2$$
for every $j\in\N$. 
Assume that we have $\ui kj \to \infty$ and
${\sssh tr{(j)}}/{\ui kj} \to \tau_r$ with $\tau_l\in[0,1]$ for $r=1,2,3$ as
$j\to\infty$.
Note that we have $\tau_1 + \tau_2 + \tau_3 = 1$ by definition.
For each $j\in\N$, we can apply \wref{thm:expected-costs} 
for $\sss \YQS{\ui{\vect t}j}{\isthreshold}\mkern6mu$ and then
consider the limiting behavior of the total costs for $j\to\infty$.
	(Letting the sample size go to infinity implies non-constant overhead per
	partitioning step for our implementation, which is not negligible
	any more. 
	For the analysis here, we simply assume an oracle that provides us
	with the desired order statistic in constant time.)

For $\discreteEntropy[{\protect\ui{\vect t}{j}}]$, \wildref[Equation]{eq:limit-g-entropy}{\eqref} shows
convergence to the entropy function
$\contentropy = \contentropy[\vect \tau] =  -\sum_{r=1}^3 \tau_r \ln(\tau_r)$ and 
for the numerators $a_C$, $a_S$, $a_\bytecodes$ and $a_\scans$, it is easily seen that
\begin{align*}
		\ui{a_C}j 
	&\wwrel\to 
		\like[l]{a^*_\bytecodes}{a^*_C} 
		\wrel\ce 
		1 + \tau_2 + (2\tau_1 + \tau_2) \tau_3 \,, 
\\		\ui{a_S}j
	&\wwrel\to
		\like[l]{a^*_\bytecodes}{a^*_S} 
		\wrel\ce 
		\tau_1 + (\tau_1 + \tau_2)\tau_3 \,,
\\		\ui{a_\bytecodes}j
	&\wwrel\to
		a^*_\bytecodes 
		\wrel\ce	
		10 + 13\tau_1 + 5\tau_2 + (\tau_1+\tau_2)(\tau_1+11\tau_3) \,,
\\		\ui{a_\scans}j
	&\wwrel\to
		a^*_\scans 
		\wrel\ce	
		1 + \tau_1
	\;.
\end{align*}
Together, the overall number of comparisons, swaps, Bytecodes and scanned elements 
converge to
$a^*_C / \contentropy$, $a^*_S / \contentropy$, $a^*_\bytecodes / \contentropy$ resp.\
$a^*_\scans / \contentropy$;
see \wref{fig:3dplot-limit-total-costs} for plots of the four as functions in $\tau_1$ and $\tau_2$.
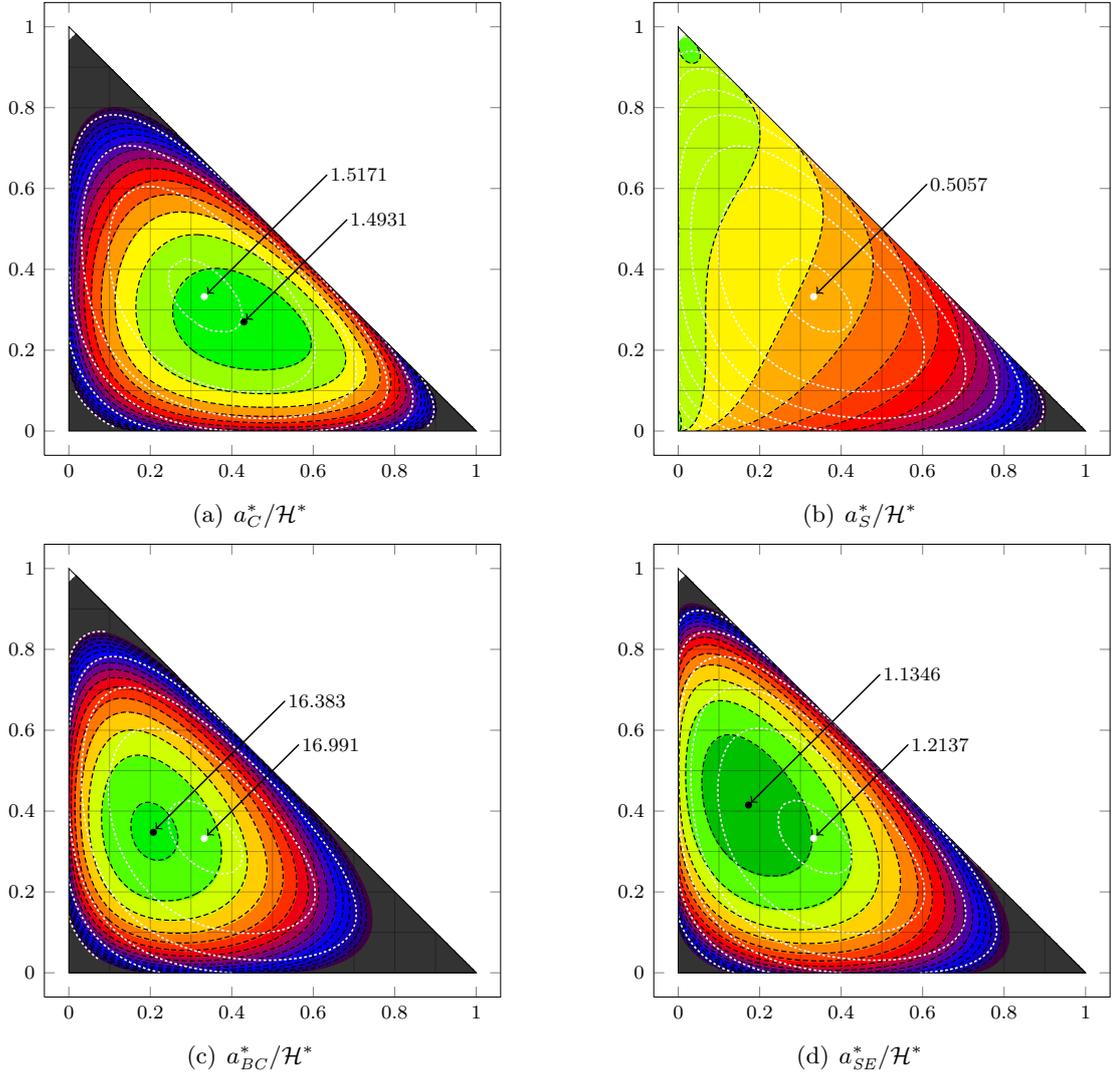
\begin{figure}
	\setlength\subfigbottomskip{0pt}
	\newcommand\contourplot[5]{%
		\resizebox{.46\linewidth}!{
			\begin{tikzpicture}%
				\begin{axis}[
						width=.55\linewidth,
						height=.55\linewidth,
						font=\scriptsize,
						enlargelimits=0.06,
						xmin=0,xmax=1,ymin=0,ymax=1,
						xtick={0,0.2,...,1},
						ytick={0,0.2,...,1},
						y label style={rotate=-90},
					]
					\addplot graphics [xmin=0,xmax=1,ymin=0,ymax=1] 
						{pslt-pics/#1-inf-entropy-paper} ;
					\draw[thin] (axis cs:0,0) -- (axis cs:1,0) -- (axis cs:0,1) -- cycle ;
					\begin{scope}[ultra thin,opacity=.3]
					\draw (axis cs: 0, .1) -| (axis cs:.9,0) ;
					\draw (axis cs: 0, .2) -| (axis cs:.8,0) ;
					\draw (axis cs: 0, .3) -| (axis cs:.7,0) ;
					\draw (axis cs: 0, .4) -| (axis cs:.6,0) ;
					\draw (axis cs: 0, .5) -| (axis cs:.5,0) ;
					\draw (axis cs: 0, .6) -| (axis cs:.4,0) ;
					\draw (axis cs: 0, .7) -| (axis cs:.3,0) ;
					\draw (axis cs: 0, .8) -| (axis cs:.2,0) ;
					\draw (axis cs: 0, .9) -| (axis cs:.1,0) ;
					\end{scope}
					\ifthenelse{\equal{#2}{}}{}{%
 						\draw[semithick,black,<-,shorten <=1.5pt] 
 							(axis cs:#2) -- +(45:{6em+#5}) 
 							node[anchor=west,inner sep=1pt] {$#3$} ;
 					}
 					\draw[semithick,black,<-,shorten <=1.5pt] 
 						(axis cs:0.3333,0.3333) -- +(45:6em-#5) 
 						node[anchor=west,inner sep=1pt]	{$#4$} ;
				\end{axis}
			\end{tikzpicture}
		}%
	}%
	\subfigure[$a^*_C / \contentropy$]{%
		\contourplot{comparisons}{0.4288,0.2688}{1.4931}{1.5171}{-0.5em}%
	}\hfill%
	\subfigure[$a^*_S / \contentropy$]{%
		\contourplot{swaps}{}{}{0.5057}{0pt}%
	}\hfill%
	\subfigure[$a^*_{\bytecodes} / \contentropy$]{%
		\contourplot{bytecodes}{0.2068,0.3486}{16.383}{16.991}{1em}%
	}\hfill%
	\subfigure[$a^*_{\scans} / \contentropy$]{%
		\contourplot{scans}{0.1716,0.4142}{1.1346}{1.2137}{1em}%
	}%
	\caption{%
		Contour plots for the limits of the leading-term coefficient of the
		overall number of comparisons, swaps, executed Bytecode instructions and scanned elements,
		as functions of $\vect\tau$.
		$\tau_1$ and $\tau_2$ are given on $x$- and $y$-axis, respectively, which
		determine $\tau_3$ as $1-\tau_1-\tau_2$. 
		Black dots mark global minima, white dots show the center point
		$\tau_1=\tau_2=\tau_3=\frac13$. 
		(For swaps no minimum is attained in the open simplex, see main text).
		Black dashed lines are level lines connecting ``equi-cost-ant'' points, \ie,
		points of equal costs. 
		White dotted lines mark points of equal entropy $\contentropy$. 
	}
	\label{fig:3dplot-limit-total-costs}
\end{figure}
We could not find a way to compute the minima of these functions analytically.
However, all three functions have isolated minima that can be
approximated well by numerical methods.

The number of comparisons is minimized for 
\begin{align*}
		\vect\tau^*_C 
	&\wwrel\approx
		(0.428846,0.268774,0.302380) \;.
\end{align*}
For this choice, the expected number of comparisons is
asymptotically $1.4931 \, n\ln n$.
For swaps, the minimum is not attained inside the open simplex, but for
the extreme points $\vect\tau^*_S = (0,0,1)$ and 
$\vect\tau_S^{*\prime} = (0,1,0)$.
The minimal value of the coefficient is $0$, so the expected number of swaps
drops to $o(n\ln n)$ for these extreme points.
Of course, this is a very bad choice w.\,r.\,t.\ other cost measures, \eg,
the number of comparisons becomes quadratic, which
again shows the limitations of tuning an algorithm to one of its
basic operations in isolation.
The minimal asymptotic number of executed Bytecodes of
roughly $16.3833 \, n\ln n$ is obtained for 
\begin{align*}
		\vect\tau^*_\bytecodes 
	&\wwrel\approx
		(0.206772,0.348562,0.444666) \;.
\end{align*}
Finally, the least number of scanned elements, 
which is asymptotically $1.1346 \,n\ln n$,
is achieved for 
\begin{align*}
		\vect\tau^*_\scans 
	&\wwrel=
		(q^2,q,q) \quad\text{with}\quad q = \sqrt2 - 1
\\	&\wwrel\approx
		(0.171573,0.414214,0.414214) \;.
\end{align*}

\medskip\noindent
We note again that the optimal choices
heavily differ depending on the employed cost measure and that the minima differ
significantly from the symmetric choice $\vect\tau=(\frac13,\frac13,\frac13)$.

\subsection{Comparison with Classic Quicksort}
\label{sec:YQS-vs-other-QSs}

\subsubsection{Known Results for Classic Quicksort}
\label{sec:CQS-known-results}

Similarly to our \wref{thm:expected-costs}, one can analyze the costs of
classic Quicksort (CQS) with pivot sampling parameter $\vect t = (t_1,t_2) \in \N^2$, 
where the (single) pivot $P$ is chosen as the $(t_1+1)$st-largest from
a sample of $k=k(\vect t) = t_1+t_2+1$ elements, see \citet{Martinez2001}.
With 
$\discreteEntropy(t_1,t_2) \ce \sum_{r=1}^2 \frac{t_r+1}{k+1}(\harm{k+1} - \harm{t_r+1})$
defined similarly as in \wref[Equation]{eq:discrete-entropy}, we have the following results.

\begin{theorem}[Expected Costs of CQS]
\label{thm:expected-costs-CQS}
	Generalized Classic Quicksort with pivot sampling
	parameter $\vect t = (t_1,t_2)$ performs on average
	$ \sss{C}n{\CQS} \sim \sssh aC{\CQS} \!/ {\discreteEntropy} \, n\ln n$
	comparisons,
	$ \sss Sn{\CQS} \sim \sssh aS{\CQS} \!/ {\discreteEntropy} \, n\ln n$
	swaps and
	$ \sss {\scans}n{\CQS} \sim \sssh a{\scans}{\CQS} \!/ {\discreteEntropy} \, n\ln n$
	element scans to sort a random permutation of $n$ elements, where 
	\begin{align*}
			\sssh aC{\CQS}
		&\wwrel=
			\sssh a{\scans}{\CQS}
		\wwrel=
			1
		\qquad\text{and}
	\\
			\sssh aS{\CQS}
		&\wwrel=
			\frac{(t_1 + 1)(t_2 + 1)}{(k+1)(k+2)} \;.
	\end{align*}
	Moreover, if the partitioning loop is implemented as in Listing~4
	of~\citep{Wild2012thesis}, it executes on average 
	$\sss \bytecodes n{\CQS} \sim \sssh a{\bytecodes}{\CQS} \!/ {\discreteEntropy} \, n\ln n$
	Java Bytecode instructions to sort a random permutation of size $n$ with
	\begin{align*}
			\sssh a{\bytecodes}{\CQS}
		&\wwrel=
			6 \sssh aC{\CQS} + 18 \sssh aS{\CQS}
			\;.
	\end{align*}
\qed\end{theorem}

\medskip\noindent
\textbf{Remark:}
In CQS, each element reached by a scanning index results in exactly one comparison 
(namely with the pivot).
Therefore, the number of scanned elements and the number of key comparisons are exactly
the same in CQS.

\subsubsection{Pivots from Fixed Positions}
\label{sec:YQS-vs-CQS-no-sampling}

The first theoretical studies of the new Quicksort variant invented by 
Yaroslavskiy assumed that pivots are chosen from fixed positions of the input.
Trying to understand the reasons for its running time advantages we 
analyzed comparisons, swaps and the number of executed Bytecode instructions for
YQS and CQS. 
However, comparing all related findings to corresponding results for classic Quicksort, 
we observed that YQS needs about $5\.\%$ less comparisons than CQS, but performs 
about twice as many swaps, needs $65\.\%$ more write accesses and executes about 
$20\.\%$ more Bytecodes on average \citep{Wild2013Quicksort}.
What is important here is that these results hold not only asymptotically,
but already for practical $n$. 
(Without pivot sampling, an exact solution of the recurrences remains feasible.)
Thus, it was somehow straightforward to utter the following conjecture.

\medskip\noindent
\textbf{Conjecture 5.1 of \citet{Wild2013Quicksort}:} {\slshape
``The efficiency of Yaroslavskiy's algorithm in practice is caused by 
advanced features of modern processors. 
In models that assign constant cost contributions to single instructions\,---\,\ie, 
locality of memory accesses and instruction pipelining are ignored\,---\,classic 
Quicksort is more efficient.''
}

\medskip\noindent
\citet{Kushagra2014} then were the first to provide
strong evidence for this conjecture by showing that YQS needs significantly 
less cache misses than CQS. 
Very recently, we were able to exclude the effects of pipelined execution 
from the list of potential explanations; 
both algorithms CQS and YQS give rise to about the same number of branch misses 
on average, so their rollback costs cannot be responsible 
for the differences in running time \citep{MartinezNebelWild2015}.

In this paper we present a precise analysis of the number of scanned elements
per partitioning step (cf.\ \wref{lem:distribution-partitioning-scans}).
Plugging this result into the precise solution of the dual-pivot Quicksort
recurrence without pivot sampling, we get the precise total number
of scanned elements:
\begin{itemize}
\item 
	YQS scans $1.6 n \ln(n) - 2.2425 n + \Oh(\log n)$ elements on average, while
\item 
	CQS needs $2 n \ln(n) - 2.3045 n + \Oh(\log n)$ 
	element scans on average.
		\\(Recall that scanned elements and comparisons coincide in CQS, so we
		can reuse results for comparisons, see \eg\ \citep{Sedgewick1977}.) 
\end{itemize}
Both results are actually known precisely, but the sublinear terms are
really negligible for reasonable input sizes.

Obviously, the number of scanned elements is significantly smaller in YQS
that in CQS \emph{for all} $n$.
Accordingly, and in the light of all the results mentioned before, 
we assume our conjecture to be verified (for pivots taken from fixed positions):
YQS is more efficient in practice than CQS because it needs less element scans
and thus uses the memory hierarchy more efficiently.

Note that asymptotically, YQS needs $25\.\%$ less element scans, but at the same 
time executes $20\.\%$ more Bytecodes.
In terms of practical running time, it seems plausible that both 
Bytecodes and scanned elements yield their share.
In experiments conducted by one of the authors,
YQS was $13\.\%$ faster in Java and $10\.\%$ faster in C++ \citep{Wild2012thesis},
which is not explained well by either cost measure in isolation.

One might assume that a sensible model for actual running time is a 
\emph{linear combination} of Bytecodes and scans
$$
	Q \wwrel= (1-\mu)\cdot \bytecodes + \mu\cdot\scans
$$
for an (unknown) parameter $\mu\in[0,1]$.
Intuitively, $\mu$ is the relative importance of the number of scanned elements for
total running time.
Inserting the results for CQS and YQS and solving $Q^\CQS / Q^\YQS = 1.1$ for $\mu$,
we get $\mu\approx0.95$.
(The solution actually depends on $n$, so there is one solution for every input size.
However, we get $0.93\le \mu\le 0.96$ for all $n\ge100$.)
This means\,---\,assuming the linear model is correct\,---\,that $95\.\%$ of
the running time of Quicksort are caused by element scans and
only $5\.\%$ by executed Bytecodes.
Stated otherwise, a single scanned element is as costly as executing 20 Bytecode 
instructions.

\subsubsection{Pivots from Samples of Size $k$}
\label{sec:YQS-vs-CQS-with-sampling}

While the last section discussed the most elementary versions
of CQS and YQS, we will now come back to the case where pivots 
are chosen from a sample.
To compare the single-pivot CQS with the dual-pivot YQS,
we need two pivot sampling parameters $\vect t$,
which we here call 
$\vect t^\CQS \in \N^2$ and $\vect t^\YQS \in \N^3$,
respectively.
Of course, they potentially result in different sample sizes
$k^\CQS = \sssh t1\CQS + \sssh t2\CQS + 1$ and
$k^\YQS = \sssh t1\YQS + \sssh t2\YQS + \sssh t3\YQS + 2$.

Analytic results for general pivot sampling are only available
as leading-term asymptotics, so we have to confine ourselves to 
the comparison of CQS and YQS on very large inputs.
Still, we consider it unsound to compare, say, 
YQS with a sample size $k^\YQS = 100$ to CQS with sample size 
$k^\CQS = 3$, where one algorithm is allowed to use much more
information about the input to make its decision for good pivot
values than the other.
Moreover, even though sample size analytically only affect 
the linear term of costs,
the former would in practice spend a non-negligible amount
of its running time sorting the large samples, 
whereas the latter knows its pivot after just three quick 
key comparisons.
For a fair competition, we will thus keep the sample sizes 
in the same range.

\begin{table}
	\def\comps{comparisons}
	\def\swaps{swaps}
	\def\bytec{Bytecodes}
	\def\se{scanned elements}
	\def\cqs{classic Quicksort}
	\def\yqs{Yaroslavskiy's Quicksort}
	\plaincenter{%
	\begin{tabular}{ c r D..{6} D..{6} }
	\toprule
		$k$ & cost measure & \multicolumn1c{\cqs} & \multicolumn1c{\yqs} \\
	\midrule 
		\multirow4*{no samp\rlap{ling}} & \comps & 2 & 1.9 \\
				& \swaps & 0.\overline3 & 0.6 \\
				& \bytec & 18 & 21.7 \\
				& \se & 2 & 1.6 \\
	\midrule 
		\multirow4*{5} 
				& \comps & 1.6216  & 1.7043 \\
				& \swaps & 0.3475  & 0.5514 \\
				& \bytec & 15.9846 & 19.2982 \\
				& \se    & 1.6216  & 1.4035 \\
	\midrule 
		\multirow4*{11} 
				& \comps & 1.5309  & 1.6090 \\
				& \swaps & 0.3533  & 0.5280 \\
				& \bytec & 15.5445 & 18.1269 \\
				& \se    & 1.5309  & 1.3073 \\
	\midrule 
		\multirow4*{17} 
				& \comps & 1.5012  & 1.5779 \\
				& \swaps & 0.3555  & 0.5204 \\
				& \bytec & 15.4069 & 17.7435 \\
				& \se    & 1.5012  & 1.2758 \\
	\midrule 
		\multirow4*{23} 
				& \comps & 1.4864  & 1.5625 \\
				& \swaps & 0.3567  & 0.5166 \\
				& \bytec & 15.3401 & 17.5535 \\
				& \se    & 1.4864  & 1.2601 \\
	\bottomrule
	\end{tabular}%
	}	
	\caption{%
		Comparison of CQS and YQS whose pivots are chosen equidistantly 
		from samples of the given sizes.
		All entries give the (approximate) 
		leading-term coefficient of the asymptotic
		cost for the given cost measure.
		By $0.\overline 3$ we mean the repeating decimal $0.333\ldots = \frac13$.
	}
	\label{tab:CQS-vs-YQS-median-tertiles-of-6t-1}
\end{table}

Once the sample size is fixed, one can still choose different
order statistics of the sample.
As the optimal choices for YQS are so sensitive to the employed
cost measure, we will first focus on choosing symmetric pivots,
\ie, 
$\vect t^\CQS = (t^\CQS,t^\CQS)$ and
$\vect t^\YQS = (t^\YQS,t^\YQS,t^\YQS)$, 
for integers $t^\CQS$ and $t^\YQS$,
such that the sample sizes are exactly the same.
This effectively limits the allowable sample sizes to 
$k=6\lambda-1$ for integers $\lambda\ge1$;
\wref{tab:CQS-vs-YQS-median-tertiles-of-6t-1} shows
the results up to $\lambda=4$.

As $k$ increases, the algorithms improve in all cost measures,
except for the number of swaps in CQS.
The reason is that swaps profit from unbalanced pivots, which we make less
likely by sampling 
(see \citep{Martinez2001} and \citep{Wild2012thesis} for a more detailed discussion).
Moreover, the (relative) ranking of the two algorithms w.\,r.\,t.\ each cost
measure in isolation is the same for all sample sizes and thus similar to
the case without sampling (see \wref{sec:YQS-vs-CQS-no-sampling})\,---\,%
with a single exception:
without sampling, YQS need $5\.\%$ less comparisons than CQS, 
but for all values of $k$ in \wref{tab:CQS-vs-YQS-median-tertiles-of-6t-1}, 
YQS actually needs $5\.\%$ \emph{more} comparisons!
As soon as the variance of the ranks of pivots is reduced by sampling,
the advantage of YQS to exploit skewed pivots to save comparisons
through clever use of asymmetries in the code is no longer enough to beat CQS
if the latter chooses its pivot as median of a sample of the same size.
This remains true if we allow YQS to choose the order statistics that
minimize the number of comparisons: we then get as leading-term coefficients
of the number of comparisons
$1.7043$, $1.5848$, $1.5554$ and $1.5396$ for $k=5$, $11$, $17$ and $23$, respectively,
which still is significantly more than for CQS with median-of-$k$.

This is a quite important observation, as it shows that the 
number of key comparisons cannot be the reason for YQS's success in practice: 
for the library implementations, YQS has always been compared to 
CQS with pivot sampling, \ie, to an algorithm that needs \emph{less} comparisons
than YQS.
To be precise, the Quicksort implementation used in Java 6 is the
version of \citet{Bentley1993} which uses the \textsl{``ninther''} as pivot:
Take three samples of three elements each, pick the median of each
of the samples and then make the median of the three medians our pivot.
The expected number of key comparisons used by this algorithm
has been computed by \citet{Durand2003pseudonine}. The leading-term
coefficient is $\frac{12\.600}{8\.027} \approx 1.5697$, 
ranking between CQS with median-of-seven and median-of-nine.
The version of Yaroslavskiy's Quicksort used in Java 7 uses the
tertiles-of-five as pivots and needs (asymptotically) 
$1.7043 \,n \ln n$ comparisons.

Similarly, CQS needs less swaps and Bytecode instructions than YQS.
If we, however, compare the same two algorithms in terms of the number
of scanned elements they need, YQS clearly wins with
$1.4035\, n\ln n$ vs.\ $1.5697\, n \ln n$ 
in the asymptotic average.
Even quantitatively, this offers a plausible explanation of 
running time differences:
The Java~7 Quicksort saves $12\.\%$ of the element scans over
the version in Java~6, which roughly matches speedups observed
in running time studies.

	One should note at this point, however, that the library versions
	are \emph{not} direct implementations of the basic partitioning algorithms
	as given in \wref{alg:generalized-yaroslavskiy} for YQS.
	For example, the variant of \citet{Bentley1993} actually does a three-way
	partitioning to efficiently deal with inputs with many equal keys
	and the Java 7 version of YQS uses similar tweaks.
	The question, whether scanned elements (or cache misses) are the 
	dominating factor in the running time of these algorithms, 
	needs further study.

We conclude that also for the pivot sampling strategies employed in practice,
YQS clearly outperforms CQS in the number of scanned elements.
It is most likely that this more efficient use of the memory hierarchy 
makes YQS faster in practice.

\section{Conclusion}
\label{sec:conclusion}

In this paper, we give the precise leading-term asymptotic of the average costs
of Quicksort with Yaroslavskiy's dual-pivot partitioning method and selection of
pivots as arbitrary order statistics of a constant-size sample for a variety
of different cost measures:
the number of key comparisons and the number of swaps (as classically used
for sorting algorithms), but also the number of executed Java Bytecode
instructions and the number of scanned elements, a new cost measure
that we introduce as simple model for the number of cache misses.

The inherent asymmetries in Yaroslavskiy's partitioning algorithm
lead to the situation that the symmetric choice for pivots, the tertiles
of the sample, is \emph{not} optimal:
a deliberate, well-dosed skew in pivot selection improves overall
performance.
For the optimal skew, we have to find a trade-off between several 
counteracting effects and the result is very sensitive to the 
employed cost measure.
The precise analysis in this paper can provide valuable guidance in choosing
the right sampling scheme.

Whereas cache misses are complicated in detail and machine-dependent, scanned
elements are a precisely defined, abstract cost measure that is as elementary 
as key comparisons or swaps.
At the same time, it provides a reasonable approximation for the number of 
incurred cache misses, 
and we show in particular that the number of scanned elements is well-suited
to compare different Quicksort variants w.\,r.\,t.\ their efficiency in the 
external-memory model.

Comparing classic single-pivot Quicksort with Yaroslavskiy's dual-pivot
Quicksort in terms of scanned elements finally yields a convincing analytical
explanation why the latter is found to be more efficient in practice:
Yaroslavskiy's algorithm needs much less element scans and thus 
uses the memory hierarchy more efficiently, 
with and without pivot sampling.

In light of the complexity of modern machines, 
it is implausible that a single simple cost measure captures 
all contributions to running time;
rather, it seems likely that the number of scanned elements 
(memory accesses)
and the number of executed instructions in the CPU both have 
significant influence.
With algorithms as excessively studied and tuned as Quicksort,
we have reached a point where slight changes in the underlying hardware 
architecture can shift the weights of these factors enough to make 
variants of an algorithm superior on today's machines which
were not competitive on yesterday's machines:
CPU speed has increased much more than memory speed, shifting the weights
towards algorithms that save in scanned elements, 
like Yaroslavskiy's dual-pivot Quicksort.

\paragraph{Future work}
A natural extension of this work would be the computation of
the linear term of costs, which is not negligible for moderate $n$.
This will require a much more detailed analysis as sorting the samples and
dealing with short subarrays contribute to the linear term of costs, but then
allows to compute the optimal choice for \isthreshold, as well.
While in this paper only expected values were considered,
the distributional analysis of \wref{sec:distributional-analysis} can be used as a
starting point for analyzing the distribution of overall costs.
Yaroslavskiy's partitioning can also be used in Quickselect
\citep{WildMahmoud2013}; the effects of generalized pivot sampling there
are yet to be studied.
Finally, other cost measures, like the number of
symbol comparisons \citep{Vallee2009symbolComparisons,Fill2012},
would be interesting to analyze.

	\newenvironment{acknowledgement}{\section*{Acknowledgements}}{}

\begin{acknowledgement}
We thank two anonymous reviewers for their careful 
reading and helpful comments.
\end{acknowledgement}

\small
\bibliography{quicksort-refs-springer}

\clearpage
\appendix
\section*{{\LARGE Appendix}}
\section{Index of Used Notation}
\label{app:notations}
\def\mydots{\xleaders\hbox to.75em{\hfill.\hfill}\hfill}

\newlength\tmpLenNotations
\newenvironment{notations}[1][10em]{%
	\small
	\newcommand\notationentry[1]{%
		\settowidth\tmpLenNotations{##1}%
		\ifthenelse{\lengthtest{\tmpLenNotations > \labelwidth}}{%
			\parbox[b]{\labelwidth}{%
				\makebox[0pt][l]{##1}\\%
			}%
		}{%
			\mbox{##1}%
		}%
		\mydots\relax%
	}%
	\begin{list}{}{%
		\setlength\labelsep{0em}%
		\setlength\labelwidth{#1}%
		\setlength\leftmargin{\labelwidth+\labelsep+1em}%
		\renewcommand\makelabel{\notationentry}%
	}
	\newcommand\notation[1]{\item[##1]}
	\raggedright
}{%
	\end{list}
}

In this section, we collect the notations used in this paper.
(Some might be seen as ``standard'', but we think
including them here hurts less than a potential misunderstanding caused by
omitting them.)

\subsection*{Generic Mathematical Notation}
\begin{notations}
\notation{$0.\overline 3$}
	repeating decimal; $0.\overline3 = 0.333\ldots = \frac13$. \\
	The numerals under the line form the repeated part of 
	the decimal number.
\notation{$\ln n$}
	natural logarithm.
\notation{linearithmic}
	A function is ``linearithmic'' if it has order of growth $\Theta(n \log n)$.
\notation{$\vect x$}
	to emphasize that $\vect x$ is a vector, it is written in \textbf{bold};\\
	components of the vector are not written in bold: $\vect x = (x_1,\ldots,x_d)$.
\notation{$X$}
	to emphasize that $X$ is a random variable it is Capitalized.
\notation{$\harm{n}$}
	$n$th harmonic number; $\harm n = \sum_{i=1}^n 1/i$.
\notation{$\dirichlet(\vect \alpha)$}
	Dirichlet distributed random variable, 
	$\vect \alpha \in \R_{>0}^d$.
\notation{$\multinomial(n,\vect p)$}
	multinomially distributed random variable; 
	$n\in\N$ and $\vect p \in [0,1]^d$ with $\sum_{i=1}^d p_i = 1$.
\notation{$\hypergeometric(k,r,n)$}
	hypergeometrically distributed random variable;
	$n\in\N$, $k,r,\in\{1,\ldots,n\}$.  
\notation{$\bernoulli(p)$}
	Bernoulli distributed random variable;
	$p\in[0,1]$.
\notation{$\uniform(a,b)$}
	uniformly in $(a,b)\subset\R$ distributed random variable. 
\notation{$\BetaFun(\alpha_1,\ldots,\alpha_d)$}
	$d$-dimensional Beta function; defined in
	\wildpageref[Equation]{eq:def-beta-function}{\eqref}.
\notation{{$\E[X]$}}
	expected value of $X$; we write $\E[X\given Y]$ for the conditional expectation
	of $X$ given $Y$.
\notation{$\Prob(E)$, $\Prob(X=x)$}
	probability of an event $E$ resp.\ probability for random variable $X$ to
	attain value $x$.
\notation{$X\eqdist Y$}
	equality in distribution; $X$ and $Y$ have the same distribution.
\notation{$X_{(i)}$}
	$i$th order statistic of a set of random variables $X_1,\ldots,X_n$,\\
	\ie, the $i$th smallest element of $X_1,\ldots,X_n$. 
\notation{$\indicator{E}$}
	indicator variable for event $E$, \ie, $\indicator{E}$ is $1$ if $E$
	occurs and $0$ otherwise.
\notation{$a^{\underline b}$, $a^{\overline b}$}
	factorial powers notation of \citet{ConcreteMathematics}; 
	``$a$ to the $b$ falling resp.\ rising''.
\end{notations}

\subsection*{Input to the Algorithm}
\begin{notations}
\notation{$n$}
	length of the input array, \ie, the input size.
\notation{\arrayA}
	input array containing the items $\arrayA[1],\ldots,\arrayA[n]$ to be
	sorted; initially, $\arrayA[i] = U_i$.
\notation{$U_i$}
	$i$th element of the input, \ie, initially $\arrayA[i] = U_i$.\\
	We assume $U_1,\ldots,U_n$ are i.\,i.\,d.\ $\uniform(0,1)$ distributed.
\end{notations}

\subsection*{Notation Specific to the Algorithm}
\begin{notations}
\notation{$\vect t \in \N^3$}
	pivot sampling parameter, see \wpref{sec:general-pivot-sampling}.
\notation{$k=k(\vect t)$}
	sample size; defined in terms of $\vect t$ as $k(\vect t) = t_1+t_2+t_3+2$.
\notation{\isthreshold}
	Insertionsort threshold; for $n\le\isthreshold$, Quicksort recursion is
	truncated and we sort the subarray by Insertionsort.
\notation{$M$}
	cache size; the number of array elements that fit into the idealized cache;
	we assume $M\ge B$, $B\mid M$ ($M$ is a multiple of $B$) and $B\mid n$;
	see \wref{sec:validation-scans-vs-cache-misses}.
\notation{$B$}
	block size; the number of array elements that fit into one cache block/line;
	see also $M$.
\notation{$\YQS$, $\generalYarostM$}
	abbreviation for dual-pivot Quicksort with Yaroslavskiy's partitioning method,
	where pivots are chosen by generalized pivot sampling with parameter $\vect t$
	and where we switch to Insertionsort for subproblems of size at most
	\isthreshold.
\notation{$\CQS$}
	abbreviation for classic (single-pivot) Quicksort using Hoare's partitioning,
	see \eg\ \citep[p.\,329]{Sedgewick1977};
	a variety of notations are with $\CQS$ in the superscript 
	to denote the corresponding quantities for classic Quicksort,\\
	\eg, $\sss Cn\CQS$ is the number of (partitioning) comparisons needed
	by CQS on a random permutation of size $n$.
\notation{$\vect V \in \N^k$}
	(random) sample for choosing pivots in the first partitioning step.
\notation{$P$, $Q$}
	(random) values of chosen pivots  in the first partitioning step.
\notation{small element}
	element $U$ is small if $U<P$.
\notation{medium element}
	element $U$ is medium if $P<U<Q$.
\notation{large element}
	element $U$ is large if $Q < U$.
\notation{sampled-out element}
	the $k-2$ elements of the sample that are \emph{not} chosen as pivots.	
\notation{ordinary element}
	the $n-k$ elements that have not been part of the sample.
\notation{$k$, $g$, $\ell$}
	index variables used in Yaroslavskiy's partitioning method, see
	\wpref{alg:partition}. 
\notation{$\positionsets{K}$, $\positionsets{G}$, $\positionsets{L}$}
	set of all (index) values attained by pointers $k$, $g$ resp.\ $\ell$ during the first
	partitioning step; see \wpref{sec:yaroslavskiys-partitioning-method} and proof of
	\wpref{lem:distribution-partitioning-comparisons}.
\notation{$\numberat{c}{P}$}
	$c\in\{s,m,l\}$, $\positionsets{P} \subset \{1,\ldots,n\}$\\
	(random) number of $c$-type ($s$mall, $m$edium or $l$arge) elements 
	that are initially located at positions in $\positionsets{P}$, \ie,
	$
		\numberat{c}{P} \wrel= 
		\bigl|\{
			i \in \positionsets{P} : U_i \text{ has type } c
		\}\bigr|. 
	$
\notation{$\latK$, $\satK$, $\satG$}
	see $\numberat{c}{P}$
\notation{$\chi$}
	(random) point where $k$ and $g$ first meet.
\notation{$\delta$}
	indicator variable of the random event that $\chi$ is on a large
	element, \ie, $\delta = \indicator{U_\chi > Q}$.
\notation{$C_n^{\typetype}$}
	with $\typetype \in \{\typeroot,\typeleft,\typemiddle,\typeright\}$;
	(random) costs of a (recursive) call to $\proc{GeneralizedYaroslavskiy}(\arrayA,\id{left},\id{right},\typetype)$
	where $\arrayA[\id{left}..\id{right}]$ contains $n$ elements, \ie, $\id{right}-\id{left}+1 = n$.
	The array elements are assumed to be in random order, 
	except for the $t_1$, resp.\ $t_2$ leftmost elements 
	for $C_n^{\typeleft}$ and $C_n^{\typemiddle}$ 
	and the $t_3$ rightmost elements for $C_n^{\typeright}$;\\
	for all $\mathtt{type}$s holds $C_n^{\typetype} \eqdist C_n + \Oh(n)$, see \wref{sec:recurrence-quicksort}.
\notation{$T_n^{\typetype}$}
	with $\typetype \in \{\typeroot,\typeleft,\typemiddle,\typeright\}$;
	the costs of the first partitioning step of a call to 
	$\proc{GeneralizedYaroslavskiy}(\arrayA,\id{left},\id{right},\typetype)$;
	for all $\mathtt{type}$s holds $T_n^{\typetype} \eqdist T_n + \Oh(1)$, see \wref{sec:recurrence-quicksort}.
\notation{$T_n$}
	the costs of the first partitioning step, where \emph{only} costs
	of procedure \proc{Partition} are counted, 
	see \wref{sec:recurrence-quicksort}.
\notation{$\iscost_n^{\typetype}$}
	with $\typetype \in \{\typeroot,\typeleft,\typemiddle,\typeright\}$;
	as $C_n^{\typetype}$, but the calls are 
	$\proc{InsertionSortLeft}(\arrayA, \id{left}, \id{right}, 1)$ for  $\iscost_n^{\typeroot}$,
	$\proc{InsertionSortLeft}(\arrayA, \id{left}, \id{right}, \max\{t_1,1\})$ for $\iscost_n^{\typeleft}$
	$\proc{InsertionSortLeft}(\arrayA, \id{left}, \id{right}, \max\{t_2,1\})$ for $\iscost_n^{\typemiddle}$ and
	$\proc{InsertionSortRight}(\arrayA, \id{left}, \id{right}, \max\{t_3,1\})$ for $\iscost_n^{\typeright}$.
\notation{$\iscost_n$}
	(random) costs of sorting a random permutation of size $n$ with Insertionsort. 
\notation{$C_n$, $S_n$, $\bytecodes_n$, $\scans_n$}
	(random) number of comparisons\,/\,swaps\,/\,Bytecodes\,/\,scanned elements
	of \generalYarostM on a random permutation of size $n$ that are caused in
	procedure \proc{Partition}; see \wref{sec:cost-measures} for more information
	on the cost measures;
	in \wref{sec:recurrence-quicksort}, $C_n$ is used as general placeholder
	for any of the above cost measures.  
\notation{\mbox{$\toll{C}$, $\toll{S}$, $\toll{\bytecodes}$, $\toll{\scans}$}}
	(random) number of comparisons\,/\,swaps\,/\,Bytecodes\,/\,element scans 
	of the first partitioning step of \generalYarostM on a 
	random permutation of size $n$;\\
	$\toll[n]{C}$, $\toll[n]{S}$ and $\toll[n]{\bytecodes}$ when we want to
	emphasize dependence on $n$.
\notation{$a_C$, $a_S$, $a_{\bytecodes}$, $a_{\scans}$}
	coefficient of the linear term of $\E[\toll[n]{C}]$, $\E[\toll[n]{S}]$, 
	$\E[\toll[n]{\bytecodes}]$ and $\E[\toll[n]{\scans}]$; 
	see \wpref{thm:expected-costs}.
\notation{$\discreteEntropy$}
	discrete entropy; defined in
	\wildpageref[Equation]{eq:discrete-entropy}{\eqref}.
\notation{{$\contentropy[\vect p]$}}
	continuous (Shannon) entropy with basis $e$; defined in
	\wildpageref[Equation]{eq:limit-g-entropy}{\eqref}.
\notation{$\vect J\in\N^3$}
	(random) vector of subproblem sizes for recursive calls;\\
	for initial size $n$, we have $\vect J \in \{0,\ldots,n-2\}^3$ with
	$J_1+J_2+J_3 = n-2$.
\notation{$\vect I\in\N^3$}
	(random) vector of partition sizes, \ie, the number of small, medium resp.\
	large \emph{ordinary} elements;
	for initial size $n$, we have $\vect I \in \{0,\ldots,n-k\}^3$ with
	$I_1+I_2+I_3 = n-k$;\\
	$\vect J = \vect I + \vect t$ and conditional on $\vect D$ we
	have $\vect I \eqdist \multinomial(n-k,\vect D)$.
\notation{{$\vect D\in[0,1]^3$}}
	(random) spacings of the unit interval $(0,1)$ induced by the pivots $P$ and
	$Q$, \ie, $\vect D = (P,Q-P,1-Q)$;
	$\vect D \eqdist \dirichlet(\vect t + 1)$.
\notation{$a^*_C$, $a^*_S$, $a^*_{\bytecodes}$, $a^*_\scans$}
	limit of $a_C$, $a_S$, $a_{\bytecodes}$ resp.\ $a_{\scans}$ for the optimal sampling
	parameter $\vect t$ when $k\to\infty$.
\notation{$\vect\tau_C^*$, $\vect\tau_S^*$, $\vect\tau_{\bytecodes}^*$, $\vect\tau_{\scans}^*$}
	optimal limiting ratio $\vect t / k \to \vect \tau_C^*$ such that $a_C \to
	a^*_C$ (resp.\ for $S$, $\bytecodes$ and $\scans$).
\end{notations}

\clearpage
\section{Properties of Distributions}
\label{app:distributions}

We herein collect definitions and basic properties of the distributions used in
this paper. 
They will be needed for computing expected values in
\wref{app:proof-of-lem-expectations}.
This appendix is an update of Appendix~C in \citep{NebelWild2014}, 
which we include here for the reader's convenience.

We use the notation $x^{\overline n}$ and $x^{\underline n}$ of
\citet{ConcreteMathematics} for rising and falling factorial powers, respectively.

\subsection{Dirichlet Distribution and Beta Function}
\label{sec:dirichlet-dist-beta-function}
For $d\in\N$ let $\Delta_d$ be the standard $(d-1)$-dimensional simplex, \ie, 
\begin{align}
\label{eq:def-delta-d}
		\Delta_d
	&\wwrel\ce 
		\biggl\{
			x = (x_1,\ldots,x_d) 
			\wrel: 
			\forall i : x_i \ge 0 \; 
			\rel\wedge 
			\sum_{\mathclap{1\le i \le d}} x_i = 1
		\biggr\} \;.
\end{align}
Let $\alpha_1,\ldots,\alpha_d > 0$ be positive reals.
A random variable $\vect X \in \R^d$ is said to have the 
\emph{Dirichlet distribution} with \emph{shape parameter} 
$\vect\alpha \ce (\alpha_1,\ldots,\alpha_d)$\,---\,abbreviated as
$\vect X \eqdist \dirichlet(\vect\alpha)$\,---\,if it has a density given by
\begin{align}
\label{eq:def-dirichlet-density}
		f_{\vect X}(x_1,\ldots,x_d)
	&\wwrel\ce \begin{cases}
			\frac1{\BetaFun(\vect\alpha)} \cdot
			x_1^{\alpha_1 - 1} \cdots x_d^{\alpha_d-1} ,
			& \text{if } \vect x \in \Delta_d \,; \\
			0 , & \text{otherwise} \..
		\end{cases}
\end{align}
Here, $\BetaFun(\vect\alpha)$ is the \emph{$d$-dimensional Beta function}
defined as the following Lebesgue integral:
\begin{align}
\label{eq:def-beta-function}
		\BetaFun(\alpha_1,\ldots,\alpha_d)
	&\wwrel\ce
		\int_{\Delta_d} x_1^{\alpha_1 - 1} \cdots x_d^{\alpha_d-1} \; \mu(d \vect x)
	\;.
\end{align}
The integrand is exactly the density without the normalization constant
$\frac1{\BetaFun(\alpha)}$, hence $\int f_X \,d\mu = 1$ as needed for
probability distributions.

The Beta function can be written in terms of the Gamma function 
$\Gamma(t) = \int_0^\infty x^{t-1} e^{-x} \,dx$
as
\begin{align}
\label{eq:beta-function-via-gamma}
		\BetaFun(\alpha_1,\ldots,\alpha_d)
	&\wwrel=
		\frac{\Gamma(\alpha_1) \cdots \Gamma(\alpha_d)}
		{\Gamma(\alpha_1+\cdots+\alpha_d)} \;.
\end{align}
(For integral parameters $\vect\alpha$, a simple inductive argument and partial
integration suffice to prove~\wref{eq:beta-function-via-gamma}.)\\
Note that $\dirichlet(1,\ldots,1)$ corresponds to the uniform distribution over
$\Delta_d$.
For integral parameters $\vect\alpha\in\N^d$, $\dirichlet(\vect\alpha)$ is the
distribution of the \emph{spacings} or \emph{consecutive differences} induced by
appropriate order statistics of i.\,i.\,d.\ uniformly in $(0,1)$ distributed
random variables, as summarized in the following proposition. 
\begin{proposition}[{{\citealt[Section\,6.4]{David2003}}}]
\label{pro:spacings-dirichlet-general-dimension}
	Let $\vect\alpha \in \N^d$ be a vector of positive integers and set $k \ce -1 +
	\sum_{i=1}^d \alpha_i$. Further let $V_1,\ldots,V_{k}$ be $k$ random variables
	i.\,i.\,d.\ uniformly in $(0,1)$ distributed.
	Denote by \smash{$V_{(1)}\le \cdots \le V_{(k)}$} their corresponding
	order statistics.
	We select some of the order statistics according to $\vect \alpha$: 
	for $j=1,\ldots,d-1$ define \smash{$W_j \ce V_{(p_j)}$}, where $p_j \ce
	\sum_{i=1}^j \alpha_i$. Additionally, we set $W_0 \ce 0$ and $W_d \ce 1$.
	
	Then, the \textit{consecutive distances} (or \textit{spacings}) $D_j \ce W_j -
	W_{j-1}$ for $j=1,\ldots,d$ induced by the selected
	order statistics $W_1,\ldots,W_{d-1}$ are Dirichlet
	distributed with parameter $\vect \alpha$: 
	\begin{align*}
			(D_1,\ldots,D_d) 
		&\wwrel\eqdist 
			\dirichlet(\alpha_1,\ldots,\alpha_d) \;.
	\end{align*}%
\qed\end{proposition}
\smallskip

In the computations of \wref{sec:expectations}, mixed moments of Dirichlet distributed
variables will show up, which can be dealt with using the following general
statement.
\begin{lemma}
\label{lem:dirichlet-mixed-moments}
	Let $\vect X = (X_1,\ldots,X_d) \in \R^d$ be a $\dirichlet(\vect\alpha)$
	distributed random variable with parameter $\vect\alpha = (\alpha_1,\ldots,\alpha_d)$.
	Let further $m_1,\ldots,m_d \in \N$ be non-negative integers and abbreviate
	the sums $A \ce \sum_{i=1}^d \alpha_i$ and $M \ce \sum_{i=1}^d m_i$. 
	Then we	have 
	\begin{align*}
			\E\bigl[ X_1^{m_1} \cdots X_d^{m_d} \bigr]
		&\wwrel=
			\frac{\alpha_1^{\overline{m_1}} \cdots \alpha_d^{\overline{m_d}}}
				{A^{\overline M}} \;.
	\end{align*}
\end{lemma}

\begin{proof}
Using $\frac{\Gamma(z+n)}{\Gamma(z)} = z^{\overline n}$ for all
$z\in\R_{>0}$ and $n \in \N$, we compute
\begin{align}
		\E\bigl[ X_1^{m_1} \cdots X_d^{m_d} \bigr]
	&\wwrel=
		\int_{\Delta_d} 
			x_1^{m_1} \cdots x_d^{m_d} \cdot
			\frac{x_1^{\alpha_1-1} \cdots x_d^{\alpha_d-1}}{\BetaFun(\vect\alpha)}
		\; \mu(dx)
	\\	&\wwrel=
		\frac{\BetaFun(\alpha_1 + m_1,\ldots,\alpha_d + m_d)}
			{\BetaFun(\alpha_1,\ldots,\alpha_d)}
	\\	&\wwrel{\eqwithref{eq:beta-function-via-gamma}}
		\frac{ \alpha_1^{\overline{m_1}} \cdots \alpha_d^{\overline{m_d}} }
			{ A^{\overline M} } \;.
\end{align}
\end{proof}

For completeness, we state here a two-dimensional Beta integral with an
additional logarithmic factor that is needed in \wref{app:CMT-solution} (see
also \citealt[Appendix~B]{Martinez2001}):
\begin{align}
		\BetaFun_{\ln}(\alpha_1,\alpha_2) 
	&\wwrel\ce 
		- \int_0^1 x^{\alpha_1-1} (1-x)^{\alpha_2-1} \ln x \, dx
\notag\\	&\wwrel{\like[r]\ce=}
		\BetaFun(\alpha_1, \alpha_2) 
		(\harm{\alpha_1+\alpha_2-1} -	\harm{\alpha_1-1}) \;.
\label{eq:beta-log}
\end{align}

For integral parameters $\vect\alpha$, the proof is elementary: 
By partial integration, we can find a recurrence equation for $\BetaFun_{\ln}$:
\begin{align*}
		\BetaFun_{\ln}(\alpha_1,\alpha_2)
	&\wwrel=
		\frac1{\alpha_1} \BetaFun(\alpha_1,\alpha_2) \bin+ 
		\frac{\alpha_2-1}{\alpha_1} \BetaFun_{\ln}(\alpha_1+1,\alpha_2-1) \;.
\end{align*}
Iterating this recurrence until we reach the base case $\BetaFun_{\ln}(a,0) =
\frac1{a^2}$ and using \wref{eq:beta-function-via-gamma} to expand the
Beta function, we obtain \wref{eq:beta-log}.

\needspace{5cm}
\subsection{Multinomial Distribution}
\label{sec:multinomial-distribution}

Let $n,d \in \N$ and $k_1,\ldots,k_d\in\N$. \emph{Multinomial coefficients} are
the multidimensional extension of binomials:
\begin{align*}
		\binom{n}{k_1,k_2,\ldots,k_d}
	&\wwrel\ce 
	\begin{cases} \displaystyle
			\frac{n!}{k_1! k_2! \cdots k_d!} ,
		& \displaystyle \text{if } 
			n=\sum_{i=1}^d k_i \;; \\[1ex]
		0 , & \text{otherwise} .
	\end{cases}
\end{align*}
Combinatorially, $\binom{n}{k_1,\ldots,k_d}$ is the number of ways to
partition a set of $n$ objects into $d$ subsets of respective sizes
$k_1,\ldots,k_d$ and thus 
they appear naturally in the \weakemph{multinomial theorem}:
\begin{align}
\label{eq:multinomial-theorem}
		(x_1 + \cdots + x_d)^n
	&\wwrel= \mkern-10mu
		\sum_{\substack{i_1,\ldots,i_d \in \N \\ i_1+\cdots+i_d = n}} \mkern-5mu 
			\binom{n}{i_1,\ldots,i_d} \; x_1^{i_1} \cdots x_d^{i_d} 
		\qquad\qquad\text{for } n\in\N \;.
\end{align}

Let $p_1,\ldots,p_d \in [0,1]$ such that $\sum_{i=1}^d p_i = 1$.
A random variable $\vect X\in\N^d$ is said to have \emph{multinomial
distribution} with parameters $n$ and $\vect p = (p_1,\ldots,p_d)$\,---\,written shortly
as~$\vect X \eqdist \multinomial(n,\vect p)$\,---\,if for any 
$\vect i = (i_1,\ldots,i_d) \in \N^d$
holds
\begin{align*}
		\Prob(\vect X = \vect i)
	&\wwrel=
		\binom{n}{i_1,\ldots,i_d} \; p_1^{i_1} \cdots p_d^{i_d} \;.
\end{align*}

We need some expected values involving multinomial variables.
They can be expressed as special cases of the following mixed factorial moments.

\begin{lemma}
	\label{lem:multinomial-mixed-factorial-moments}
	Let $p_1,\ldots,p_d \in [0,1]$ such that $\sum_{i=1}^d p_i =1$
	and consider a $\multinomial(n,\vect p)$ distributed variable
	$\vect X = (X_1,\ldots,X_d) \in \N^d$.
	Let further $m_1,\ldots,m_d \in \N$ be non-negative integers and abbreviate 
	their sum as $M \ce \sum_{i=1}^d m_i$. 
	Then we have 
	\begin{align*} 
			\E\bigl[
				(X_1)^{\underline{m_1}} \cdots (X_d)^{\underline{m_d}} 
			\bigr] 
		&\wwrel=
			n^{\underline M} \, p_1^{m_1} \cdots p_d^{m_d} \;.
	\end{align*}
\end{lemma}

\begin{proof}
We compute
\begin{align}
		\E\bigl[ (X_1)^{\underline{m_1}} \cdots (X_d)^{\underline{m_d}} \bigr]
	&\wwrel=
		\sum_{\vect x \in \N^d} 
			x_1^{\,\underline{m_1}}\cdots x_d^{\,\underline{m_d}}
			\binom{n}{x_1,\ldots,x_d} \;
			p_1^{x_1} \cdots p_d^{x_d} 
	\nonumber\\ &\wwrel=
		n^{\underline M} \, p_1^{m_1} \cdots p_d^{m_d} \times{} 
	\nonumber\\*&\wwrel\ppe
		\sum_{\substack{\vect x \in \N^d : \\ \forall i : x_i \ge m_i}}
			\mkern -10mu \binom{n-M}{x_1-m_1,\ldots,x_d-m_d} \;
		p_1^{x_1-m_1} \cdots p_d^{x_d-m_d} 
	\nonumber\\ &\wwrel{\eqwithref{eq:multinomial-theorem}}
		n^{\underline M} \, p_1^{m_1} \cdots p_d^{m_d}
		\;  
		\bigl(\.\. \smash{ \underbrace{p_1 + \cdots + p_d} _ {=1} } \.\. \bigr)^{n-M}
	\nonumber\\ &\wwrel=
		n^{\underline M} \, p_1^{m_1} \cdots p_d^{m_d}
		\;.
\end{align}
\end{proof}

\clearpage

\section[Proof of Lemma 6.1]{Proof of \wref{lem:expectations}}
\label{app:proof-of-lem-expectations}

In this appendix, we give the computations needed to prove \wref{lem:expectations}.
They were also given in Appendix~D of \citep{NebelWild2014}, 
but we reproduce them here for the reader's convenience.

We recall that
$		\vect D
	\eqdist	
		\dirichlet(\vect t + 1)
$ and 
$
		\vect I
	\eqdist
		\multinomial(n-k,\vect D)
$
and start with the simple ingredients: $\E[I_j]$ for $j=1,2,3$.
\begin{align}
		\E[I_j]	
	&\wwrel=	
		\E_{\vect D} \bigl[ \E[ I_j \given \vect D = \vect d] \bigr]
\nonumber\\	&\wwrel{\eqwithref[r]{lem:multinomial-mixed-factorial-moments}}	
							\E_{\vect D} \bigl[ D_j (n-k) \bigr] 
\nonumber\\	&\wwrel{\eqwithref[r]{lem:dirichlet-mixed-moments}}	
							(n-k) \frac{t_j+1}{k+1}
		\;.
\label{eq:expectation-Ij} 
\end{align}
The term $\E\bigl[\bernoulli\bigl(\frac{I_3}{n-k}\bigr)\bigr]$ is then easily
computed using~\wref{eq:expectation-Ij}:
\begin{align}
\label{eq:expectation-delta}
		\E\bigl[\bernoulli\bigl(\tfrac{I_3}{n-k}\bigr)\bigr]
	&\wwrel=
		\frac{\E[{I_3}]}{n-k}
	\wwrel=
		\frac{t_3+1}{k+1}
	\wwrel{\wwrel=} \Theta(1) \;.
\end{align}
This leaves us with the hypergeometric variables; using the well-known formula
$\E[\hypergeometric(k,r,n)] = k\frac rn$, we find
\begin{align}
		\E\bigl[ \hypergeometric(I_1+I_2, I_3, n-k) \bigr]
	&\wwrel=
		\E_{\vect I} \Bigl[ 
			\E\bigl[ \hypergeometric(i_1+i_2, i_3, n-k) \given \vect I = \vect i \bigr] 
		\Bigr]
	\nonumber\\	&\wwrel= 
		\E\left[ \frac{(I_1+I_2) I_3}{n-k} \right]
	\nonumber\\	&\wwrel=
		\E_{\vect D} \left[ 
			  \frac{ \E[ I_1 I_3 \given \vect D] 
			+ \E[ I_2 I_3 \given \vect D ] }{n-k}
		\right]
	\nonumber\\ &\wwrel{\eqwithref[r]{lem:multinomial-mixed-factorial-moments}}
		\frac{ (n-k)^{\underline 2} \E[D_1 D_3] + 
				(n-k)^{\underline 2} \E[D_2 D_3] } 
			{n-k}
	\nonumber\\ &\wwrel{\eqwithref[r]{lem:dirichlet-mixed-moments}}
		\frac{\bigl((t_1+1)+(t_2+1)\bigr)(t_3+1)}{(k+1)^{\overline2}} (n-k-1)
		\;.
\label{eq:expectation-sm-at-G}  
\end{align}
The second hypergeometric summand is obtained similarly. 
\hfill\proofSymbol

\clearpage
\section{Solution to the Recurrence}
\label{app:CMT-solution}

This appendix is an update of Appendix~E in \citep{NebelWild2014}, 
we include it here for the reader's convenience.

An elementary proof can be given for
\wref{thm:leading-term-expectation} using
\citeauthor{Roura2001}'s \emph{Continuous Master Theorem} (CMT)
\citep{Roura2001}.
The CMT applies to a wide class of full-history recurrences whose coefficients
can be well-approximated asymptotically by a so-called \emph{shape function}
$w:[0,1] \to \R$. 
The shape function describes the coefficients only depending on the \emph{ratio}
$j/n$ of the subproblem size $j$ and the current size $n$ (not depending on $n$
or $j$ itself) and it smoothly continues their behavior to any real number
$z\in[0,1]$.
This continuous point of view also allows to compute precise asymptotics
for complex discrete recurrences via fairly simple integrals.

\begin{theorem}[{{\citealt[Theorem~18]{Martinez2001}}}]
\label{thm:CMT}
	Let $F_n$ be recursively defined~by
	\begin{align}
	\label{eq:CMT-recurrence}
		F_n \wwrel= \begin{cases}
			b_n,	&\text{for~} 0 \le n < N; \\
			\displaystyle{ \vphantom{\bigg|}
				t_n \bin+ \smash{\sum_{j=0}^{n-1} w_{n,j} \, F_j}, 
			} 	&\text{for~} n \ge N\,
		\end{cases}
	\end{align}
	where the toll function satisfies $t_n \sim K n^\alpha \log^\beta(n)$ as
	$n\to\infty$ for constants $K\ne0$, $\alpha\ge0$ and $\beta > -1$.
	Assume there exists a function $w:[0,1]\to \R$, 
	such that 
	\begin{align}
	\label{eq:CMT-shape-function-condition}
		\sum_{j=0}^{n-1} \,\biggl|
			w_{n,j} \bin- \! \int_{j/n}^{(j+1)/n} \mkern-15mu w(z) \: dz
		\biggr|
		\wwrel= \Oh(n^{-d}),
		\qquad\qquad(n\to\infty),
	\end{align}
	for a constant $d>0$.
	With \smash{$\displaystyle H \ce 1 - \int_0^1 \!z^\alpha w(z) \, dz$}, we
	have the following cases:
	\begin{enumerate}[itemsep=0ex]
		\item If $H > 0$, then $\displaystyle F_n \sim \frac{t_n}{H}$.
		\item \label{case:CMT-H0} 
		If $H = 0$, then 
		$\displaystyle
		F_n \sim \frac{t_n \ln n}{\tilde H}$ with 
		$\displaystyle \tilde H = -(\beta+1)\int_0^1 \!z^\alpha \ln(z) \, w(z) \, dz$.
		\item \label{case:CMT-theta-nc}
		If $H < 0$, then $F_n \sim \Theta(n^c)$ for the unique
		$c\in\R$ with $\displaystyle\int_0^1 \!z^c w(z) \, dz = 1$.
	\end{enumerate}
\qed\end{theorem}

\smallskip\noindent
The analysis of single-pivot Quicksort with pivot sampling is the application
par excellence for the CMT \citep{Martinez2001}. 
We will generalize this work of \citeauthor{Martinez2001} to the dual-pivot
case.

Note that the recurrence for $F_n$ depends \emph{linearly} on $t_n$,
so whenever $t_n = t'_n + t_n^{\pprime}$, we can apply the CMT to both the summands of
the toll function separately and sum up the results.
In particular, if we have an asymptotic expansion for $t_n$, we get an asymptotic
expansion for $F_n$; the latter might however get truncated in precision when
we end up in \wildref[case]{case:CMT-theta-nc}{\ref*} of \wref{thm:CMT}.

\medskip
\noindent
Our \wildtpageref[Equation]{eq:ECn-recurrence}{\eqref} has the form of
\wref{eq:CMT-recurrence} with
$$
		w_{n,j}
	\wwrel=
		\sum_{r=1}^3 \Prob\bigl( J_r = j \bigr) \;.
$$
Recall that $\vect J = \vect I + \vect t$ and that
$\vect I \eqdist \multinomial(n-k,\vect D)$ conditional on $\vect D$,
which in turn is a random variable with distribution 
$\vect D \eqdist \dirichlet(\vect t+1)$.

The probabilities $\Prob(J_r = j) = \Prob(I_r = j-t_r)$
can be computed using that the marginal distribution of
$I_r$ is binomial $\binomial(N,D_r)$, where we abbreviate by $N \ce n-k$ the
number of ordinary elements. 
It is convenient to consider $\vect{\tilde D} \ce (D_r,1-D_r)$,
which is distributed like
$\vect{\tilde D} \eqdist \dirichlet(t_r+1,k-t_r)$.
For $i\in[0..N]$ holds
\begin{align}
		\Prob(I_r=i)
	&\wwrel=
		\E_{\vect D}\bigl[ \E_{\vect J}[ \indicator{I_r=i} \given \vect D ] \bigr]
\notag\\	&\wwrel=
		\E_{\vect D}\bigl[
			\tbinom Ni \tilde D_1^i \tilde D_2^{N-i} 
		\bigr]
\notag\\	&\wwrel{\eqwithref[r]{lem:dirichlet-mixed-moments}}
		\binom Ni 
		\frac{(t_r+1)^{\overline i}(k-t_r)^{\overline{N-i}}}
			{(k+1)^{\overline N}}\;.
\label{eq:prob-Il-equals-i}
\end{align}

\subsection{Finding a Shape Function}
In general, a good guess for the shape function is $w(z) = \lim_{n\to\infty}
n\,w_{n,zn}$ \citep{Roura2001} and, indeed, this will work out for our
weights.
We start by considering the behavior for large $n$ of the terms 
$\Prob(I_r = zn + \rho)$ for $r=1,2,3$, where $\rho$ does not depend on $n$.
Assuming $zn+\rho \in \{0,\ldots,n\}$, we compute
\begin{align}
		\Prob(I_r=zn+\rho)
	&\wwrel= 
		\binom N{zn+\rho} 
		\frac{(t_r+1)^{\overline{zn+\rho}}(k-t_r)^{\overline{(1-z)n-\rho}}}
			{(k+1)^{\overline N}}
\notag\\	&\wwrel=
		\frac{N!}{(zn+\rho)!((1-z)n-\rho)!} 
		\frac{\displaystyle \frac{(zn + \rho + t_r)!} {t_r!}  \,
				\frac{\bigl((1-z)n - \rho + k - t_r - 1\bigr)!} {(k-t_r-1)!}}
		{\displaystyle \frac{(k+N)!}{k!}}
\notag\\	&\wwrel=
		\underbrace{\frac{k!}{t_r!(k-t_r-1)!}} _ {{}= 1/\BetaFun(t_r+1,k-t_r)}
		\frac{(zn+\rho+t_r)^{\underline{t_r}}  \,
				\bigl((1-z)n - \rho + k - t_r - 1\bigr)^{\underline{k-t_r-1}}}
		{n^{\underline k}}
		\,,
\intertext{%
	and since this is a \emph{rational} function in $n$\,,
}
	&\wwrel=
		\frac1{\BetaFun(t_r+1,k-t_r)}  
		\frac{ (zn)^{t_r} ((1-z)n)^{k-t_r-1}}{n^k}
		\cdot \Bigl(
			1 \bin+ \Oh(n^{-1})
		\Bigr)
\notag\\ 	&\wwrel=
		\underbrace{
			\frac1{\BetaFun(t_r+1,k-t_r)} z^{t_r} (1-z)^{k-t_r-1}
		} _ {\equalscolon w_r(z)}
		\cdot \Bigl(
			n^{-1} \bin+ \Oh(n^{-2})
		\Bigr)\,,
		\qquad (n\to\infty).
\label{eq:CMT-n-prob-Il-zn-limit}
\end{align}
Thus 
$n\.\Prob(J_r = zn)  \rel=  n\.\Prob(I_r = zn-t_r)  \rel\sim  w_r(z)$, and
our candidate for the shape function is
\begin{align*}
		w(z)
	&\wwrel=
		\sum_{r=1}^3 w_r(z) 
	\wwrel= 
		\sum_{r=1}^3 \frac{z^{t_r}(1-z)^{k-t_r-1}}{\BetaFun(t_r+1,k-t_r)} \;.
\end{align*}
Note that $w_r(z)$ is the density function of a $\dirichlet(t_r+1,k-t_r)$ 
distributed random variable.

It remains to verify condition \wref{eq:CMT-shape-function-condition}.
We first note using \wref{eq:CMT-n-prob-Il-zn-limit} that
\begin{align}
		n \. w_{n,zn} 
	&\wwrel=
		w(z) \bin+ \Oh(n^{-1})
	\;.
\label{eq:CMT-w-n-zn-asymptotic}
\end{align}
Furthermore as $w(z)$ is a \emph{polynomial} in $z$, its derivative exists and
is finite in the compact interval $[0,1]$, so its absolute value is bounded by
a constant~$C_w$.
Thus $w:[0,1]\to\R$ is \textsl{Lipschitz-continuous} with Lipschitz constant
$C_w$:
\begin{align}
\label{eq:CMT-wz-Lipschitz}
	\forall z,z'\in[0,1] 
	&\wwrel:
		\bigl|w(z) - w(z')\bigr| 
		\wrel\le 
		C_w |z-z'|
		\;.
\end{align}
For the integral from \wref{eq:CMT-shape-function-condition}, we then have
\begin{align*}
		\sum_{j=0}^{n-1} \,\biggl|
			w_{n,j} \bin- \! \int_{j/n}^{(j+1)/n} \mkern-15mu w(z) \: dz
		\biggr|
	&\wwrel=
		\sum_{j=0}^{n-1} \,\biggl|
			\int_{j/n}^{(j+1)/n} \mkern-15mu n \. w_{n,j} - w(z) \: dz
		\biggr|
\\	&\wwrel\le
		\sum_{j=0}^{n-1} \frac1n \cdot 
			\max_{z\in \bigl[\frac jn, \frac{j+1}n \bigr]}
				\Bigl| n \. w_{n,j} - w(z) \Bigr|
\\	&\wwrel{\eqwithref{eq:CMT-w-n-zn-asymptotic}}
		\sum_{j=0}^{n-1} \frac1n \cdot
			\Biggl[
				\max_{\;z\in \bigl[\frac jn, \frac{j+1}n \bigr]}
				\Bigl| w(j/n) - w(z)\Bigr|  \wbin+ \Oh(n^{-1})
			\Biggr]
\\	&\wwrel\le
		\Oh(n^{-1}) \wbin+
			\max_{\substack{z,z'\in[0,1]:\\ |z-z'|\le 1/n}}
			\bigl| w(z) - w(z')\bigr|
\\	&\wwrel{\relwithref{eq:CMT-wz-Lipschitz}{\le}}
		\Oh(n^{-1}) \wbin+
		C_w \frac1n
\\	&\wwrel= \Oh (n^{-1}) \,,
\end{align*}
which shows that our $w(z)$ is indeed a shape function of our recurrence
(with $d=1$).

\subsection{Applying the CMT}
With the shape function $w(z)$ we can apply \wref{thm:CMT} with $\alpha=1$,
$\beta=0$ and $K=a$.
It turns out that \wildref[case]{case:CMT-H0}{\ref*} of the CMT applies:
\begin{align*}
		H
	&\wwrel=
		1 \bin- \int_0^1 z \, w(z) \, dz
\\	&\wwrel=
		1 \bin-\sum_{r=1}^3 \int_0^1 z \, w_r(z) \, dz
\\	&\wwrel=
		1 \bin-\sum_{r=1}^3 \frac1{\BetaFun(t_r+1,k-t_r)} \BetaFun(t_r+2,k-t_r)
\\	&\wwrel{\eqwithref{eq:beta-function-via-gamma}} 
		1 \bin - \sum_{r=1}^3\frac{t_r+1}{k+1}
	\wwrel= 0\;.
\end{align*}  
For this case, the leading-term coefficient of the solution is $t_n \ln(n) /
\tilde H = n \ln(n) / \tilde H$ with
\begin{align*}
		\tilde H
	&\wwrel=
		- \int_0^1 z \ln(z) \, w(z) \, dz
\\	&\wwrel=
		\sum_{r=1}^3 \frac1{\BetaFun(t_r+1,k-t_r)} \BetaFun_{\ln}(t_r+2,k-t_r)
\\	&\wwrel{\eqwithref{eq:beta-log}}
		\sum_{r=1}^3 
			\frac{\BetaFun(t_r+2,k-t_r)(\harm{k+1} - \harm{t_r+1})}
				{\BetaFun(t_r+1,k-t_r)} 
\\	&\wwrel=
		\sum_{r=1}^3 \frac{t_r+1}{k+1}(\harm{k+1} - \harm{t_r+1})\;.
\end{align*}
So indeed, we find $\tilde H = \discreteEntropy$ as claimed in
\wref{thm:leading-term-expectation}, 
concluding the proof for the leading term.

\smallskip
As argued above, the error bound is obtained by a second application of
the CMT, where the toll function now is $K\cdot n^{1-\epsilon}$ for a $K$
that gives an upper bound of the toll function:
$\E[T_n] - an \le K n^{1-\epsilon}$ for large $n$.
We thus apply \wref{thm:CMT} with $\alpha=1-\epsilon$, $\beta=0$ and $K$.
We note that $f_c : \R_{\ge1} \to \R$ with 
$f_c(z) = \Gamma(z)/\Gamma(z+c)$ is a strictly \emph{decreasing} function in $z$ 
for any positive fixed $c$
and hence the beta function $\BetaFun$ is strictly decreasing in 
all its arguments by \wref{eq:beta-function-via-gamma}.
With that, we compute
\begin{align*}
		H
	&\wwrel=
		1 \bin- \int_0^1 z^{1-\epsilon} \, w(z) \, dz
\\	&\wwrel=
		1 \bin-\sum_{r=1}^3 \frac{\BetaFun(t_r+2-\epsilon,k-t_r)}{\BetaFun(t_r+1,k-t_r)}
\\	&\wwrel{<} 
		1 \bin-\sum_{r=1}^3 \frac{\BetaFun(t_r+2,k-t_r)}{\BetaFun(t_r+1,k-t_r)}
	\wwrel=0 \;.
\end{align*}
Consequently, \wildref[case]{case:CMT-theta-nc}{\ref*} applies.
We already know from above that the exponent that makes $H$ become $0$ is
$\alpha=1$, so the $F_n = \Theta(n)$.
This means that a toll function that is bounded by $\Oh(n^{1-\epsilon})$ 
for $\epsilon>0$ contributes only to the linear term in overall costs of
Quicksort, and this is independent of the pivot sampling parameter~$\vect t$.
Putting both results together yields 
\wref{thm:leading-term-expectation}.

\medskip
Note that the above arguments actually \emph{derive}\,---\,not only prove
correctness of\,---\,the precise leading-term asymptotics of a quite involved 
recurrence equation. 
Compared with \citeauthor{hennequin1991analyse}'s original proof via generating
functions, it needed less mathematical theory.

\end{document}